\documentclass[a4paper,11pt]{article}
\pdfoutput=1 

\usepackage{jheppub} 

\usepackage[T1]{fontenc} 

\usepackage{enumerate}
\usepackage{graphicx}
\usepackage{color}
\usepackage{physics}
\usepackage{url,hyperref}
\usepackage{mathrsfs}
\usepackage{amssymb}
\usepackage{amsmath}
\usepackage{amsthm}
\usepackage[mathscr]{euscript}
\usepackage{subcaption} 
\usepackage{bm}
\DeclareMathAlphabet\mathbfcal{OMS}{cmsy}{b}{n}
\usepackage[dvipsnames]{xcolor}

\newcommand{\del}{\partial}

\renewcommand{\tilde}{\widetilde}
\renewcommand{\bar}{\overline}

\usepackage{tikz}
\usetikzlibrary{positioning, calc}
\usetikzlibrary{calc}
\usetikzlibrary{arrows}
\usepackage{tikz-3dplot}
\usetikzlibrary{fadings}
\usetikzlibrary{decorations.pathreplacing,decorations.markings,decorations.pathmorphing}
\tikzset{snake it/.style={decorate, decoration=snake}}
\usetikzlibrary{patterns,patterns.meta}
\usetikzlibrary{decorations}

\tikzset{
    >=stealth',
    punkt/.style={
           rectangle,
           rounded corners,
           draw=black, very thick,
           text width=6.5em,
           minimum height=2em,
           text centered},
    pil/.style={
           ->,
           thick,
           shorten <=2pt,
           shorten >=2pt,},
  on each segment/.style={
    decorate,
    decoration={
      show path construction,
      moveto code={},
      lineto code={
        \path [#1]
        (\tikzinputsegmentfirst) -- (\tikzinputsegmentlast);
      },
      curveto code={
        \path [#1] (\tikzinputsegmentfirst)
        .. controls
        (\tikzinputsegmentsupporta) and (\tikzinputsegmentsupportb)
        ..
        (\tikzinputsegmentlast);
      },
      closepath code={
        \path [#1]
        (\tikzinputsegmentfirst) -- (\tikzinputsegmentlast);
      },
    },
  },
  mid arrow/.style={postaction={decorate,decoration={
        markings,
        mark=at position .5 with {\arrow[#1]{stealth'}}
      }}}
}

\newtheorem{theorem}{Theorem}

\newtheorem{lemma}[theorem]{Lemma}

\theoremstyle{definition}

\newtheorem{definition}[theorem]{Definition}
\newtheorem{protocol}[theorem]{Protocol}
\newtheorem{remark}[theorem]{Remark}
\newtheorem{assumption}[theorem]{Assumption}
\newtheorem*{argument}{Argument}

\title{The connected wedge theorem and its consequences}

\author[a,b]{Alex May,}
\author[a, c]{Jonathan Sorce,}
\author[b]{and Beni Yoshida}

\affiliation[a]{Stanford Institute for Theoretical Physics, Stanford University, 382 Via Pueblo Mall, Stanford, CA 94305-4060, U.S.A.}
\affiliation[b]{Perimeter Institute for Theoretical Physics, Waterloo, Ontario N2L 2Y5, Canada}
\affiliation[c]{Center for Theoretical Physics, Massachusetts Institute of Technology, 182 Memorial Dr, Cambridge, MA 02142, U.S.A.}

\emailAdd{may@phas.ubc.ca}
\emailAdd{jsorce@mit.edu}
\emailAdd{byoshida@perimeterinstitute.ca}

\abstract{In the AdS/CFT correspondence, bulk causal structure has consequences for boundary entanglement. In quantum information science, causal structures can be replaced by distributed entanglement for the purposes of information processing. In this work, we deepen the understanding of both of these statements, and their relationship, with a number of new results. 
Centrally, we present and prove a new theorem, the $n$-to-$n$ connected wedge theorem, which considers $n$ input and $n$ output locations at the boundary of an asymptotically AdS$_{2+1}$ spacetime described by AdS/CFT. 
When a sufficiently strong set of causal connections exists among these points in the bulk, a set of $n$ associated regions in the boundary will have extensive-in-N mutual information across any bipartition of the regions. 
The proof holds in three bulk dimensions for classical spacetimes satisfying the null curvature condition and for semiclassical spacetimes satisfying standard conjectures.
The $n$-to-$n$ connected wedge theorem gives a precise example of how causal connections in a bulk state can emerge from large-N entanglement features of its boundary dual.
It also has consequences for quantum information theory: it reveals one pattern of entanglement which is sufficient for information processing in a particular class of causal networks.
We argue this pattern is also necessary, and give an AdS/CFT inspired protocol for information processing in this setting.

Our theorem generalizes the $2$-to-$2$ connected wedge theorem proven in arXiv:1912.05649. We also correct some errors in the proof presented there, in particular a false claim that existing proof techniques work above three bulk dimensions.
}

\begin{document} 
\maketitle
\flushbottom

\section{Introduction}

Among the most basic puzzles in quantum gravity is how a gravitating spacetime emerges from underlying quantum mechanical degrees of freedom. 
A sharp context in which to study this problem is the AdS/CFT correspondence, where a $(d+1)$-dimensional asymptotically anti-de Sitter spacetime is recorded into the state of a $d$-dimensional conformal field theory living at the asymptotic boundary. 
Here, a puzzle arises from the observation that the causal structure of the bulk theory is state-dependent, while the causal structure of the boundary is fixed.
In particular, for certain choices of boundary points and bulk spacetimes, the set of causal relationships among those points may be richer in the bulk than on the boundary.
For instance, in vacuum AdS$_3,$ one can choose points $\{c_1, c_2, r_1, r_2\}$ so that there exists a bulk point $p$ where timelike signals from $c_1,c_2$ can meet, then travel to points $r_1$ and $r_2$, while no such point $p$ exists on the boundary.
As we add more points --- say, $\{c_1, \dots, c_n, r_1, \dots, r_n\}$ --- we can construct complicated causal relationships in the bulk that are not present on the boundary.
See figure \ref{fig:causalstructures} for examples. 

This raises the following question: \emph{What data in a boundary state records the pattern of causal relationships that exist in its bulk dual?}

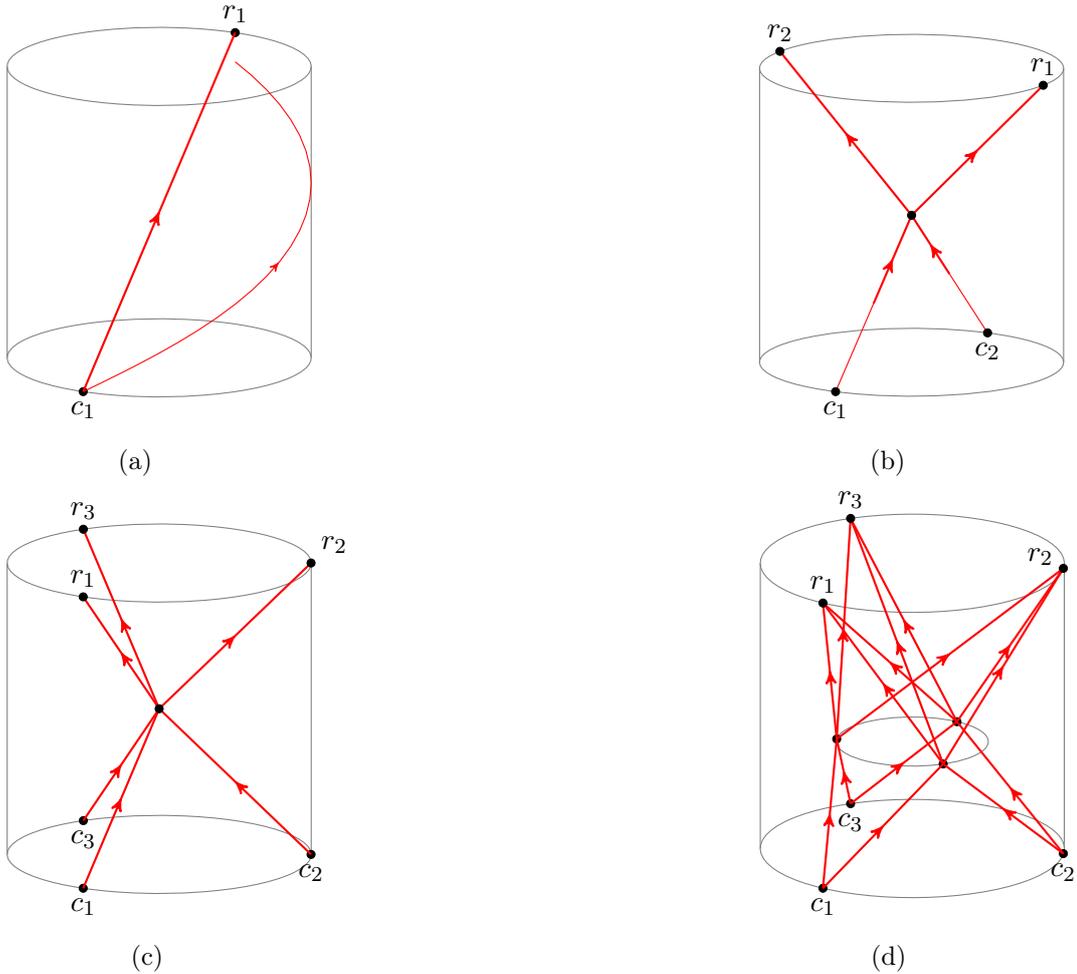
\begin{figure}
    \centering
    \subfloat[\label{fig:1to1cylinder}]{
    \tdplotsetmaincoords{15}{0}
    \begin{tikzpicture}[scale=1.0,tdplot_main_coords]
    \tdplotsetrotatedcoords{0}{30}{0}
    \draw[gray] (-2,1,0) -- (-2,5,0);
    \draw[gray] (2,1,0) -- (2,5,0);
    
    \begin{scope}[tdplot_rotated_coords]
    
    \begin{scope}[canvas is xz plane at y=1]
    \draw[gray] (0,0) circle[radius=2] ;
    \end{scope}
    
    \begin{scope}[canvas is xz plane at y=5]
    \draw[gray] (0,0) circle[radius=2] ;
    \end{scope}
    
    \draw plot [mark=*, mark size=1.5] coordinates{(0,1,-2)};
    \node[below] at (0,1,-2) {$c_1$};
    
    \draw plot [mark=*, mark size=1.5] coordinates{(0,5,2)};
    \node[above] at (0,5,2) {$r_1$};
    
    \draw[thick,red,mid arrow] (0,1,-2) -- (0,5,2);
    
    \draw[domain=0:180,variable=\x,red,mid arrow] plot ({-2*cos(-\x-90)}, {1+4*\x/200}, {2*sin(-\x-90)});
    
    \end{scope}
    \end{tikzpicture}
    }
    \hfill
    \centering
    \subfloat[\label{fig:2to2cylinder}]{
    \tdplotsetmaincoords{13}{0}
        \begin{tikzpicture}[scale=1.0,tdplot_main_coords]
    \tdplotsetrotatedcoords{0}{30}{0}
    \draw[gray] (-2,1,0) -- (-2,5,0);
    \draw[gray] (2,1,0) -- (2,5,0);
    
    \begin{scope}[tdplot_rotated_coords]
    
    \begin{scope}[canvas is xz plane at y=1]
    \draw[gray] (0,0) circle[radius=2] ;
    \end{scope}
    
    \begin{scope}[canvas is xz plane at y=5]
    \draw[gray] (0,0) circle[radius=2] ;
    \end{scope}
    
    \draw[red] (0,1,-2) -- (0,2,-1);
    \draw[red] (0,1,2) -- (0,2,1);
    
    \draw[thick, red,mid arrow] (0,3,0) -- (2,5,0);
    \draw[thick, red,mid arrow] (0,3,0) -- (-2,5,0);
    
    \draw[thick,red,mid arrow] (0,2,-1) -- (0,3,0);
    \draw[thick,red,mid arrow] (0,2,1) -- (0,3,0);
    
    \draw plot [mark=*, mark size=1.5] coordinates{(2,5,0)};
    \node[above] at (2,5,0) {$r_1$};
    \draw plot [mark=*, mark size=1.5] coordinates{(-2,5,0)};
    \node[above] at (-2,5,0) {$r_2$};
    
    \draw plot [mark=*, mark size=1.5] coordinates{(0,1,-2)};
    \node[below] at (0,1,-2) {$c_1$};
    \draw plot [mark=*, mark size=1.5] coordinates{(0,1,2)};
    \node[below] at (0,1,2) {$c_2$};
    \draw plot [mark=*, mark size=1.5] coordinates{(0,3,0)};
    
    \end{scope}
    \end{tikzpicture}
    }
    \hfill
    \subfloat[\label{fig:3to3cylinder}]{
    \tdplotsetmaincoords{15}{0}
    \begin{tikzpicture}[scale=1.0,tdplot_main_coords]
    \tdplotsetrotatedcoords{0}{30}{0}
    \draw[gray] (-2,1,0) -- (-2,5,0);
    \draw[gray] (2,1,0) -- (2,5,0);
    
    \begin{scope}[tdplot_rotated_coords]
    
    \begin{scope}[canvas is xz plane at y=1]
    \draw[gray] (0,0) circle[radius=2] ;
    \end{scope}
    
    \begin{scope}[canvas is xz plane at y=5]
    \draw[gray] (0,0) circle[radius=2] ;
    \end{scope}
    
    \draw plot [mark=*, mark size=1.5] coordinates{(0,1,-2)};
    \node[below] at (0,1,-2) {$c_1$};
    
    \draw plot [mark=*, mark size=1.5] coordinates{(1.73,1,1)};
    \node[below] at (1.73,1,1) {$c_2$};
    
    \draw plot [mark=*, mark size=1.5] coordinates{(-1.73,1,1)};
    \node[below] at (-1.73,1,1) {$c_3$};
    
    \draw[thick,red,mid arrow] (-1.73,1,1) -- (0,3,0);
    \draw[thick,red,mid arrow] (0,3,0) -- (-1.73,5,1);
    
    \draw[thick,red,mid arrow] (1.73,1,1) -- (0,3,0) ;
    \draw[thick,red,mid arrow] (0,3,0) -- (1.73,5,1);
    
    \draw[thick,red,mid arrow] (0,1,-2) -- (0,3,0);
    \draw[thick,red,mid arrow] (0,3,0) -- (0,5,-2);
    
    \draw plot [mark=*, mark size=1.5] coordinates{(0,3,0)};
    
    \draw plot [mark=*, mark size=1.5] coordinates{(0,5,-2)};
    \node[above] at (0,5,-2) {$r_1$};
    
    \draw plot [mark=*, mark size=1.5] coordinates{(-1.73,5,1)};
    \node[above] at (-1.73,5,1) {$r_3$};
    
    \draw plot [mark=*, mark size=1.5] coordinates{(1.73,5,1)};
    \node[above right] at (1.73,5,1) {$r_2$};

    \end{scope}
    \end{tikzpicture}
    }
    \hfill
    \centering
    \subfloat[\label{fig:3to3cylinderalternate}]{
    \tdplotsetmaincoords{19}{0}
    \begin{tikzpicture}[scale=1.0,tdplot_main_coords]
    \tdplotsetrotatedcoords{0}{36}{0}
    \draw[gray] (-2,1,0) -- (-2,5,0);
    \draw[gray] (2,1,0) -- (2,5,0);
    
    \begin{scope}[tdplot_rotated_coords]
    
    \begin{scope}[canvas is xz plane at y=1]
    \draw[gray] (0,0) circle[radius=2] ;
    \end{scope}
    
    \begin{scope}[canvas is xz plane at y=5]
    \draw[gray] (0,0) circle[radius=2] ;
    \end{scope}
    
    \begin{scope}[canvas is xz plane at y=2.5]
    \draw[gray] (0,0) circle[radius=1] ;
    \end{scope}
    
    \draw plot [mark=*, mark size=1.5] coordinates{(0,1,-2)};
    \node[below] at (0,1,-2) {$c_1$};
    \draw plot [mark=*, mark size=1.5] coordinates{(1.73,1,1)};
    \node[below] at (1.73,1,1) {$c_2$};
    \draw plot [mark=*, mark size=1.5] coordinates{(-1.73,1,1)};
    \node[below] at (-1.73,1,1) {$c_3$};

    \draw plot [mark=*, mark size=1.5] coordinates{(0.866,2.5,-0.5)};
    \draw plot [mark=*, mark size=1.5] coordinates{(-0.866,2.5,-0.5)};
    \draw plot [mark=*, mark size=1.5] coordinates{(0,2.5,1)};
    
    \draw[red,thick,mid arrow] (0,1,-2) -- (-0.866,2.5,-0.5);
    \draw[red,thick,mid arrow] (0,1,-2) -- (0.866,2.5,-0.5);
    
    \draw[red,thick,mid arrow] (1.73,1,1) -- (0,2.5,1);
    \draw[red,thick,mid arrow] (1.73,1,1) -- (0.866,2.5,-0.5);
    
    \draw[red,thick,mid arrow] (-1.73,1,1) -- (0,2.5,1);
    \draw[red,thick,mid arrow] (-1.73,1,1) -- (-0.866,2.5,-0.5);
    
    \draw[red,thick,mid arrow] (0,2.5,1) -- (-1.73,5,1);
    \draw[red,thick,mid arrow] (0,2.5,1) -- (1.73,5,1);
    \draw[red,thick,mid arrow] (0,2.5,1) -- (0,5,-2);
    
    \draw[red,thick,mid arrow] (-0.866,2.5,-0.5)  -- (-1.73,5,1);
    \draw[red,thick,mid arrow] (-0.866,2.5,-0.5)  -- (1.73,5,1);
    \draw[red,thick,mid arrow] (-0.866,2.5,-0.5)  -- (0,5,-2);
    
    \draw[red,thick,mid arrow] (0.866,2.5,-0.5)  -- (-1.73,5,1);
    \draw[red,thick,mid arrow] (0.866,2.5,-0.5)  -- (1.73,5,1);
    \draw[red,thick,mid arrow] (0.866,2.5,-0.5)  -- (0,5,-2);
    
    \draw plot [mark=*, mark size=1.5] coordinates{(0,5,-2)};
    \node[above] at (0,5,-2) {$r_1$};
    \draw plot [mark=*, mark size=1.5] coordinates{(-1.73,5,1)};
    \node[above] at (-1.73,5,1) {$r_3$};
    \draw plot [mark=*, mark size=1.5] coordinates{(1.73,5,1)};
    \node[above left] at (1.73,4.9,1) {$r_2$};
    
    \end{scope}
    \end{tikzpicture}
    }
    \caption{(a) With one input point and one output point the bulk causal structure is fixed by the boundary causal structure: the Gao-Wald theorem \cite{gao2000theorems} dictates that bulk light rays are never faster than boundary light rays. (b) With two input and two output points the bulk causal structure may be stronger than the boundary one. Shown here is a case where the four points scatter causally in the bulk but not in the boundary. An essential part of reconciling this discrepancy is entanglement in the boundary state, as discussed in \cite{may2019quantum, may2020holographic, may2021holographic}. (c, d) For three or more input and output points there are many possible bulk causal structures, which again may feature scattering regions not appearing in the boundary. In this article we study how these richer causal structures are supported through boundary entanglement.}
    \label{fig:causalstructures}
\end{figure}

In \cite{may2019quantum, may2020holographic, may2021holographic}, this gravitational problem was shown to be related to the problem of understanding \textit{causal networks} in quantum information processing.
A causal network is a directed graph, where each vertex represents a spacetime location where quantum operations can take place, and each directed edge represents a one-way causal relationship that allows quantum systems to be transferred between the corresponding vertices.
Example causal networks are sketched in figure \ref{fig:causalnetworks}.
A causal network has ``input'' vertices, where all edges are outgoing, and ``output'' vertices, where all edges are ingoing.
They appear in varied contexts in quantum information theory, including in relativistic quantum information \cite{kent2012quantum,hayden2016summoning,hayden2019localizing,dolev2021distributing,dolev2019constraining} and position based cryptography \cite{chandran2009position,kent2011quantum,kent2006tagging, malaney2010location}, while similar objects appear in network coding \cite{leung2010quantum} and quantum foundations \cite{chaves2021causal}.
In these contexts, one asks: \emph{For a given network, what operations can be performed on a quantum system distributed across its input vertices?}
This question becomes richer if one allows entangled states to be distributed among vertices in the network, especially among the input vertices.
One then asks: \textit{Given an operation, what pattern of entanglement is needed to perform it using this network? How much of that entanglement do we need?}

\begin{figure}
    \centering
    \subfloat[\label{fig:networks1to1}]{
    \begin{tikzpicture}[scale=0.65]
    
    \begin{scope}[shift={(0,0)}]
    \draw[mid arrow, red] (-2,-2) -- (2,2);
    
    \draw plot [mark=*, mark size=2] coordinates{(-2,-2)};
    \draw plot [mark=*, mark size=2] coordinates{(2,2)};
    \end{scope}
    
    \begin{scope}[shift={(6,0)}]
    \draw[mid arrow, red] (-2,-2) -- (2,2);
    
    \draw plot [mark=*, mark size=2] coordinates{(-2,-2)};
    \draw plot [mark=*, mark size=2] coordinates{(2,2)};
    \end{scope}
    
    \end{tikzpicture}
    }
    \hfill
    \centering
    \subfloat[\label{fig:networks2to2}]{
    \begin{tikzpicture}[scale=0.65]
    
    \begin{scope}[shift={(0,0)}]
    \draw[mid arrow, red] (0,0) -- (2,2);
    \draw[mid arrow, red] (0,0) -- (-2,2);
    \draw[mid arrow, red] (-2,-2) -- (0,0);
    \draw[mid arrow, red] (2,-2) -- (0,0);
    
    \draw plot [mark=*, mark size=2] coordinates{(0,0)};
    \draw plot [mark=*, mark size=2] coordinates{(-2,-2)};
    \draw plot [mark=*, mark size=2] coordinates{(2,2)};
    \draw plot [mark=*, mark size=2] coordinates{(-2,2)};
    \draw plot [mark=*, mark size=2] coordinates{(2,-2)};
    \end{scope}
    
    \begin{scope}[shift={(6,0)}]
    \draw[mid arrow, red] (-2,-2) -- (0,0);
    \draw[red, mid arrow] (0,0) -- (2,2);
    \draw[mid arrow, red] (2,-2) -- (0.1,-0.1);
    \draw[mid arrow, red] (-0.1,0.1) -- (-2,2);
    \draw[mid arrow, red] (-2,-2) -- (-2,2);
    \draw[mid arrow, red] (2,-2) -- (2,2);
    
    \draw plot [mark=*, mark size=2] coordinates{(-2,-2)};
    \draw plot [mark=*, mark size=2] coordinates{(2,2)};
    \draw plot [mark=*, mark size=2] coordinates{(-2,2)};
    \draw plot [mark=*, mark size=2] coordinates{(2,-2)};
    \end{scope}
    
    \end{tikzpicture}
    }
    \hfill
    \centering
    \subfloat[\label{fig:networks3to3}]{
    \begin{tikzpicture}[scale=0.65]

    \begin{scope}[shift={(0,0)}]
    \draw[mid arrow, red] (0,0) -- (2,2);
    \draw[mid arrow, red] (0,0) -- (-2,2);
    \draw[mid arrow, red] (-2,-2) -- (0,0);
    \draw[mid arrow, red] (2,-2) -- (0,0);
    \draw[mid arrow, red] (0,-2) -- (0,0);
    \draw[mid arrow, red] (0,0) -- (0,2);
    
    \draw plot [mark=*, mark size=2] coordinates{(0,0)};
    \draw plot [mark=*, mark size=2] coordinates{(-2,-2)};
    \draw plot [mark=*, mark size=2] coordinates{(2,2)};
    \draw plot [mark=*, mark size=2] coordinates{(-2,2)};
    \draw plot [mark=*, mark size=2] coordinates{(2,-2)};
    \draw plot [mark=*, mark size=2] coordinates{(0,-2)};
    \draw plot [mark=*, mark size=2] coordinates{(0,2)};
    \end{scope}
    
    \begin{scope}[shift={(6,0)}]
    \draw[mid arrow, red] (0,0) -- (2,2);
    \draw[mid arrow, red] (0,0) -- (-2,2);
    \draw[mid arrow, red] (-2,-2) -- (0,0);
    \draw[mid arrow, red] (2,-2) -- (0,0);
    \draw[mid arrow, red] (2,-2) -- (1,0);
    \draw[mid arrow, red] (-2,-2) -- (-1,0);
    \draw[mid arrow, red] (0,-2) -- (1,0);
    \draw[mid arrow, red] (0,-2) -- (-1,0);
    
    \draw[mid arrow, red] (-1,0) -- (-2,2);
    \draw[mid arrow, red] (1,0) -- (2,2);
    \draw[mid arrow, red] (1,0) -- (0,2);
    \draw[mid arrow, red] (-1,0) -- (0,2);
    
    \draw plot [mark=*, mark size=2] coordinates{(0,0)};
    \draw plot [mark=*, mark size=2] coordinates{(-2,-2)};
    \draw plot [mark=*, mark size=2] coordinates{(2,2)};
    \draw plot [mark=*, mark size=2] coordinates{(-2,2)};
    \draw plot [mark=*, mark size=2] coordinates{(2,-2)};
    \draw plot [mark=*, mark size=2] coordinates{(0,-2)};
    \draw plot [mark=*, mark size=2] coordinates{(0,2)};
    
    \draw plot [mark=*, mark size=2] coordinates{(-1,0)};
    \draw plot [mark=*, mark size=2] coordinates{(1,0)};
    \end{scope}

    \end{tikzpicture}
    }
    \hfill
    \subfloat[\label{fig:networks3to3alternate}]{
    
    \begin{tikzpicture}[scale=0.65]
    
    \begin{scope}[shift={(0,0)}]
    \draw[mid arrow, red] (0,0) -- (2,2);
    \draw[mid arrow, red] (0,0) -- (-2,2);
    \draw[mid arrow, red] (-2,-2) -- (0,0);
    \draw[mid arrow, red] (2,-2) -- (0,0);
    \draw[mid arrow, red] (2,-2) -- (1,0);
    \draw[mid arrow, red] (-2,-2) -- (-1,0);
    \draw[mid arrow, red] (0,-2) -- (1,0);
    \draw[mid arrow, red] (0,-2) -- (-1,0);
    
    \draw[mid arrow, red] (-1,0) -- (-2,2);
    \draw[mid arrow, red] (1,0) -- (2,2);
    \draw[mid arrow, red] (1,0) -- (0,2);
    \draw[mid arrow, red] (-1,0) -- (0,2);
    
    \draw[mid arrow, red] (0,0) -- (0,2);
    \draw[mid arrow, red] (-1,0) -- (2,2);
    \draw[mid arrow, red] (1,0) -- (-2,2);
    
    \draw plot [mark=*, mark size=2] coordinates{(0,0)};
    \draw plot [mark=*, mark size=2] coordinates{(-2,-2)};
    \draw plot [mark=*, mark size=2] coordinates{(2,2)};
    \draw plot [mark=*, mark size=2] coordinates{(-2,2)};
    \draw plot [mark=*, mark size=2] coordinates{(2,-2)};
    \draw plot [mark=*, mark size=2] coordinates{(0,-2)};
    \draw plot [mark=*, mark size=2] coordinates{(0,2)};
    
    \draw plot [mark=*, mark size=2] coordinates{(-1,0)};
    \draw plot [mark=*, mark size=2] coordinates{(1,0)};
    \end{scope}
    
    \begin{scope}[shift={(6,0)}]
    \draw[mid arrow, red] (0,0) -- (2,2);
    \draw[mid arrow, red] (0,0) -- (-2,2);
    \draw[mid arrow, red] (-2,-2) -- (0,0);
    \draw[mid arrow, red] (2,-2) -- (0,0);
    \draw[mid arrow, red] (2,-2) -- (1,0);
    \draw[mid arrow, red] (-2,-2) -- (-1,0);
    \draw[mid arrow, red] (0,-2) -- (1,0);
    \draw[mid arrow, red] (0,-2) -- (-1,0);
    
    \draw[mid arrow, red] (-1,0) -- (-2,2);
    \draw[mid arrow, red] (1,0) -- (2,2);
    \draw[mid arrow, red] (1,0) -- (0,2);
    \draw[mid arrow, red] (-1,0) -- (0,2);
    
    \draw plot [mark=*, mark size=2] coordinates{(0,0)};
    \draw plot [mark=*, mark size=2] coordinates{(-2,-2)};
    \draw plot [mark=*, mark size=2] coordinates{(2,2)};
    \draw plot [mark=*, mark size=2] coordinates{(-2,2)};
    \draw plot [mark=*, mark size=2] coordinates{(2,-2)};
    \draw plot [mark=*, mark size=2] coordinates{(0,-2)};
    \draw plot [mark=*, mark size=2] coordinates{(0,2)};
    
    \draw plot [mark=*, mark size=2] coordinates{(-1,0)};
    \draw plot [mark=*, mark size=2] coordinates{(1,0)};
    \end{scope}
    
    \end{tikzpicture}
    }
    \caption{Causal networks. 
    Vertices with no incoming edges are input points, where quantum systems are input to the network, and vertices with no outgoing edges are output points, where quantum systems are output from the network.
    Here, we show examples of pairs of causal networks which are related by the AdS/CFT correspondence. 
    On the left we show the causal network describing the bulk, on the right we show a (simplified) causal network describing the boundary, with labels corresponding to the configurations shown in figure \ref{fig:causalstructures}.
    In (a) the bulk and boundary networks agree, while in the remaining cases they differ.
    The AdS/CFT correspondence implies that these pairs of networks allow for the same information processing tasks to be completed, when the boundary network has access to whatever entanglement is present in the holographic CFT. 
   }
    \label{fig:causalnetworks}
\end{figure}
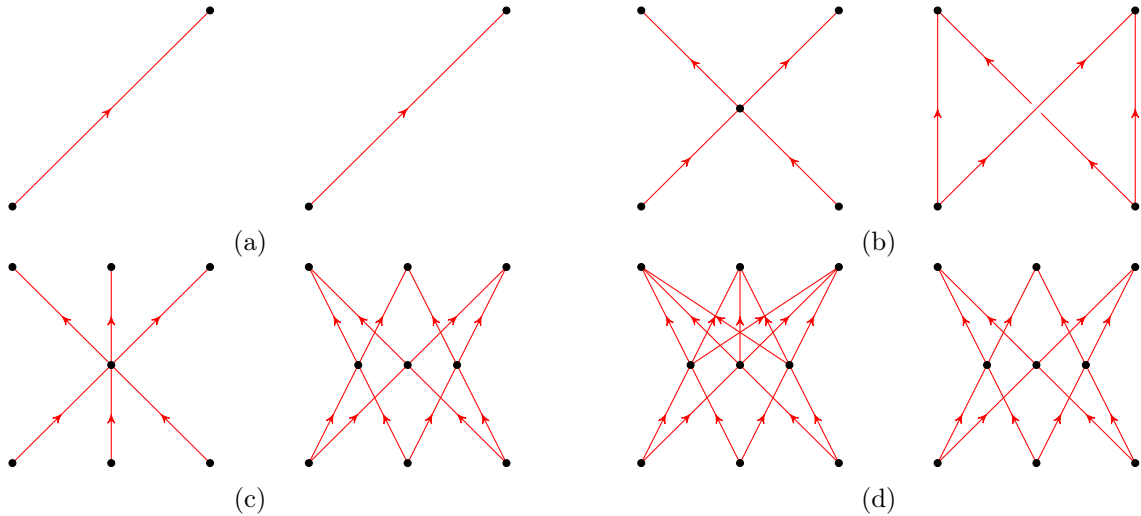

The gravitational problem and the causal network problem are related by the observation that in gravity, once a bulk spacetime and a set of boundary points have been specified, the causal relationships among those points in the bulk and on the boundary give rise to different causal networks.
The dual nature of AdS/CFT, however, means that operations that can be performed in the bulk network can also be performed in the boundary network.
On some networks, one can show that certain operations can be performed only if the network has access to a resource state with particular properties; this way of reasoning about information processing in causal networks falls under the study of ``quantum tasks.''
The quantum tasks perspective allows for tools from the study of information processing on causal networks to be applied toward understanding bulk emergence in AdS/CFT, and vice versa.

A first step toward establishing this relationship was made in \cite{may2019quantum, may2020holographic, may2021holographic}, where ideas from the study of causal networks were used to argue that a certain $2$-to-$2$ causal structure in a bulk spacetime must be supported by a certain large-N pattern of mutual information in its boundary state.
Using the (quantum) extremal surface formula for boundary entanglement \cite{ryu2006holographic, ryu2006aspects, headrick2007holographic, hubeny2007covariant, lewkowycz2013generalized, faulkner2013quantum, dong2016deriving, engelhardt2015quantum, dong2018entropy}, this was converted into a statement relating the causal structure of a spacetime to the geometry of its extremal surfaces.
That statement was proven directly using the tools of general relativity, up to some technical caveats addressed in the present work.

In this paper, we establish a much richer link between the theory of spacetime emergence and the theory of causal networks.
Our main result is a generalization of the theorem proved in \cite{may2019quantum, may2020holographic, may2021holographic}, which reveals boundary-entanglement consequences for a larger class of bulk causal structures than the one discussed in those papers.
We call this generalization the \emph{$n$-to-$n$ connected wedge theorem}.
In addition to revealing a way that bulk causal structure is supported by boundary correlations, the $n$-to-$n$ connected wedge theorem establishes a correspondence between certain causal networks: a network corresponding to bulk causal structure, and a network corresponding to boundary causal structure plus boundary-state entanglement.
AdS/CFT implies that any quantum operation that can be done in the bulk network must be implementable in the boundary network when it has access to the entanglement resources present in the boundary state.
We show examples of networks related in this way in figure \ref{fig:causalnetworks}.\footnote{In fact we show simplified boundary networks here, this is discussed in section \ref{sec:tasks-background} in more detail.}
Our second main result is to understand the connected wedge theorem from the perspective of quantum information.
We show that the minimal bulk causal network required by the theorem is in fact sufficient to perform arbitrary quantum operations \textit{without} any initial correlation.
We also show that the corresponding boundary networks require entanglement to perform the same quantum operations, and that the pattern of entanglement they require is exactly the one implied by the $n$-to-$n$ connected wedge theorem.

As a complementary result, for the $n=3$ case of the connected wedge theorem, we give an example of a protocol for completing a quantum operation using the boundary causal network.
The boundary protocol works by using the entanglement resource to encode input systems into an entanglement-assisted quantum error correcting code (EAQECC), which then evolves under transversal operations at intermediary layers of the boundary network, before being decoded at the output vertices.
The pattern of entanglement required by the code is exactly the pattern implied by the $3$-to-$3$ connected wedge theorem, and since semiclassical AdS/CFT is a quantum error correcting code \cite{almheiri2015bulk}, boundary time evolution will generally look like local evolution of an encoded state.
We discuss what the EAQECC model might teach us about the relationship between bulk and boundary dynamics in quantum gravity.

We can summarize the context for our work with figure \ref{fig:triangle}, which depicts relationships among our three objects of study: bulk causal structure, bulk extremal surfaces, and boundary correlation. 
The (quantum) extremal surface formula relates extremal surfaces and boundary correlations, as measured by the von Neumann entropy. 
Geometric tools in general relativity establish relationships between causal structure and extremal surfaces.
Reasoning about information processing in causal networks relates bulk causal structure and boundary correlation via the language of quantum tasks.
Importantly, the connections among these three ideas allows geometric tools and information-theoretic tools to inform each other.

\begin{figure}
    \centering
    \begin{tikzpicture}[scale=0.8]
    
    \node[align=center] at (0,0) {\textbf{bulk} \\ \textbf{causal features}};

    \node[align=center] at (-4,-4) {\textbf{bulk} \\\textbf{extremal surfaces}};
    \draw[->,blue] (-0.75,-0.75) -- (-3.25,-3.25);
    
    \node[align=center] at (4,-4) {\textbf{boundary} \\\textbf{correlation}};
    \draw[->,blue] (0.75,-0.75) -- (3.25,-3.25);
    
    \draw[blue,<->] (-1.75,-4) -- (2.3,-4);
    
    \node[below] at (0,-4.25) {\small{(quantum)} };
    \node[below] at (0,-4.8) {\small{extremal surface formula} };
    
    \node[right,align=center] at (2.25,-2) {\small{quantum} \\ \small{tasks}};
    \node[left,align=center] at (-2.4,-2) {\small{general} \\ \small{relativity}};
    
    \end{tikzpicture}
    \caption{The expanded geometry-correlation connection present in AdS/CFT.}
    \label{fig:triangle}
\end{figure}
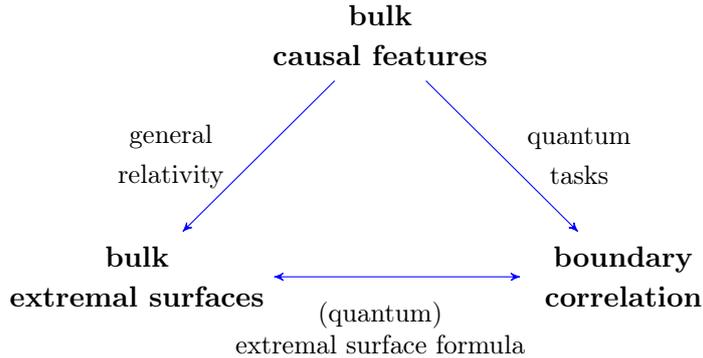

The plan of the paper is as follows.

In section \ref{sec:overview}, we review the motivation behind the $2$-to-$2$ connected wedge theorem discussed in \cite{may2019quantum, may2020holographic, may2021holographic}, and state our new, $n$-to-$n$ generalization.

In section \ref{sec:timelike-cylinders}, we prove several lemmas about the causal structure of timelike cylinders that will be needed both in the geometric proof of the $n$-to-$n$ connected wedge theorem in section \ref{sec:GR-proof}, and in the theorem's information-theoretic interpretation in section \ref{sec:QI-argument}.

In section \ref{sec:GR-proof}, we prove the $n$-to-$n$ connected wedge theorem stated in section \ref{sec:overview}.

In section \ref{sec:tasks-background} we review the framework for thinking about holography in terms of quantum tasks, which was introduced in \cite{may2019quantum} and developed in \cite{may2020holographic, may2021bulk, may2021holographic, may2021quantum}.
The framework presented in section \ref{sec:tasks-background} has some differences from those papers; insights from the present paper have led us to emphasize the use of causal networks and to refine the correspondence between bulk and boundary quantum tasks.

In section \ref{sec:QI-argument}, we give an information-theoretic argument for the $n$-to-$n$ connected wedge theorem. We first show that the bulk causal structure required by the theorem is sufficient to perform any bulk quantum task, and further that it is sufficient to efficiently perform a generalization of the $\textbf{B}_{84}$ task discussed in \cite{may2019quantum,may2020holographic,may2021quantum}. 
We then argue that for the boundary to imitate this task under the causal restrictions of the boundary cylinder, they must share entanglement in the pattern implied by the $n$-to-$n$ connected wedge theorem.

In section \ref{sec:boundaryprotocols}, we discuss how the pattern of entanglement guaranteed by the connected wedge theorem can be used to perform quantum tasks in the boundary network. 
We propose a class of protocols based of entanglement assisted quantum error correcting codes (EAQECC). 
We discuss how these protocols may capture interesting features of boundary dynamics.

In the discussion, we make a geometric observation that the relevant bulk causal objects, the $2$-to-all scattering regions, sit inside of the entanglement wedge of the boundary input regions. 
This result is extends an observation made in \cite{may2020holographic}. 
We also make a number of remarks and comment on open questions. 

In appendix \ref{sec:GR}, we collect relevant terminology from general relativity and some lemma's dealing with causal structure we need for the proof of the connected wedge theorem. 
In appendix \ref{sec:bulkprotocolappendix}, we give a detailed protocol which completes arbitrary quantum tasks in the bulk causal structure. 
In appendix \ref{app:lower-bound-proof} we prove a lemma lower bounding the mutual information necessary for the $\mathbf{B}_{84}$ task. 
Various versions of this lemma (or pieces of its proof) have appeared in \cite{may2019quantum,may2020holographic,may2021holographic,may2021quantum}, but we give a complete and simplified version. 
In appendix \ref{sec:CodeConstructions} we describe the error-correcting codes used in our boundary protocols.  

Finally we note that some of the techniques discussed in section \ref{sec:boundaryprotocols} have overlap with the discussion in the forthcoming paper \cite{dolev2022holography}, which will discuss the causal network appearing in figure \ref{fig:networks2to2} and its relationship to holography. 

\vspace{0.2cm}
\noindent \textbf{Summary of notation:}
\vspace{0.2cm}

\noindent Spacetime notation: (See also appendix \ref{app:GR-travel-guide}.)
\begin{itemize}
    \item We use italic capital letters $\mathcal{A}, \mathcal{B}, \mathcal{C}, \dots$ for boundary spacetime regions.
    \item The entanglement wedge of a boundary region $\mathcal{A}$ is denoted by $E_{\mathcal{A}}$. 
    \item The quantum extremal surface associated to region $\mathcal{A}$ is denoted $\gamma_{\mathcal{A}}$.
    \item We use plain capital letters $A, B, C, \dots$ to refer to bulk spacetime regions.
    \item We use script capital letters such as $\mathscr{F}$ to refer to bulk null surfaces.
    \item We use $J^\pm(A)$ to denote the causal future or past of region $A$ taken in the bulk geometry, and $\hat{J}^\pm(\mathcal{A})$ to denote the causal future or past of $\mathcal{A}$ within the boundary spacetime.
    Bulk sets like $J^+(A)$ are taken within the conformally completed spacetime that includes the asymptotic boundary, so they can include boundary points.
    \item For spacetime sets $S_1, \dots, S_n,$ we write
    \begin{equation}
        J(S_1, \dots, S_j \rightarrow S_{j+1} \dots S_n) \equiv J^+(S_1) \cap \dots \cap J^+(S_j) \cap J^-(S_{j+1}) \cap \dots \cap J^-(S_n).
    \end{equation}
    An analogous notation for boundary causality holds with $J$ replaced by $\hat{J}.$
\end{itemize}
Quantum notation:
\begin{itemize}
    \item We use capital letters to denote quantum systems $A,B,C,...$
    \item We use boldface, script capital letters for quantum channels, $\mathbfcal{N}(\cdot)$, $\mathbfcal{T}(\cdot)$,...
    \item We use boldface capital letters to denote unitaries or isometries, $\mathbf{U}, \mathbf{V},...$
\end{itemize}

\section{Overview of the connected wedge theorem}
\label{sec:overview}

We first elaborate the notation for causal sets summarized at the end of the introduction.
Further details, along with the definitions of many technical terms used below, are collected in appendix \ref{sec:glossary}.

Given a spacetime $\mathcal{M}$ with a conformal boundary, and given a subset $S$ of the spacetime-with-boundary, the set $J^+(S)$ denotes the set of all points in $\mathcal{M}$ or its boundary that can be reached from points in $S$ by future-directed causal curves.
The set $\hat{J}^+(S)$ denotes the set of all points in the conformal boundary of $\mathcal{M}$ that can be reached by future-directed causal curves that lie entirely in the boundary.
Substituting the symbol ``+'' for ``-'' is equivalent to replacing ``future'' with ``past.''
For a collection of sets $S_1, \dots, S_n$, we use the notation
\begin{equation}
    J(S_1, \dots, S_j \rightarrow S_{j+1}, \dots, S_n)
        \equiv J^+(S_1) \cap \dots \cap J^+(S_j) \cap J^-(S_{j+1}) \cap \dots \cap J^-(S_{n}).
\end{equation}
Analogously, we define
\begin{equation}
    \hat{J}(S_1, \dots, S_j \rightarrow S_{j+1}, \dots, S_n)
        \equiv \hat{J}^+(S_1) \cap \dots \cap \hat{J}^+(S_j) \cap \hat{J}^-(S_{j+1}) \cap \dots \cap \hat{J}^-(S_n).
\end{equation}

With this notation in hand, we can now review the $n=2$ case of the connected wedge theorem, considered previously in \cite{may2019quantum, may2020holographic, may2021holographic}.

\subsection{The 2-to-2 case}

The setting of the $2$-to-$2$ connected wedge theorem considers four points $\{c_1, c_2, r_1, r_2\}$ in the boundary of an asymptotically globally AdS$_{2+1}$ spacetime.
See figure \ref{fig:causalstructures}.
We choose these points such that within the boundary spacetime, $r_1$ and $r_2$ are each in the future of $c_1$ and in the future of $c_2$.
In our notation, this means
\begin{align}
    \mathcal{V}_1
        \equiv \hat{J}(c_1\rightarrow r_1,r_2)
        \neq \varnothing,\\
    \mathcal{V}_2
        \equiv \hat{J}(c_2\rightarrow r_1,r_2)
        \neq \varnothing.
\end{align}
We will refer to $\mathcal{V}_1$ and $\mathcal{V}_2$ as \emph{input regions}.
Further, we ensure that there are no boundary points in the future of both $c_1$ and $c_2$ and in the past of both $r_1$ and $r_2.$
In our notation, this means
\begin{equation}
    \hat{J}(c_1, c_2 \rightarrow r_1, r_2)
        = \varnothing.
\end{equation}
We say this configuration of points ``cannot scatter on the boundary.''
As discussed in the introduction, for some configurations of boundary points with this property it is possible to choose bulk spacetimes for which the configuration \textit{can} scatter, meaning that  $J(c_1,c_2\rightarrow r_1,r_2)$ is nonempty and in fact has nonempty interior. 
This is the simplest example of a ``causal discrepancy'' between the bulk and boundary spacetimes, various examples of which were shown in figure \ref{fig:causalnetworks}. 

If the chosen bulk spacetime corresponds to a physical state in the AdS/CFT correspondence,\footnote{There are plenty of examples of physical states in which such a causal discrepancy between bulk and boundary can be realized; for example, the AdS$_{2+1}$ vacuum where $c_1$ and $c_2$ are antipodal points on the boundary slice $t=0$ and $r_1$ and $r_2$ are antipodal points on the boundary slice $t=\pi+\epsilon$, rotated by $\pi/2$ relative to $c_1$ and $c_2$.} then a quantum information argument given in \cite{may2019quantum} suggests consequences for the entanglement structure of its boundary dual.
Consider introducing quantum systems $A_1$, $A_2$ at $c_1$ and $c_2$, and arranging for these systems to travel into the region $J(c_1,c_2\rightarrow r_1,r_2)$, undergo some interaction unitary $\mathbf{U}_{A_1A_2}$, then emerge again at points $r_1$ and $r_2$.
The boundary description of this process must have the same inputs and outputs, so that the same unitary is implemented, but it is unclear how this unitary can be realized in the boundary given that there is no scattering region in the boundary spacetime where $A_1$ and $A_2$ can interact. 
To resolve this, \cite{may2019quantum} argued that the boundary interaction is reproduced using the causal network shown at right in figure \ref{fig:networks2to2}, and pointed out that for this network to reproduce the bulk interaction, entanglement between the input regions $\mathcal{V}_1$ and $\mathcal{V}_2$ is necessary.\footnote{In fact, having a large amount of entanglement between these regions is also sufficient to reproduce any bulk interaction, though many questions remain as to how much entanglement is needed for a given interaction. See e.g. \cite{may2022complexity} for a discussion.} 
Based on this, \cite{may2019quantum} conjectured that whenever $J(c_1,c_2\rightarrow r_1,r_2)$ is nonempty in the bulk, so that an interaction can happen there, the mutual information between regions $\mathcal{V}_1$ and $\mathcal{V}_2$ is large in the sense of scaling as $\Omega(\ell_{\text{AdS}}/G_N)$ in the semiclassical limit $\ell_{\text{AdS}}/G_N \rightarrow \infty$.

In holographic states where boundary entanglement entropies can be computed using the Hubeny-Rangamani-Ryu-Takayanagi (HRRT) formula \cite{hubeny2007covariant}, the statement that $\mathcal{V}_1$ and $\mathcal{V}_2$ share extensive mutual information can be cast as a statement about areas of bulk extremal surfaces.
In particular, the statement is that the entanglement wedge of $\mathcal{V}_1 \cup \mathcal{V}_2$ is connected.
The statement that a $2$-to-$2$ causal discrepancy between a bulk state and its boundary dual implies a connected entanglement wedge is known as the \textit{$2$-to-$2$ connected wedge theorem}.

\begin{theorem}[\textbf{$2$-to-$2$ connected wedge theorem}] \label{thm:first-connected-wedge-theorem}
    Let $c_1, c_2, r_1, r_2$ be points on a global AdS$_{2+1}$ boundary such that
    the sets $\mathcal{V}_1 \equiv \hat{J}(c_1 \rightarrow r_1, r_2)$ and $\mathcal{V}_2 \equiv \hat{J}(c_2 \rightarrow r_1, r_2)$  are nonempty, while the set $\hat{J}(c_1, c_2 \rightarrow r_1, r_2)$ is empty.

    Let $\mathcal{M}$ be an asymptotically AdS$_{2+1},$ AdS-hyperbolic spacetime satisfying the null curvature condition, with at least one global AdS$_{2+1}$ boundary on which $c_1, c_2, r_1, r_2$ are specified.

    If $J(c_1, c_2 \rightarrow r_1, r_2)$ has nonempty interior, and if the HRRT surface of $\mathcal{V}_1 \cup \mathcal{V}_2$ can be found using the maximin formula \cite{maximin, maximin2}, then the entanglement wedge of $\mathcal{V}_1 \cup \mathcal{V}_2$ is connected.
\end{theorem}
This theorem was proven\footnote{The proof presented in \cite{may2020holographic} had some oversights that are addressed in the present paper when we prove the theorem's $n$-to-$n$ generalization; for example, spacetimes with singularities were not treated carefully.} in \cite{may2020holographic} using the tools of classical general relativity. It was also shown that the conclusion of the theorem still holds for semiclassical holographic states that do not satisfy the null curvature condition, so long as they satisfy standard conjectures that generalize classical tools to semiclassical spacetimes.

A generalization of this theorem was found in \cite{may2021holographic} by applying the ideas of entanglement wedge reconstruction \cite{CKNR, jafferis2016relative, dong2016reconstruction, cotler2019entanglement} to the quantum information reasoning originally given in \cite{may2019quantum}. 
Along with the input regions $\mathcal{V}_1$ and $\mathcal{V}_2$, here we define the output regions
\begin{align}
    \mathcal{W}_1 &\equiv \hat{J}(c_1,c_2\rightarrow r_1),\\
    \mathcal{W}_2 &\equiv \hat{J}(c_1,c_2\rightarrow r_2),
\end{align}
and denote the entanglement wedge of a boundary region $\mathcal{X}$ by $E_{\mathcal{X}}$.
The key observation is that a bulk process with input systems $A_1,A_2$ beginning anywhere in $\mathcal{V}_1$ and $\mathcal{V}_2$, undergoing some interaction unitary, then ending up in $\mathcal{W}_1$ and $\mathcal{W}_2$ is equally well described in the boundary by a network of the form shown in the right most network in figure \ref{fig:networks2to2}. 
In particular, entanglement between the input regions is again required to reproduce this interaction in the boundary picture.
Thus, whenever the four entanglement wedges share a scattering region in the bulk, the boundary regions $\mathcal{V}_1$ and $\mathcal{V}_2$ should be entangled at order $\Omega(\ell_{AdS}/G_N)$. 
By causal wedge inclusion \cite{maximin, headrick2014causality, engelhardt2015quantum}, the bulk future of $E_{\mathcal{V}_{j}}$ always contains the bulk future of $c_j,$ and the bulk past of $E_{\mathcal{W}_j}$ always contains the bulk past of $r_j.$
So there can be cases when a 2-to-2 bulk causal interaction exists for the entanglement wedges $E_{\mathcal{V}_1}, E_{\mathcal{V}_2}, E_{\mathcal{W}_1}, E_{\mathcal{W}_2},$ but not for the boundary points $c_1, c_2, r_1, r_2.$ 
This implies the following theorem, which is suggested by the reasoning of this paragraph, is strictly stronger than theorem \ref{thm:first-connected-wedge-theorem}.

\begin{theorem}
    Let $c_1, c_2, r_1, r_2$ be points on a global AdS$_{2+1}$ boundary such that
    the sets $\mathcal{V}_1$ and $\mathcal{V}_2$ (defined as above)  are nonempty, while the set $\hat{J}(c_1, c_2 \rightarrow r_1, r_2)$ is empty.

    Let $\mathcal{M}$ be an asymptotically AdS$_{2+1},$ AdS-hyperbolic spacetime satisfying the null curvature condition, with at least one global AdS$_{2+1}$ boundary on which $c_1, c_2, r_1, r_2$ are specified. Let $E_{\mathcal{V}_j}$ and $E_{\mathcal{W}_j}$ be the entanglement wedges of $\mathcal{V}_j$ and $\mathcal{W}_j.$ 

    If $J(E_{\mathcal{V}_1}, E_{\mathcal{V}_2} \rightarrow E_{\mathcal{W}_1}, E_{\mathcal{W}_2})$ has nonempty interior, and if the HRRT surface of $\mathcal{V}_1 \cup \mathcal{V}_2$ can be found using the maximin formula \cite{maximin, maximin2}, then the entanglement wedge of $\mathcal{V}_1 \cup \mathcal{V}_2$ is connected.
\end{theorem}

In \cite{may2021holographic}, this theorem was shown to follow from identical arguments to the ones given in \cite{may2020holographic}.
We view this theorem as more fundamental than the ``points'' version discussed above, so we will simply call it the $2$-to-$2$ connected wedge theorem.

\subsection{The \texorpdfstring{$n$}{TEXT}-to-\texorpdfstring{$n$}{TEXT} theorem}

Now, we can develop the more general $n$-to-$n$ connected wedge theorem. 
We consider a specification of $2n$ points $c_1, \dots, c_n, r_1, \dots, r_n$ on a global AdS$_{2+1}$ boundary, and define input and output regions in analogy with the 2-to-2 case:
\begin{align}
    \mathcal{V}_j & \equiv \hat{J}(c_j \rightarrow r_1, \dots, r_n), \\
    \mathcal{W}_j & \equiv \hat{J}(c_1, \dots, c_n \rightarrow r_j).
\end{align}
We will require that each $\mathcal{V}_j$ is nonempty (which implies that each $\mathcal{W}_j$ is nonempty), which is the statement that each input can causally signal all outputs.
We will further require that all 2-to-$n$ and $n$-to-2 causal regions are empty on the boundary:
\begin{equation}
    \hat{J}(c_j, c_k \rightarrow r_1, \dots, r_n)
        = \hat{J}(c_1, \dots, c_n \rightarrow r_j, r_k) = \varnothing.
\end{equation}
This is a natural generalization of the ``no $2$-to-$2$ boundary interaction'' requirement imposed in the previous subsection.
Note however that we allow $2$-to-$(n-1)$ and $(n-1)$-to-$2$ scattering regions to exist in the boundary. 

For a bulk spacetime containing this configuration on one of its asymptotic boundaries, the $n$-to-$n$ connected wedge theorem gives a causal condition under which the entanglement wedge of $\mathcal{V}_1 \cup \dots \cup \mathcal{V}_n$ is connected.
The most naive generalization of the causal condition that appeared in the previous subsection would be the existence of an $n$-to-$n$ scattering region, i.e., an open interior for the region $J(c_1, \dots, c_n \rightarrow r_1, \dots, r_n).$
This would lend itself naturally to the kind of quantum information argument given in \cite{may2019quantum} and mentioned above for the $2$-to-$2$ theorem; such a bulk causal structure allows $n$ systems $A_1, \dots, A_n$ to come together and interact, and reproducing these interactions in the weaker boundary causal structure plausibly places constraints on the entanglement structure of the boundary state.
However, while we will see that the existence of an $n$-to-$n$ scattering region is \textit{sufficient} to prove connectedness of $E_{\mathcal{V}_1 \cup \dots \cup \mathcal{V}_n}$, the same conclusion can be reached via a much weaker causal condition.

It may help to discuss a simple example before the general case.
For $n=3,$ our gravitational proof shows that connectedness of $E_{\mathcal{V}_1 \cup \mathcal{V}_2 \cup \mathcal{V}_3}$ follows when \textit{at least two} of the bulk regions
\begin{align}
    J(E_{\mathcal{V}_1}, E_{\mathcal{V}_2} \rightarrow E_{\mathcal{W}_1}, E_{\mathcal{W}_2}, E_{\mathcal{W}_3}), \\
    J(E_{\mathcal{V}_1}, E_{\mathcal{V}_3} \rightarrow E_{\mathcal{W}_1}, E_{\mathcal{W}_2}, E_{\mathcal{W}_3}), \\
    J(E_{\mathcal{V}_2}, E_{\mathcal{V}_3} \rightarrow E_{\mathcal{W}_1}, E_{\mathcal{W}_2}, E_{\mathcal{W}_3}),
\end{align}
have nonempty interior.
Having a 3-to-3 scattering region would imply that all three of these sets have nonempty interior; having two of these sets with nonempty interior, however, does not imply the existence of a 3-to-3 scattering region.

For larger $n$, the importance of 2-to-all bulk causal structures persists.
The appropriate generalization to arbitrary $n$ is most easily phrased in terms of the \textit{$2$-to-all causal graph}, $\Gamma_{2 \rightarrow \text{all}}$, defined by
\begin{align}
        \Gamma_{2\rightarrow \text{all}} \equiv \{(j,k): J(E_{\mathcal{V}_j}, E_{\mathcal{V}_k} \rightarrow E_{\mathcal{W}_1}, \dots, E_{\mathcal{W}_n}) \text{ has nonempty interior} \}.
\end{align}
In words, $\Gamma_{2\rightarrow \text{all}}$ is the graph whose vertices are $\{1, \dots, n\},$ and where vertices $j$ and $k$ are connected by an edge if the bulk region $J(E_{\mathcal{V}_j}, E_{\mathcal{V}_k} \rightarrow E_{\mathcal{W}_1}, \dots, E_{\mathcal{W}_n})$ has nonempty interior.
The causal condition checked by the $n$-to-$n$ connected wedge theorem is connectedness of this graph, which we think of as a ``sufficient'' family of bulk causal connections in a sense that we will momentarily make precise.

\begin{theorem} \label{thm:n-to-n-first-statement}
    Let $c_1, \dots, c_n, r_1, \dots, r_n$ be points on a global AdS$_{2+1}$ boundary such that
    the sets $\mathcal{V}_j$ and $\mathcal{W}_j$ are nonempty, while also satisfying
    \begin{equation}
        \hat{J}(c_j, c_k \rightarrow r_1, \dots, r_n) = \hat{J}(c_1, \dots, c_n \rightarrow r_j, r_k) = \varnothing.
    \end{equation}
    Let $\mathcal{M}$ be an asymptotically AdS$_{2+1},$ AdS-hyperbolic spacetime satisfying the null curvature condition, with at least one global AdS$_{2+1}$ boundary on which $c_1, \dots, c_n, r_1, \dots, r_n$ are specified.
    
    If the $2$-to-all causal graph $\Gamma_{2\rightarrow\text{all}}$ is connected, and if the HRRT surface of $\mathcal{V}_1 \cup \dots \cup \mathcal{V}_n$ can be found using the maximin formula \cite{maximin, maximin2}, then the entanglement wedge of $\mathcal{V}_1 \cup \dots \cup \mathcal{V}_n$ is connected.
\end{theorem}

Like in the 2-to-2 case, the proof given in section \ref{sec:GR-proof} will also work for spacetimes that satisfy the quantum maximin formula \cite{quantum-maximin} and the quantum focusing conjecture \cite{QFC}, even if they do not satisfy the null curvature condition.
In the purely classical case, the condition that the entanglement wedge of $\mathcal{V}_1 \cup \dots \cup \mathcal{V}_n$ is connected is equivalent to the statement that the mutual information of any bipartition, $I(\mathcal{V}_{j_1} \cup \dots \cup \mathcal{V}_{j_m} : \mathcal{V}_{j_{m+1}} \cup \dots \cup \mathcal{V}_{j_n})$, is extensive in $\ell_{\text{AdS}}/G_N.$
In the quantum case, the statement that the entanglement wedge is connected can be stronger, since island contributions can cause boundary regions to share extensive mutual information even while their entanglement wedge is disconnected.

Recall that in the $2$-to-$2$ theorem the bulk causal structure suffices to perform an arbitrary unitary, while in the weaker boundary causal structure, entanglement was necessary to reproduce this arbitrary unitary. 
Considering the $n$-to-$n$ theorem and working in analogy, we should guess that the bulk causal structure dictated by the $n$-to-$n$ theorem --- that $\Gamma_{2\rightarrow \text{all}}$ be connected --- suffices to complete any quantum task, and that in the boundary having the pattern of entanglement specified by the theorem is necessary for arbitrary tasks to be reproduced.
Indeed we argue for both of these statements directly in quantum information theory in section \ref{sec:QI-argument}.
That a connected $2$-to-all causal graph suffices for arbitrary tasks we show by constructing an explicit protocol, using tools similar to those in \cite{beigi2011simplified,dolev2019constraining}. 
That the pattern of entanglement specified by the theorem is necessary to complete arbitrary tasks we argue using similar techniques to the ones in \cite{may2019quantum,may2020holographic,may2021holographic}. 

\section{The causal structure of a timelike cylinder}\label{sec:timelike-cylinders}

For the definitions of technical terms that are not given explicitly below, a glossary is provided in appendix \ref{app:glossary}.

A \textit{timelike cylinder} is a two-dimensional spacetime with metric
\begin{equation}
    ds^2 = - dt^2 + d \theta^2,
\end{equation}
with $\theta \in [0, 2 \pi)$ and $t \in (-\infty, \infty).$
Since a global boundary of an asymptotically AdS$_{2+1}$ spacetime has the same causal structure as a timelike cylinder, it will be helpful for us to establish some basic facts about the cylinder.

First, there are only two possible orientations of light rays on the AdS$_{2+1}$ boundary. 
Light rays that, when followed to the future, have increasing angular coordinate will be called ``positively oriented'' and will be labeled with a ``$+$'' superscript. 
Light rays that, when followed to the future, have decreasing angular coordinate will be called ``negatively oriented'' and will be labeled with a ``$-$'' superscript. 

We consider a set of $2n$ points $\{c_1, \dots, c_n, r_1, \dots, r_n\}$ that satisfy the assumptions of theorem \ref{thm:n-to-n-first-statement}.
As in section \ref{sec:overview}, we define the input regions by
\begin{equation}
    \mathcal{V}_j \equiv \hat{J}^+(c_j \rightarrow r_1, \dots, r_n).
\end{equation}
Note that while we are currently working entirely within the cylinder spacetime, we continue to use the notation ``$\hat{J}$'' to avoid confusion when we realize this cylinder as a boundary of another spacetime.
We will also assume that the input points $\{c_1, \dots, c_n\}$ are labeled in order of increasing angular coordinate, so that $\mathcal{V}_j$ and $\mathcal{V}_{j+1}$ are adjacent.
We make no such assumption about the output points $\{r_1, \dots, r_n\}.$

We will show that each $\mathcal{V}_j$ is a causal diamond, i.e., that it is the domain of dependence of a single achronal interval.
The future horizon of each of these causal diamonds will be made up of two null segments, each of which is part of a light ray that intersects one of the points $\{r_1, \dots, r_n\}.$
We will show further that the negatively oriented segment on the future horizon of $\mathcal{V}_j,$ and the positively oriented segment on the future horizon of $\mathcal{V}_{j+1},$ intersect the same output point $r_k.$
In more pedestrian language, this means that the right-endpoint of $\mathcal{V}_j$ and the left-endpoint of $\mathcal{V}_{j+1}$ are ``cut off'' by past-directed light rays leaving the same output point $r_k.$
Proving these statements will require several consecutive lemmas.

\begin{definition}
    The future-directed, positively oriented light ray leaving $c_j$ will be denoted $\gamma_{c_j, +}.$
    The future-directed, negatively oriented light ray leaving $c_j$ will be denoted $\gamma_{c_j, -}.$
    Analogous definitions hold for the past-directed light rays leaving $r_j,$ which will be denoted $\gamma_{r_j, +}$ and $\gamma_{r_j, -}.$
\end{definition}

Because $c_j$ is in $\mathcal{V}_j,$ and $\mathcal{V}_j$ is closed\footnote{The timelike cylinder is globally hyperbolic, so by remark \ref{rem:closed-futures}, futures and pasts of points are closed, and consequently each $\mathcal{V}_j$ is a finite intersection of closed sets.} there must be some final point along $\gamma_{c_j, -}$ that lies in $\mathcal{V}_j,$ after which $\gamma_{c_j, -}$ exits $\mathcal{V}_j.$
We denote this point $x_j.$ The analogous point for $\gamma_{c_j, +}$ is denoted $y_j.$ This is sketched in figure \ref{fig:inputregionstructure}.

\begin{figure}
    \centering
    \begin{tikzpicture}[scale=0.75]
    
    \draw[red] (0,6) -- (6,0) -- (12,6);
    \draw[red] (2,6) -- (8,0);
    \draw[red] (10,6) -- (4,0);
    
    \draw[ForestGreen,very thick] (7,1.1) -- (6,0.1) -- (5,1.1);
    \node[below] at (6,1) {$\sigma_j$};
    
    \draw[thick, black] (0,0) -- (12,0) -- (12,6) -- (0,6) -- (0,0);
    
    \draw plot [mark=*, mark size=2] coordinates{(7,1)};
    \node[right] at (7,1) {$y_j$};
    
    \draw plot [mark=*, mark size=2] coordinates{(7,1)};
    \node[left] at (5,1) {$x_j$};
    
    \draw plot [mark=*, mark size=2] coordinates{(5,1)};
    
    \draw plot [mark=*, mark size=2] coordinates{(6,0)};
    \node[below] at (6,0) {$c_j$};
    
    \draw plot [mark=*, mark size=2] coordinates{(2,0)};
    \node[below] at (2,0) {$c_{j-1}$};
    
    \draw plot [mark=*, mark size=2] coordinates{(10,0)};
    \node[below] at (10,0) {$c_{j+1}$};

    \draw plot [mark=*, mark size=2] coordinates{(2,6)};
    \node[above] at (10,6) {$r_{j_-}$};
    
    \draw plot [mark=*, mark size=2] coordinates{(10,6)};
    \node[above] at (2,6) {$r_{j_+}$};
    
    \node[above left,rotate=45] at (8.7,4.5) {$\gamma_{r_{j_-},+}$};
    \node[above right,rotate=-45] at (3.3,4.5) {$\gamma_{r_{j_+},-}$};
    
    \node[above right,rotate=45] at (8.7,2.5) {$\gamma_{c_{j},+}$};
    \node[above left,rotate=-45] at (3.3,2.5) {$\gamma_{c_{j},-}$};
    
    \end{tikzpicture}
    \caption{Illustration of a number of definitions appearing in the main text. The curve $\sigma_j$ is shown in green.}
    \label{fig:inputregionstructure}
\end{figure}
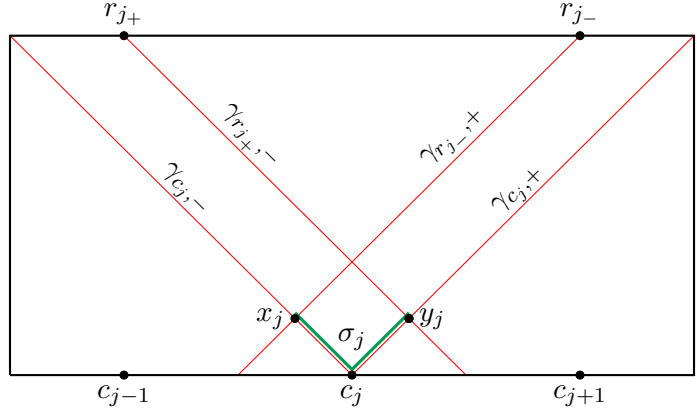

The light ray $\gamma_{c_j, -}$ can only exit $\mathcal{V}_j$ by exiting $\hat{J}^-(r_k)$ for some $k$, which it can only do by crossing the light ray $\gamma_{r_k, +}.$
We will use the label $j_-$ for the index of the light ray $\gamma_{r_{j_-}, +}$ containing $x_j,$ and the label $j_+$ for the index of the light ray $\gamma_{r_{j_+}, -}$ containing $y_j.$
See figure \ref{fig:inputregionstructure}.
For this to be well defined, we must show that the past-directed light rays intersecting $x_j$ and $y_j$ are unique.

\begin{lemma}
    For $k_1 \neq k_2,$ we have
    \begin{equation}
        \gamma_{r_{k_1}, +} \cap \gamma_{r_{k_2}, +}
        = \gamma_{r_{k_1}, -} \cap \gamma_{r_{k_2}, -}
        = \varnothing.
    \end{equation}
\end{lemma}
\begin{proof}
    If two light rays with the same orientation intersect, then they coincide everywhere.
    This would imply that the points $r_{k_1}$ and $r_{k_2}$ are null-separated along their shared light ray.
    But this is forbidden by the assumptions
    \begin{align}
        \hat{J}(c_1, \dots, c_n \rightarrow r_{k_1})
            & \neq \varnothing, \\
        \hat{J}(c_1, \dots, c_n \rightarrow r_{k_2}) 
            & \neq \varnothing, \\
        \hat{J}(c_1, \dots, c_n \rightarrow r_{k_1}, r_{k_2})
            & = \varnothing.
    \end{align}
\end{proof}

\begin{lemma} \label{lem:Vjs-spacelike-separated}
    For $j \neq k,$ we have $j_+ \neq k_+$ and $j_- \neq k_-.$
\end{lemma}
\begin{proof}
    We will show $j_+ \neq k_+$; the argument for $j_- \neq k_-$ is identical.
    Suppose, toward contradiction, that we had $j_+ = k_+ \equiv \ell.$
    Then the points $x_j$ and $x_k$ would be null-separated along their shared light ray.
    One of the points $x_j$ or $x_k$ must be in the future of the other; without loss of generality, say $x_j$ is in the future of $x_k.$
    The point $x_j$ is in the future of $c_j$ and the past of all points $\{r_1, \dots, r_n\}$ by definition.
    But since it is in the future of $x_k,$ it is also in the future of $c_k.$
    Therefore, we have
    \begin{equation}
        x_j \in \hat{J}(c_j, c_k \rightarrow r_1, \dots, r_n).
    \end{equation}
    But this set is empty by assumption.
\end{proof}

By a simple counting argument, this lemma implies that the sets $\{j_-\}_{j=1}^{n}$ and $\{j_+\}_{j=1}^n$ contain each number from $1$ to $n$ exactly once.
This fact can be used to show that each $\mathcal{V}_j$ is a causal diamond.
In particular, we will show that $\mathcal{V}_j$ is the domain of dependence of the achronal curve $\sigma_j$ formed by the segment of $\gamma_{c_j, -}$ from $c_j$ to $x_j$ together with the segment of $\gamma_{c_j, +}$ from $c_j$ to $y_j$; see figure \ref{fig:inputregionstructure}.

\begin{lemma}\label{lemma:Visdomainofsigma}
    Each $\mathcal{V}_j$ is the domain of dependence of $\sigma_j.$
\end{lemma}
\begin{proof}
    To establish the inclusion $\mathcal{V}_j \subseteq D(\sigma_j)$, let $p$ be a point in $\mathcal{V}_j,$ and $\chi$ a past-directed causal curve passing through $p.$
    This curve must eventually leave $\mathcal{V}_j$, meaning it must pass through either $\gamma_{c_j, -}$ or $\gamma_{c_j, +}.$
    But since it is in $\hat{J}^-(r_1) \cap \dots \cap \hat{J}^-(r_n)$ at the point where it leaves $\mathcal{V}_j,$ it must pass through $\gamma_{c_j, -}$ in the past of $x_j,$ or through $\gamma_{c_j, +}$ in the past of $y_j.$
    So it must pass through $\sigma_j.$
    
    To establish the reverse inclusion $D(\sigma_j) \subseteq \mathcal{V}_j,$ fix a point $p \in D(\sigma_j).$
    By definition, the past-directed, positively-oriented light ray leaving $p$ must eventually pass through $\sigma_j.$
    Let $q$ be the point where this happens.
    See figure \ref{fig:diamondproof}.
    The point $q$ is in $\mathcal{V}_j,$ so the only way $p$ could fail to be in $\mathcal{V}_j$ would be for the positively oriented light ray connecting $q$ to $p$ to pass through one of the light rays $\gamma_{r_k, -}$ before reaching $p.$
    We know from the previous lemma that this must be $\gamma_{r_{\ell_+}, -}$ for some $\ell.$
    
    Suppose this happens, and the light ray it passes through is $\gamma_{r_{j_+}, -}.$
    Then the past-directed, negatively oriented light ray leaving $p$ could not intersect $\sigma_j,$ since it would necessarily intersect $\gamma_{c_j, +}$ in the future of the point $y_j.$
    This would contradict $p \in D(\sigma_j).$
    See figure \ref{fig:diamondproof}.
    If the positively oriented light ray connecting $q$ to $p$ passes through some other light ray $\gamma_{r_{\ell_+}, -},$ then the point of intersection of these two light rays would be in $\mathcal{V}_j$, but null separated from $\mathcal{V}_{\ell},$ which contradicts our assumptions by the same argument as that given in the proof of lemma \ref{lem:Vjs-spacelike-separated}.
    
    We conclude $\mathcal{V}_j = D(\sigma_j),$ as desired.
\end{proof}

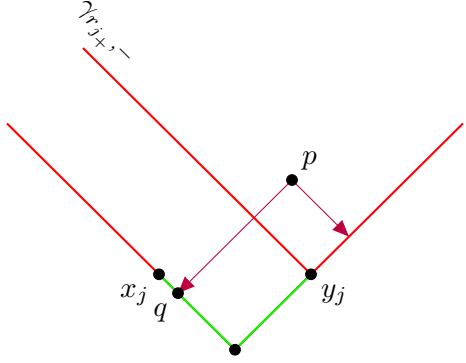
\begin{figure}
    \centering
    \begin{tikzpicture}
    
    \draw[thick,red] (0,0) -- (3,3);
    \draw[thick,red] (0,0) -- (-3,3);
    \draw[thick,red] (-2,4) -- (1,1);
    \node[above left,rotate=-45] at (-1.5,3.5) {$\gamma_{r_{j_+},-}$};
    
    \draw[purple,-triangle 45] (0.75,2.25) -- (-0.75,0.75);
    
    \draw[green, thick] (-1,1) -- (0,0) -- (1,1);
    
    \draw[purple,-triangle 45] (0.75,2.25) -- (1.5,1.5);
    
    \draw plot [mark=*, mark size=2] coordinates{(0,0)};
    
    \draw plot [mark=*, mark size=2] coordinates{(-1,1)};
    \node[below left] at (-1,1) {$x_j$};
    
    \draw plot [mark=*, mark size=2] coordinates{(1,1)};
    \node[below right] at (1,1) {$y_j$};
    
    \node[above right] at (0.75,2.25) {$p$};
    \draw plot [mark=*, mark size=2] coordinates{(0.75,2.25)};
    
    \node[below left] at (-0.75,0.75) {$q$};
    \draw plot [mark=*, mark size=2] coordinates{(-0.75,0.75)};
    
    \end{tikzpicture}
    \caption{Illustration for lemma \ref{lemma:Visdomainofsigma}. The positively oriented past moving light ray from $p$ meets $\sigma_j$ at $q$. If it crosses $\gamma_{r_{j_+},-}$, then the negatively oriented past moving light ray from $p$ meets $\gamma_{c_{j},+}$ in the future of $y_j$, contradicting $p\in D(\sigma_j)$. }
    \label{fig:diamondproof}
\end{figure}

\begin{lemma} \label{lem:adjacent-cutoffs-agree}
    The negatively oriented segment of the future horizon of $\mathcal{V}_j$ originates at the same output point as the positively oriented segment of the future horizon of $\mathcal{V}_{j+1}.$
    
    I.e., $j_+ = (j+1)_-.$
\end{lemma}
\begin{proof}
    Fix $\mathcal{V}_j.$
    Let $\Sigma$ be a Cauchy surface for the cylinder containing all the points $x_1, \dots x_n$ and $y_1, \dots, y_n.$
    Let $p$ be the point where $\gamma_{r_{j_+}, +}$ intersects $\Sigma.$
    The point where $\gamma_{r_{j_+}, -}$ intersects $\Sigma$ is $y_{j}.$
    The set of all points on $\Sigma$ in the past of $r_{j_+}$ is exactly the set of all points between $p$ and $y_j$ in the direction of increasing angular coordinate.
    See figure \ref{fig:pandSigmaproof}.
    
    Since every $\mathcal{V}_k$ must be in the past of $r_{j_+},$ the point $p$ must lie in between $y_j$ and $x_{j+1}$ in the direction of increasing angular coordinate.
    But by the counting argument following lemma \ref{lem:Vjs-spacelike-separated}, we must have $p = x_k$ for some $k.$
    The only way for these two statements to be consistent is to have $p = x_{j+1},$ which proves the lemma.
\end{proof}
This lemma gives us a natural convention for labeling the output points; we will choose to label them so that $j_+ = j,$ i.e., so that the right-endpoint of $\mathcal{V}_j$ and the left-endpoint of $\mathcal{V}_{j+1}$ are contained in light rays intersecting $r_j.$

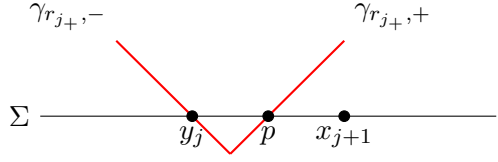
\begin{figure}
    \centering
    \begin{tikzpicture}
    
    \draw[black] (-2,0) -- (4,0);
    \node[left] at (-2,0) {$\Sigma$};
    
    \draw[thick,red] (-1,1) -- (0.5,-0.5);
    \draw[thick,red] (0.5,-0.5) -- (2,1);
    
    \node[above right] at (2,1) {$\gamma_{r_{j_+},+}$};
    \node[above left] at (-1,1) {$\gamma_{r_{j_+},-}$};
    
    \draw plot [mark=*, mark size=2] coordinates{(0,0)};
    \node[below] at (0,0) {$y_j$};
    
    \draw plot [mark=*, mark size=2] coordinates{(2,0)};
    \node[below] at (2,0) {$x_{j+1}$};
    
    \draw plot [mark=*, mark size=2] coordinates{(1,0)};
    \node[below] at (1,0) {$p$};
    
    \end{tikzpicture}
    \caption{Illustration of the proof of lemma \ref{lem:adjacent-cutoffs-agree}. The ray $\gamma_{r_{j_+},-}$ meets the Cauchy surface $\Sigma$ at $y_j$, and we label the point where $\gamma_{r_{j_+},+}$ meets sigma as $p$. Because all $\mathcal{V}_k$ must lie to the past of all $r_k$, each of the $x_i, y_i$ must lie to the past of the light rays $\gamma_{r_{j_+},+}$ and $\gamma_{r_{j_+},-}$, so $p$ must not be to the right of $x_{j+1}$. Because $p$ must be a $y_i$ or a $x_i$ by the counting argument, we must have $p=x_{j+1}$.}
    \label{fig:pandSigmaproof}
\end{figure}

Another fact we will need is that while there are no $2$-to-all scattering regions on the boundary, there are $2$-to-$k$ scattering regions. The following lemmas describe the causal connections of this type.

\begin{lemma}\label{lemma:boundarycausalstructure}
    Let $c_j$ and $c_k$ be two input points. The scattering regions
    \begin{equation}
        \hat{J}(c_j, c_{k} \rightarrow r_1, \dots, r_{j-1}, r_{k}, \dots, r_n)
    \end{equation}
    and
    \begin{equation}
        \hat{J}(c_{j}, c_k \rightarrow r_j, \dots, r_{k-1})
    \end{equation}
    are nonempty.

    Furthermore, they are maximally disjoint, in that any point in $\hat{J}^+(c_j) \cup \hat{J}^+(c_k)$ that can signal a point in $\{r_1, \dots, r_{j-1}, r_k, \dots, r_n\}$ cannot signal any point in $\{r_j, \dots, r_{k-1}\},$ and vice versa.
\end{lemma}
\begin{proof}
    See figure \ref{fig:boundarycausalstructure} for an illustration of this proof.
    
    The light rays $\gamma_{c_j, +}$ and $\gamma_{c_{k}, -}$ meet for the first time at a point we call $\chi.$
    The light rays $\gamma_{c_j, -}$ and $\gamma_{c_{k}, +}$ meet for the first time at a point we call $\chi'.$
    For a point to be in the future of both $c_j$ and $c_k,$ it must be in the future of either $\chi$ or $\chi'$, i.e., we have
    \begin{equation}
        \hat{J}^+(\chi) \cup \hat{J}^+(\chi')
            = \hat{J}^+(c_j) \cap \hat{J}^+(c_{k}).
    \end{equation}
    Since all of the output points are in both $\hat{J}^+(c_j)$ and $\hat{J}^+(c_{k}),$ this implies
    \begin{equation} \label{eq:in-chi-futures}
        \{r_1, \dots, r_n\} \subseteq \hat{J}^+(\chi) \cup \hat{J}^+(\chi').
    \end{equation}
    Which output points are in which set?
    To solve this problem, we note first that no output point can be null separated from either $\chi$ or $\chi'.$
    This is because all past-directed null geodesics leaving $\chi$ and $\chi'$ eventually overlap with portions of the past horizon of either $\mathcal{V}_j$ or $\mathcal{V}_{k}$, but the previous lemmas on the geometry of the input regions tell us that no null ray leaving an output point has this property.
    
    Fix an output point $r_k.$ Since it cannot be null separated from $\chi$ or $\chi',$ it must by equation \eqref{eq:in-chi-futures} be in the timelike future of at least one of them.
    But if $r_k$ is in the timelike future of $\chi,$ then by lemma \ref{lem:causal-timelike-addition} it is in the timelike future of all input region endpoints from $y_j$ to $x_{k}.$
    If it is in the timelike future of $\chi',$ then by the same logic it is in the timelike future of all input region endpoints from $y_{k}$ to $x_{j}.$
    But we know from lemma \ref{lem:adjacent-cutoffs-agree} that $r_k$ is null separated, \textit{not} timelike separated, from $y_k$ and $x_{k+1}.$
    This tells us that the only output points which \textit{can} be in the future of $\chi$ are $\{r_1, \dots, r_{j-1}, r_{k}, \dots, r_n\},$ and the only output points which \textit{can} be in the future of $\chi'$ are $\{r_{j}, \dots, r_{k-1}\}.$
    But since these sets are complementary, and since $\hat{J}^+(\chi)$ and $\hat{J}^+(\chi')$ collectively must contain all output points, we can conclude
    \begin{equation}
        \chi \in \hat{J}(c_j, c_k \rightarrow r_1, \dots, r_{j-1}, r_{k}, \dots, r_n)
    \end{equation}
    and
    \begin{equation}
        \chi' \in \hat{J}(c_j, c_k \rightarrow r_j, \dots, r_{k-1}).
    \end{equation}

    This proves the claim that these two scattering are nonempty.
    The second claim follows from the above discussion; any point in $\hat{J}^+(c_j) \cap \hat{J}^+(c_k)$ that can signal a point in $\{r_1, \dots, r_{j-1}, r_k, \dots, r_n\}$ must be in the future of $\chi$ and not $\chi',$ while any point in $\hat{J}^+(c_j) \cap \hat{J}^+(c_k)$ that can signal a point in $\{r_j, \dots, r_{k-1}\}$ must be in the future of $\chi'$ and not $\chi.$
\end{proof}

\begin{figure}
    \centering
    \begin{tikzpicture}[scale=0.75]
    
    \draw[red] (0,6) -- (6,0) -- (12,6);
    
    \draw[red] (2,0) -- (8,6);
    \draw[red] (2,0) -- (0,2);
    \draw[red] (12,2) -- (8,6);
    
    \draw[thick, black] (0,0) -- (12,0) -- (12,6) -- (0,6) -- (0,0);
    
    \draw plot [mark=*, mark size=2] coordinates{(10,4)};
    \node[right] at (4.1,2) {$\chi$};
    
    \draw plot [mark=*, mark size=2] coordinates{(4,2)};
    \node[right] at (10.1,4) {$\chi'$};
    
    \draw plot [mark=*, mark size=2] coordinates{(6,0)};
    \node[below] at (6,0) {$c_2$};
    
    \draw plot [mark=*, mark size=2] coordinates{(2,0)};
    \node[below] at (2,0) {$c_{1}$};
    
    \draw plot [mark=*, mark size=2] coordinates{(10,0)};
    \node[below] at (10,0) {$c_{3}$};
    
    \draw plot [mark=*, mark size=2] coordinates{(6,6)};
    \node[above] at (6,6) {$r_3$};
    
    \draw plot [mark=*, mark size=2] coordinates{(2,6)};
    \node[above] at (10,6) {$r_{1}$};
    
    \draw plot [mark=*, mark size=2] coordinates{(10,6)};
    \node[above] at (2,6) {$r_{2}$};
    
    \node[above right,rotate=45] at (7.7,1.5) {$\gamma_{c_{2},+}$};
    \node[above left,rotate=-45] at (5.7,0.3) {$\gamma_{c_{2},-}$};
    
    \node[above right,rotate=45] at (2.5,0.4) {$\gamma_{c_{1},+}$};
    \node[above left,rotate=-45] at (1.5,0.4) {$\gamma_{c_{1},-}$};
    
    \node[above left,rotate=-45] at (11.75,2.25) {$\gamma_{c_{1},-}$};
    
    \end{tikzpicture}
    \caption{Illustration of lemma \ref{lemma:boundarycausalstructure}, shown for the case $n=3$. }
    \label{fig:boundarycausalstructure}
\end{figure}
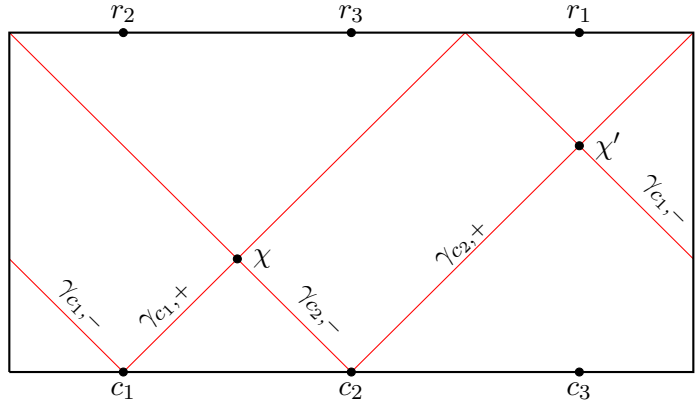

We summarize the facts about boundary causal structure understood in this section in figure \ref{fig:summaryofboundarycausalstructure}. 

\begin{figure}
    \centering
    \begin{tikzpicture}[scale=0.75]
    
    \draw[fill=blue,opacity=0.3] (2,0) -- (3,1) -- (2,2) -- (1,1);
    \draw[fill=blue,opacity=0.3] (6,0) -- (7,1) -- (6,2) -- (5,1);
    \draw[fill=blue,opacity=0.3] (10,0) -- (11,1) -- (10,2) -- (9,1);
    
    \draw[fill=green,opacity=0.3] (4,2) -- (5,3) -- (4,4) -- (3,3) -- (4,2);
    \draw[fill=green,opacity=0.3] (10,4) -- (11,5) -- (10,6) -- (9,5) -- (10,4);
    
    \draw[red] (0,6) -- (6,0) -- (12,6);
    
    \draw[red] (2,0) -- (8,6);
    \draw[red] (2,0) -- (0,2);
    \draw[red] (12,2) -- (8,6);
    \draw[red] (10,0) -- (12,2);
    \draw[red] (10,0) -- (4,6);
    \draw[red] (0,2) -- (4,6);
    
    \draw[red] (2,6) -- (0,4);
    \draw[red] (6,6) -- (0,0);
    \draw[red] (6,6) -- (12,0);
    \draw[red] (2,6) -- (8,0);
    \draw[red] (10,6) -- (4,0);
    \draw[red] (10,6) -- (12,4);
    \draw[red] (0,4) -- (4,0);
    \draw[red] (12,4) -- (8,0);
    
    \draw[thick, black] (0,0) -- (12,0) -- (12,6) -- (0,6) -- (0,0);
    
    \draw plot [mark=*, mark size=2] coordinates{(6,0)};
    \node[below] at (6,0) {$c_2$};
    
    \draw plot [mark=*, mark size=2] coordinates{(2,0)};
    \node[below] at (2,0) {$c_{1}$};
    
    \draw plot [mark=*, mark size=2] coordinates{(10,0)};
    \node[below] at (10,0) {$c_{3}$};
    
    \draw plot [mark=*, mark size=2] coordinates{(6,6)};
    \node[above] at (6,6) {$r_3$};
    
    \draw plot [mark=*, mark size=2] coordinates{(2,6)};
    \node[above] at (10,6) {$r_{1}$};
    
    \draw plot [mark=*, mark size=2] coordinates{(10,6)};
    \node[above] at (2,6) {$r_{2}$};
    
    \end{tikzpicture}
    \caption{Summary of the boundary causal structure, as revealed by the lemmas in this section. Each input region (blue) is the domain of dependence of a single interval (Lemma \ref{lemma:Visdomainofsigma}). The negatively oriented future boundary of $\mathcal{V}_1$ and positively oriented future boundary of $\mathcal{V}_2$ are defined by $r_1$, and similarly for $\mathcal{V}_2$, $\mathcal{V}_3$ and $\mathcal{V}_3,\mathcal{V}_1$ (Lemma \ref{lem:adjacent-cutoffs-agree}). The joint future of $c_1$ and $c_2$ contains a region in the past of $r_2$ and $r_3$ (green), and a disjoint region in the past of $r_1$ (also green) (Lemma \ref{lemma:boundarycausalstructure}).}
    \label{fig:summaryofboundarycausalstructure}
\end{figure}
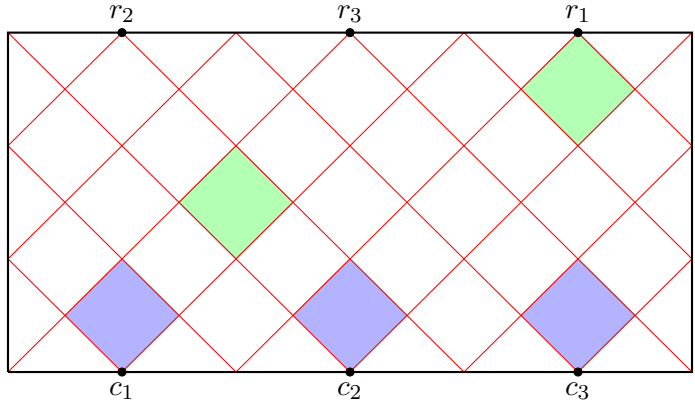

\section{Proving the \texorpdfstring{$n$}{TEXT}-to-\texorpdfstring{$n$}{TEXT} connected wedge theorem}
\label{sec:GR-proof}

We begin by recalling the statement of the $n$-to-$n$ connected wedge theorem given in section \ref{sec:overview}.

\noindent \textbf{Theorem \ref{thm:n-to-n-first-statement}:} \emph{Let $c_1, \dots, c_n, r_1, \dots, r_n$ be points on a global AdS$_{2+1}$ boundary such that the sets $\mathcal{V}_j$ and $\mathcal{W}_j$ are nonempty, while also satisfying
    \begin{equation}
        \hat{J}(c_j, c_k \rightarrow r_1, \dots, r_n) = \hat{J}(c_1, \dots, c_n \rightarrow r_j, r_k) = \varnothing.
    \end{equation}
Let $\mathcal{M}$ be an asymptotically AdS$_{2+1},$ AdS-hyperbolic spacetime satisfying the null curvature condition, with at least one global AdS$_{2+1}$ boundary on which $c_1, \dots, c_n, r_1, \dots, r_n$ are specified.}

\emph{If the $2$-to-all causal graph $\Gamma_{2\rightarrow\text{all}}$ is connected, and if the HRRT surface of $\mathcal{V}_1 \cup \dots \cup \mathcal{V}_n$ can be found using the maximin formula \cite{maximin, maximin2}, then the entanglement wedge of $\mathcal{V}_1 \cup \dots \cup \mathcal{V}_n$ is connected.}

The $2$-to-all causal graph $\Gamma_{2 \rightarrow \text{all}}$ used in the statement of the theorem is defined as the graph whose vertices are $\{1, \dots, n\},$ where there is an edge $j \sim k$ if the bulk region $J(E_{\mathcal{V}_j}, E_{\mathcal{V}_k} \rightarrow E_{\mathcal{W}_1}, \dots, E_{\mathcal{W}_n})$ has nonempty interior.
The proof has two essential tools, summarized below.

\begin{enumerate}
    \item \textit{The maximin formula}: arguments given in \cite{maximin} suggest that in many asymptotically AdS spacetimes, the HRRT surface of a boundary region $\mathcal{R}$ can be found by finding the minimal surface on each ``complete achronal slice,'' then maximizing over all such slices.
    The key consequence of this prescription is that for any boundary region $\mathcal{R}$, there exists a bulk slice $\Sigma$, which is a Cauchy slice for the conformally completed spacetime in the sense of appendix \ref{app:glossary}, on which the HRRT surface $\gamma_\mathcal{R}$ is globally minimal within its homology class.
    The main obstruction to making maximin rigorous is that its proof requires manipulating infinite area surfaces as though their areas were finite.
    Manipulations of this kind are also used in the proof of theorem \ref{thm:n-to-n-first-statement}, and the assumption that these manipulations are legitimate is included in the assumption that ``the HRRT surface of $\mathcal{V}_1 \cup \dots \cup \mathcal{V}_n$ can be found using the maximin formula.''
    
    We will actually make use of the restricted maximin formula \cite{maximin2}, which says that asymptotics of the maximin slice can be specified arbitrarily away from the spatial boundary of $\mathcal{R}$.
    In other words, we can specify an arbitrary Cauchy slice $\sigma$ of the boundary spacetime containing a complete slice of $\mathcal{R}$, and there will exist a Cauchy slice for the conformally completed spacetime whose boundary is $\sigma$ and on which $\gamma_\mathcal{R}$ is globally minimal within its homology class.
    
    \item \textit{The focusing theorem}: the Raychaudhuri equation \cite{raychaudhuri} is a differential equation that controls how nearby geodesics expand, shear, and twist relative to one another.
    An important special case of the Raychaudhuri equation is its restriction to a codimension-$1$ surface generated by null geodesics, with affine parameter $\lambda.$
    In three spacetime dimensions, this restriction is
    \begin{equation}
        \frac{d \theta}{d \lambda}
            = - \theta^2 - \sigma_{ab} \sigma^{ab} - R_{ab} k^a k^b,
    \end{equation}
    where $k^a$ is the affinely parametrized generator of the surface, $R_{ab}$ is the Ricci tensor, $\sigma_{ab}$ is the purely spatial ``shear tensor,'' and $\theta$ is the expansion of the surface.
    The expansion measures the local change in area of cross-sections of the surface in the sense that the difference in area of two cross-sections is given by the integral of $\theta$ over the portion of the surface contained between them --- see e.g. section 3.1 of \cite{may2020holographic} for a review.
    Given the null curvature condition $R_{ab} k^a k^b \geq 0,$ Raychaudhuri's equation implies
    \begin{equation}
        \frac{d \theta}{d \lambda} \leq 0,
    \end{equation}
    which is the focusing theorem.
    
    The essential consequence of the focusing theorem for the present work is that if a geodesically parametrized null surface has a cross-section with nonpositive expansion --- for example, an extremal surface --- then all ``later'' cross-sections have area no larger than the original one, since the expansion is nonpositive in between those slices.
\end{enumerate}

Using these tools, and the following assumption, we will now prove theorem \ref{thm:n-to-n-first-statement}.

\begin{assumption}
    Much of the following proof is about the geometry of two-dimensional null hypersurfaces.
    We will assume, throughout, that the intersection of two such hypersurfaces is a continuous curve, possibly with multiple connected components.
    This is not always true --- null hypersurfaces can intersect at isolated points --- but such situations do not often arise in practice, and treating them carefully adds nothing to the conceptual insight provided by theorem \ref{thm:n-to-n-first-statement}.
    Addressing these issues would amount to checking a large number of edge cases, which we have chosen not to do below.
\end{assumption}

We view this assumption as much less dramatic than the assumptions that go into the maximin formula, so we do not believe it constitutes a loss of rigor.

\subsection{Proof}

\subsubsection*{Setup}

As explained in appendix \ref{app:glossary}, the fact that the spacetime $\mathcal{M}$ is AdS-hyperbolic means that there is a conformal completion $\tilde{\mathcal{M}}$ that is globally hyperbolic with compact Cauchy surfaces.
By the assumptions of theorem \ref{thm:n-to-n-first-statement}, at least one of the boundaries of $\tilde{\mathcal{M}}$ is conformally equivalent to a flat timelike cylinder; this is the boundary on which the points $c_1, \dots, c_n, r_1, \dots, r_n$ are specified.
Because causal structure is a conformal invariant, all of the lemmas about timelike cylinders proved in section \ref{sec:timelike-cylinders} apply to this boundary.

In particular, the regions $\mathcal{V}_1, \dots, \mathcal{V}_n$ are causal diamonds, each being the domain of dependence of a single achronal interval.
The causal complement $(\mathcal{V}_1 \cup \dots \cup \mathcal{V}_n)',$ which is closure of the set of all boundary points spacelike separated from $\mathcal{V}_1 \cup \dots \cup \mathcal{V}_n,$ is also made up of $n$ causal diamonds $\mathcal{X}_1, \dots, \mathcal{X}_n$, each of which is the domain of dependence of an interval that connects two neighboring regions $\mathcal{V}_j$ and $\mathcal{V}_{j+1}$.
See figure \ref{fig:X-regions}.
Each of these $\mathcal{X}_j$ diamonds has a \textit{future horizon}, which is the set of points $p \in \mathcal{X}_j$ for which $I^+(p) \cap \mathcal{X}_j$ is empty; the domain of dependence of a future horizon is the full diamond $\mathcal{X}_j.$ Again see figure \ref{fig:X-regions}.

\begin{figure}
\centering
\begin{tikzpicture}

\draw[thick] (0,0) -- (12,0) -- (12,6) -- (0,6) -- (0,0);

\draw (0,3) -- (1.33,4.33) -- (2.66,3) -- (1.33,1.66) -- (0,3);
\draw (2.66,3) -- (3.33,3.66) -- (4,3) -- (3.33,2.33) -- (2.66,3);
\node at (1.33,3) {$\mathcal{V}_1$};
\node at (3.33,3) {$\mathcal{X}_1$};

\draw[thick,green] (2.66,3.05) -- (3.33,3.66+0.05) -- (4,3.05);
\node[above] at (3.33,3.66) {$H^+(\mathcal{X}_1)$};

\draw (4,3) -- (5.33,4.33) -- (6.66,3) -- (5.33,1.66) -- (4,3);
\draw (6.66,3) -- (7.33,3.66) -- (8,3) -- (7.33,2.33) -- (6.66,3);
\node at (5.33,3) {$\mathcal{V}_2$};
\node at (7.33,3) {$\mathcal{X}_2$};

\draw (8,3) -- (9.33,4.33) -- (10.66,3) -- (9.33,1.66) -- (8,3);
\draw (10.66,3) -- (11.33,3.66) -- (12,3) -- (11.33,2.33) -- (10.66,3);
\node at (9.33,3) {$\mathcal{V}_3$};
\node at (11.33,3) {$\mathcal{X}_3$};

\end{tikzpicture}
\caption{On the boundary cylinder, the causal complement of $\mathcal{V}_1 \cup \mathcal{V}_2 \cup \mathcal{V}_3$ is made up of three causal diamonds $\mathcal{X}_1, \mathcal{X}_2,$ and $\mathcal{X}_3.$ The future horizon of $\mathcal{X}_1$ is marked in green and denoted $H^+(\mathcal{X}_1).$}
\label{fig:X-regions}
\end{figure}
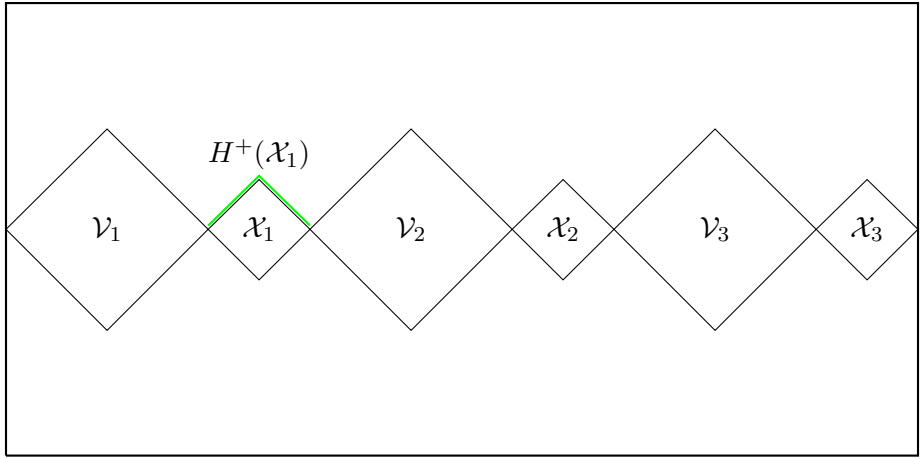

If we pick arbitrary slices of the regions $\mathcal{V}_1, \dots, \mathcal{V}_n,$ and glue them to the future horizons of each $X_j,$ we obtain a Cauchy slice for the boundary spacetime.
This is a legal setting for the restricted maximin formula \cite{maximin2}, so we are free to choose a bulk maximin slice for $\mathcal{V}_1 \cup \dots \cup \mathcal{V}_n$ that intersects each $X_j$ on its future horizon.

Let $\Sigma$ be a such a maximin slice, considered as a Cauchy surface for the conformally completed spacetime $\tilde{M}$.
Denote by ${E}_{\mathcal{V}}$ the entanglement wedge of $\mathcal{V}\equiv \mathcal{V}_1 \cup \dots \cup \mathcal{V}_n$; we remind the reader (cf. definition \ref{def:HRRT}) that $E_{\mathcal{V}}$ is the domain of dependence of the region of $\Sigma$ that is bounded by $\gamma_{\mathcal{V}}$ --- the HRRT surface of $\mathcal{V}=\mathcal{V}_1 \cup \dots \cup \mathcal{V}_n$ --- and $\Sigma \cap (\mathcal{V}_1 \cup \dots \cup \mathcal{V}_n)$.
As per the convention chosen in appendix \ref{app:glossary}, we will think of $E_{\mathcal{V}}$ as a closed set in the full, conformally completed spacetime $\tilde{\mathcal{M}},$ so that it includes points on the spacetime boundary.

We assume, toward contradiction, that $E_{\mathcal{V}}$ is disconnected.
If this is the case, then $E_{\mathcal{V}} \cap \Sigma$ must be disconnected.\footnote{If a surface $\sigma$ is connected, then its domain of dependence $D(\sigma)$ is connected, since every point in $D(\sigma)$ is connected to $\sigma$ by causal curves. Since $E_{\mathcal{V}}$ is the domain of dependence of $E_{\mathcal{V}} \cap \Sigma,$ we cannot have $E_{\mathcal{V}}$ disconnected but $E_{\mathcal{V}} \cap \Sigma$ connected.}
We give two examples of what this might look like for $n=3$ in figure \ref{fig:disconnected-EW}.
Under the assumption that $E_{\mathcal{V}} \cap \Sigma$ is disconnected, we will construct a non-$\gamma_{\mathcal{V}}$ surface on $\Sigma,$ which is homologous to $\gamma_{\mathcal{V}},$ and which has area strictly less than that of $\gamma_{\mathcal{V}}.$
This contradicts the maximin formula, and proves the theorem.

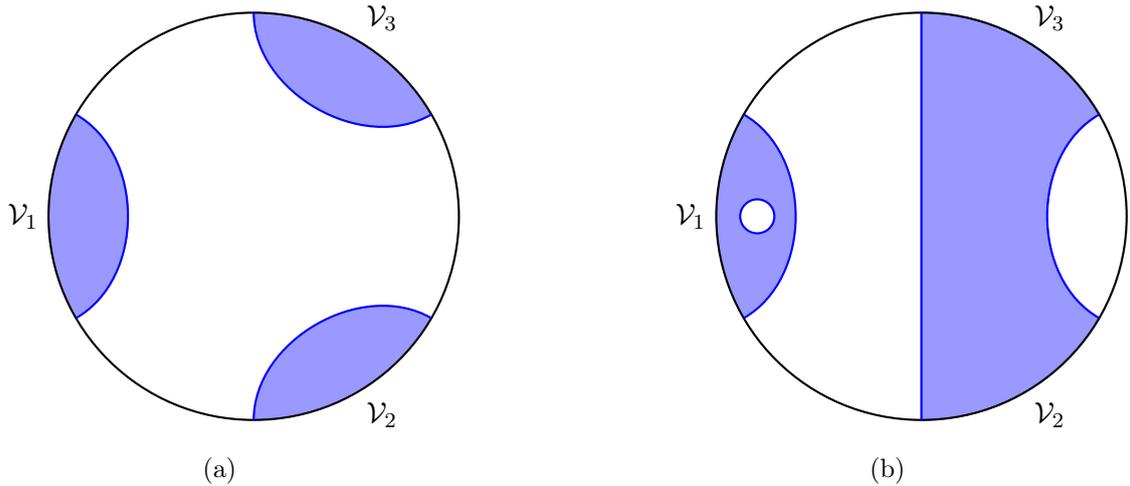
\begin{figure}
    \centering
    \subfloat[\label{fig:disconnected-EW-left}]{
    \begin{tikzpicture}[scale=0.9]
    
    \draw[domain=-120:-60,variable=\x,smooth, fill=blue,opacity=0.4] plot ({3*sin(\x)}, {3*cos(\x)}) -- (-3*0.866,3*0.5) to [out=-30,in=30] (-3*0.866,-3*0.5);
    
    \begin{scope}[rotate=120]
    \draw[domain=-120:-60,variable=\x,smooth, fill=blue,opacity=0.4] plot ({3*sin(\x)}, {3*cos(\x)}) -- (-3*0.866,3*0.5) to [out=-30,in=30] (-3*0.866,-3*0.5);
    \end{scope}
    
    \begin{scope}[rotate=240]
    \draw[domain=-120:-60,variable=\x,smooth, fill=blue,opacity=0.4] plot ({3*sin(\x)}, {3*cos(\x)}) -- (-3*0.866,3*0.5) to [out=-30,in=30] (-3*0.866,-3*0.5);
    \end{scope}
    
    \draw[thick] (0,0) circle (3);
    
    \draw[thick,blue] (-3*0.866,3*0.5) to [out=-30,in=30] (-3*0.866,-3*0.5);
    \draw[thick,blue] (0,3) to [out=-90,in=-150] (3*0.866,3*0.5);
    \draw[thick,blue] (0,-3) to [out=90,in=150] (3*0.866,-3*0.5);
    
    \node[left] at (-3,0) {$\mathcal{V}_1$};
    \node[above right] at (3*0.5,3*0.866) {$\mathcal{V}_3$};
    \node[below right] at (3*0.5,-3*0.866) {$\mathcal{V}_2$};
    
    \end{tikzpicture}
    }
    \hfill
    \subfloat[\label{fig:disconnected-EW-right}]{
    \begin{tikzpicture}[scale=0.9]
    
    \draw[domain=-120:-60,variable=\x,smooth, fill=blue,opacity=0.4] plot ({3*sin(\x)}, {3*cos(\x)}) -- (-3*0.866,3*0.5) to [out=-30,in=30] (-3*0.866,-3*0.5);
    \draw[thick,blue,fill=white] (-2.4,0) circle (0.25);
    
    \draw[domain=0:60,variable=\x,thick,blue,fill=blue,opacity=0.4] plot ({3*sin(\x+120)}, {3*cos(\x+120)}) -- (0,3) plot ({3*sin(\x)}, {3*cos(\x)}) -- (3*0.866,3*0.5) to [out=-150,in=150] (3*0.866,-3*0.5);
    
    \draw[thick] (0,0) circle (3);
    
    \draw[thick,blue] (-3*0.866,3*0.5) to [out=-30,in=30] (-3*0.866,-3*0.5);
    \draw[thick,blue] (0,3) -- (0,-3);
    \draw[thick,blue] (3*0.866,3*0.5) to [out=-150,in=150] (3*0.866,-3*0.5);
    
    \node[left] at (-3,0) {$\mathcal{V}_1$};
    \node[above right] at (3*0.5,3*0.866) {$\mathcal{V}_3$};
    \node[below right] at (3*0.5,-3*0.866) {$\mathcal{V}_2$};
    
    \end{tikzpicture}
    }
    \caption{Two examples of disconnected entanglement wedges $E_{\mathcal{V}}$ for the case $n=3$, restricted to the maximin slice $\Sigma.$ In (a) the entanglement wedge has three components. In (b) it has two components; there is one component connecting $\mathcal{V}_2$ to $\mathcal{V}_3,$ and the component containing $\mathcal{V}_1$ includes a black hole horizon.}
    \label{fig:disconnected-EW}
\end{figure}

The first step is to construct the \textit{focusing lightsheet}, given by\footnote{Note that throughout this paper $\bar{A}$ denotes the closure of a set $A$.}
\begin{equation}
    \mathscr{F} \equiv \bar{\del J^+(E_{\mathcal{V}} \cap \Sigma) - (E_{\mathcal{V}} \cap \Sigma)}.
\end{equation}
The symbol $\del J^+(E_{\mathcal{V}} \cap \Sigma)$ denotes the boundary of the future of $E_{\mathcal{V}} \cap \Sigma.$
As explained in lemma \ref{lem:lightsheet-generators}, every point in this set is either in $E_{\mathcal{V}} \cap \Sigma,$ or lies on a null geodesic with past endpoint on $\partial_0 (E_{\mathcal{V}} \cap \Sigma).$\footnote{This result uses the global hyperbolicity of $\tilde{\mathcal{M}}$ and the compactness of $\Sigma$, which implies via lemma \ref{lem:lightsheet-boundary} the identity $\del J^+(E_{\mathcal{V}} \cap \Sigma) = J^+(E_{\mathcal{V}} \cap \Sigma) - I^+(E_{\mathcal{V}} \cap \Sigma).$}
The set $\mathscr{F}$ is defined by throwing away all the points that actually lie in $E_{\mathcal{V}} \cap \Sigma,$ then adding back in the edge $\partial_0 (E_{\mathcal{V}} \cap \Sigma),$ which is just the HRRT surface $\gamma_{\mathcal{V}}.$
A sample focusing lightsheet is sketched in figure \ref{fig:sample-lightsheet}.
Crucially, every point on $\mathscr{F}$ lies on a null geodesic leaving the HRRT surface $\gamma_{\mathcal{V}}.$
Furthermore, because null geodesics starting at $\gamma_{\mathcal{V}}$ leave $\partial J^+(E_{\mathcal{V}} \cap \Sigma)$ if their expansion diverges to $-\infty$ (cf. theorem 7.27 of \cite{penrose1972techniques}), the expansion at each point of $\mathscr{F}$ is nonpositive under any choice of affine parameter.
So any complete cross-section of $\mathscr{F}$ has area less than or equal to the area of $\gamma_{\mathcal{V}}.$

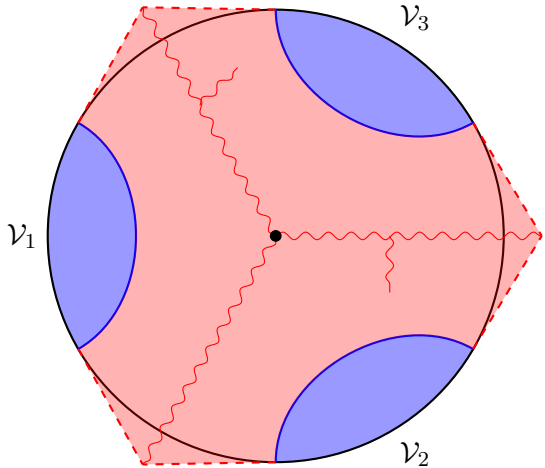
\begin{figure}
\centering
\begin{tikzpicture}

    \draw[domain=-120:-60,variable=\x,smooth, fill=blue,opacity=0.4] plot ({3*sin(\x)}, {3*cos(\x)}) -- (-3*0.866,3*0.5) to [out=-30,in=30] (-3*0.866,-3*0.5);
    
    \begin{scope}[rotate=120]
    \draw[domain=-120:-60,variable=\x,smooth, fill=blue,opacity=0.4] plot ({3*sin(\x)}, {3*cos(\x)}) -- (-3*0.866,3*0.5) to [out=-30,in=30] (-3*0.866,-3*0.5);
    \end{scope}
    
    \begin{scope}[rotate=240]
    \draw[domain=-120:-60,variable=\x,smooth, fill=blue,opacity=0.4] plot ({3*sin(\x)}, {3*cos(\x)}) -- (-3*0.866,3*0.5) to [out=-30,in=30] (-3*0.866,-3*0.5);
    \end{scope}
    
    \draw[thick] (0,0) circle (3);
    
    \draw[thick,blue] (-3*0.866,3*0.5) to [out=-30,in=30] (-3*0.866,-3*0.5);
    \draw[thick,blue] (0,3) to [out=-90,in=-150] (3*0.866,3*0.5);
    \draw[thick,blue] (0,-3) to [out=90,in=150] (3*0.866,-3*0.5);
    
    \node[left] at (-3,0) {$\mathcal{V}_1$};
    \node[above right] at (3*0.5,3*0.866) {$\mathcal{V}_3$};
    \node[below right] at (3*0.5,-3*0.866) {$\mathcal{V}_2$};
    
    \draw[red,thick,dashed] ({3*cos(30)},{3*sin(30)}) -- (3.5,0);
    \draw[red,thick,dashed] ({3*cos(30)},{3*sin(-30)}) -- (3.5,0);
    \draw[red,decorate, decoration={snake, segment length=3mm, amplitude=0.5mm}] (0,0) -- (3.5,0);
    
    \draw[red,thick,dashed] ({3*cos(150)},{3*sin(150)}) -- ({3.5*cos(120)},{3.5*sin(120)});
    \draw[red,thick,dashed] ({3*cos(90)},{3*sin(90)}) -- ({3.5*cos(120)},{3.5*sin(120)});
    \draw[red,decorate, decoration={snake, segment length=3mm, amplitude=0.5mm}] (0,0) -- ({3.5*cos(120)},{3.5*sin(120)});
    
    \draw[red,thick,dashed] ({3*cos(-150)},{3*sin(-150)}) -- ({3.5*cos(-120)},{3.5*sin(-120)});
    \draw[red,thick,dashed] ({3*cos(-90)},{3*sin(-90)}) -- ({3.5*cos(-120)},{3.5*sin(-120)});
    \draw[red,decorate, decoration={snake, segment length=3mm, amplitude=0.5mm}] (0,0) -- ({3.5*cos(-120)},{3.5*sin(-120)});
    
    \draw[red,decorate, decoration={snake, segment length=3mm, amplitude=0.5mm}] (1.5,0) -- ({1.5},{-0.75});
    
    \draw[red,decorate, decoration={snake, segment length=3mm, amplitude=0.5mm}] ({2*cos(120)},{2*sin(120)}) -- ({2*cos(120)+0.5},{2*sin(120)+0.5});
    
    \fill[red,opacity=0.3] (0,0) -- (3.5,0) -- ({3*cos(30)},{3*sin(-30)}) -- ({3*cos(-30)},{3*sin(-30)}) to [out=150,in=90] (0,{3*sin(-90)}) -- ({3.5*cos(-120)},{3.5*sin(-120)});
    
    \begin{scope}[rotate=120]
    \fill[red,opacity=0.3] (0,0) -- (3.5,0) -- ({3*cos(30)},{3*sin(-30)}) -- ({3*cos(-30)},{3*sin(-30)}) to [out=150,in=90] (0,{3*sin(-90)}) -- ({3.5*cos(-120)},{3.5*sin(-120)});
    \end{scope}
    
    \begin{scope}[rotate=-120]
    \fill[red,opacity=0.3] (0,0) -- (3.5,0) -- ({3*cos(30)},{3*sin(-30)}) -- ({3*cos(-30)},{3*sin(-30)}) to [out=150,in=90] (0,{3*sin(-90)}) -- ({3.5*cos(-120)},{3.5*sin(-120)});
    \end{scope}

    \draw plot [mark=*, mark size=2] coordinates{(0,0)};
    
\end{tikzpicture}
\caption{A focusing lightsheet $\mathscr{F}$ for the extremal surface configuration sketched in figure \ref{fig:disconnected-EW}a. This figure is drawn so that the ``future'' direction is towards the viewer. The maximin slice sketched in figure \ref{fig:disconnected-EW} is in the background, and the red shaded region is $\mathscr{F}.$ The dashed red lines are future-directed light rays on the spacetime boundary that show where $\mathscr{F}$ intersects the boundary. Wavy lines are caustics where distinct generators of $\mathscr{F}$ intersect one another.}
\label{fig:sample-lightsheet}
\end{figure}

We have a particular cross-section in mind, which we will specify by determining, for each generator of $\mathscr{F}$, a point where that generator should be truncated.
If the choice of truncation point is made continuously, then the union of all truncation points is a cross-section of $\mathscr{F}$, and has area less than or equal to the area of $\gamma_{\mathcal{V}}$.
The next few subsections are dedicated to constructing this cross-section.

\subsubsection*{Dividing curves}

To choose our truncation points, it will be helpful to introduce the \textit{past lightsheet} leaving an output wedge $E_{\mathcal{W}_j}$, defined by
\begin{equation}
    \mathscr{P}_j = \bar{\partial J^-(E_{\mathcal{W}_j} \cap \chi_j) - (E_{\mathcal{W}_j} \cap \chi_j)},
\end{equation}
where $\chi_j$ is any Cauchy surface of $\tilde{\mathcal{M}}$ containing the spatial boundary of $E_{\mathcal{W}_j}$.
AdS-hyperbolicity of $\mathcal{M}$ implies $\del J^-(E_{\mathcal{W}_j} \cap \chi_j) = J^-(E_{\mathcal{W}_j} \cap \chi_j) - I^-(E_{\mathcal{W}_j} \cap \chi_j)$ via lemma \ref{lem:lightsheet-boundary}.
Each output wedge $E_{\mathcal{W}_j}$ limits on the boundary to an output region $\mathcal{W}_{j}$, and we know from lemma \ref{lem:adjacent-cutoffs-agree} that the right- and left-generators of $\del \hat{J}^{-}(\mathcal{W}_{j})$ intersect the right- and left-endpoints of two adjacent input regions.
As in section \ref{sec:timelike-cylinders}, we choose to label the output regions so that the right generator of $\del \hat{J}^{-}(\mathcal{W}_{j})$ intersects the right endpoint of $\mathcal{V}_{j},$ and the left generator of $\del \hat{J}^-(\mathcal{W}_{j})$ intersects the left endpoint of $\mathcal{V}_{j+1}$.

We will now, for each $j$, give an algorithm for defining a special bulk curve $\alpha_j$, called the \textbf{dividing curve}, that starts at the right endpoint of $\mathcal{V}_{j}$ and ends at the left endpoint of $\mathcal{V}_{j+1}.$
If the component of $\gamma_{\mathcal{V}}$ containing the right endpoint of $\mathcal{V}_{j}$ has, for its other endpoint, the left endpoint of $\mathcal{V}_{j+1},$ then we will define the dividing curve $\alpha_j$ to be this component of $\gamma_{\mathcal{V}}.$
In any other case, we know that the other endpoint of this component of $\gamma_{\mathcal{V}}$ lies in $\hat{I}^-(\mathcal{W}_{j})$, and therefore in $I^-(E_{\mathcal{W}_j}).$
Causal consistency of entanglement wedge reconstruction, which is known to follow either from the null curvature condition \cite{maximin, headrick2014causality} or from the generalized second law \cite{engelhardt2015quantum}, implies that the right endpoint of $\mathcal{V}_{j}$ is not in $I^-(E_{\mathcal{W}_j}),$ since it is not in $\hat{I}^-(\mathcal{W}_{j}).$
So if we start at the right endpoint of $\mathcal{V}_{j}$ and follow $\gamma_{\mathcal{V}}$ into the bulk, we will eventually enter $I^-(E_{\mathcal{W}_j})$ just after crossing through $\mathscr{P}_j \cap \mathscr{F}.$
See figure \ref{fig:dividing-curve-start}a.
The first portion of the dividing curve $\alpha_j$ will be the segment of $\gamma_{\mathcal{V}}$ that lies between this intersection and our starting point, which was the right endpoint of $\mathcal{V}_{j}.$

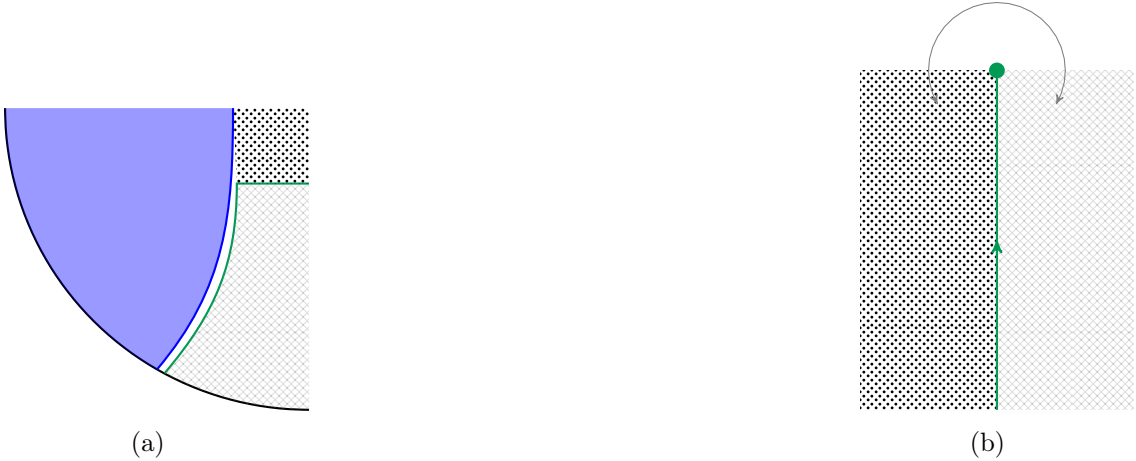
\begin{figure}
    \centering
    \subfloat[\label{fig:withoutmatter}]{
    \begin{tikzpicture}[scale=1]
    
    \draw[thick,domain=-180:-90,variable=\x,smooth] plot ({4*sin(\x)}, {4*cos(\x)});
    
    \draw[blue,thick] ({4*sin(-150)},{4*cos(-150)}) to [out=50,in=-90] (-1,0);
    \draw[ForestGreen,thick] ({4*sin(-150)+0.1},{4*cos(-150)-0.05}) to [out=50,in=-90] (-1+0.05,-1) -- (0,-1);
    
    \fill[pattern={crosshatch},pattern color=gray, opacity=0.5,domain=-180:-150,variable=\x,smooth] ({4*sin(-150)+0.1},{4*cos(-150)-0.05}) to [out=50,in=-90] (-1+0.05,-1) -- (0,-1) -- ({4*sin(-180)}, {4*cos(-180)}) plot ({4*sin(\x)}, {4*cos(\x)});
    
    \fill[pattern={crosshatch dots},pattern color=black] (-1+0.05,-1) -- (0,-1) -- (0,0) -- (-1,0);
    
    \fill[blue,opacity=0.4,domain=-90:-150,variable=\x,smooth] ({4*sin(-150)},{4*cos(-150)}) to [out=50,in=-90] (-1,0) -- (-4,0) plot ({4*sin(\x)}, {4*cos(\x)});
    
    \end{tikzpicture}
    }
    \hfill
    \subfloat[\label{fig:withmatter}]{
    \begin{tikzpicture}[scale=0.9]

    \fill[pattern={crosshatch dots},pattern color=black] (0,0) rectangle +(-2,5);
    \fill[pattern={crosshatch},pattern color=gray, opacity=0.5] (0,0) rectangle +(2,5);
    
    \draw[gray,domain=-30:210,<->] plot ({cos(\x)},{sin(\x)+5});
    
    \draw[thick,ForestGreen, mid arrow] (0,0) -- (0,5);
    \draw[thick, fill=green, ForestGreen] (0,5) circle (0.1);
    
    \end{tikzpicture}
    }
    \caption{(a) The first portion of the dividing curve $\alpha_j$ starts at the right endpoint of $\mathcal{V}_{j},$ and follows the HRRT surface $\gamma_{\mathcal{V}}$ until it crosses into $I^-(E_{\mathcal{W}_j}).$ At this point, the dividing curve starts to follow $\mathscr{P}_j \cap \mathscr{F}.$ Points with dot-like shading are in $I^-(E_{\mathcal{W}_j}),$ and points with cross-like shading are in $J^-(E_{\mathcal{W}_j})^c.$ (b) If the dividing curve terminated anywhere in the interior of $\mathscr{F},$ then we could draw a curve from points in $I^-(E_{\mathcal{W}_j})$ to points in $J^-(E_{\mathcal{W}_j})^c$ without ever crossing $\mathscr{P}_j.$}
    \label{fig:dividing-curve-start}
\end{figure}

At this intersection, we turn into the interior of $\mathscr{F}$ by following the curve $\mathscr{P}_j \cap \mathscr{F}.$
This curve has the property that all points of $\mathscr{F}$ sufficiently close to the right-hand side are outside of $J^-(E_{\mathcal{W}_j}),$ and all points of $\mathscr{F}$ sufficiently close to the left-hand side are inside of $I^-(E_{\mathcal{W}_j}).$
See figure \ref{fig:dividing-curve-start}a.
This curve cannot terminate anywhere in the interior of $\mathscr{F},$ because then we would be able to draw paths that connect points in the interior of $J^-(E_{\mathcal{W}_j})$ to points outside of $J^-(E_{\mathcal{W}_j})$ without ever passing through $\del J^-(E_{\mathcal{W}_j})$, as in figure \ref{fig:dividing-curve-start}b.
This curve also cannot ever reach a spacetime singularity, since that would require the singularity to be in the future of $\gamma_{\mathcal{V}}$ but the past of $E_{\mathcal{W}_j},$ and AdS-hyperbolicity forbids singularities that are not strictly in the future or strictly in the past; see lemma \ref{lem:future-past-singularities}.
So the component of $\mathscr{P}_j \cap \mathscr{F}$ that we are following must end on the boundary of $\mathscr{F},$ either on a component of $\gamma_{\mathcal{V}}$ or on the spacetime boundary.
However, because we know that $J^-(E_{\mathcal{W}_j})$ contains $\hat{J}^-(\mathcal{W}_{j}),$ and $\hat{J}^-(\mathcal{W}_{j})$ intersects $\mathscr{F}$ only at the right-endpoint of $\mathcal{V}_{j}$ and the left-endpoint of $\mathcal{V}_{j+1}$, these are the only places where $\mathscr{P}_j \cap \mathscr{F}$ could intersect the boundary; since these points are in $\gamma_{\mathcal{V}}$, we can say without loss of generality that the curve we have been following terminates on a component of $\gamma_{\mathcal{V}}.$

Once our curve has hit the boundary of $\mathscr{F}$ by colliding with a component of $\gamma_{\mathcal{V}},$ one of three things can happen: (a) we have reached the left-endpoint of $\mathcal{V}_{j+1}$, (b) there is a unique direction along $\gamma_{\mathcal{V}}$ whose points are not contained in $I^-(E_{\mathcal{W}_j}),$ or (c) $\mathscr{P}_j$ is tangent to $\gamma_{\mathcal{V}}$ at this point, and ``bounces off''.
These are sketched in figure \ref{fig:dividing-curve-collisions}.
In case (c), we just keep following $\mathscr{P}_j$; in case (b), we go along $\gamma_{\mathcal{V}}$ in the direction that is not contained in $I^-(E_{\mathcal{W}_j})$; in case (a), we stop.

\begin{figure}
    \centering
    \subfloat[\label{fig:dividingcurvemeetspartialV}]{
    \begin{tikzpicture}
    
    \draw[thick,domain=-15:-90] plot ({4*cos(\x)},{4*sin(\x)});
    \draw[thick,blue] ({4*cos(-60)},{4*sin(-60)}) to [out=110,in=-120] ({2.2},{-1});
    \fill[blue,opacity=0.4,domain=-15:-60] ({4*cos(-60)},{4*sin(-60)}) to [out=110,in=-120] ({2.2},{-1}) -- ({4*cos(-15)},{4*sin(-15)}) plot ({4*cos(\x)},{4*sin(\x)});
    
    \draw[ForestGreen,thick] ({4*cos(-60)},{4*sin(-60)}) -- (1,-1);
    
    \fill[pattern={crosshatch},pattern color=gray, opacity=0.5,domain=-90:-60,variable=\x,smooth] ({4*cos(-60)},{4*sin(-60)}) -- (1,-1) -- (0,-1) -- ({4*cos(-90)},{4*sin(-90)}) plot ({4*cos(\x)},{4*sin(\x)});
    
    \fill[pattern={crosshatch dots},pattern color=black] ({4*cos(-60)},{4*sin(-60)}) -- (1,-1) -- ({2.2},{-1}) to [out=-120,in=110] ({2},{-3.464});
    
    \node[below right]  at ({4*cos(-30)},{4*sin(-30)}) {$\mathcal{V}_{j+1}$};
    
    \end{tikzpicture}
    }
    \hfill
    \subfloat[\label{fig:dividingcurvemeetsgammaV}]{
    \begin{tikzpicture}[scale=0.8]
    
    \fill[blue,opacity=0.4,domain=180:90] plot ({4*cos(\x)},{4*sin(\x)}) -- (-1,1);
    \draw[thick,blue,domain=180:90] plot ({4*cos(\x)},{4*sin(\x)});
    
    \draw[ForestGreen,domain=135:180,thick] ({4*cos(135)},{4*sin(135)}) -- ({4.05*cos(135)-1},{4.05*sin(135)+1}) plot ({4.05*cos(\x)},{4.05*sin(\x)});
    
    \fill[pattern={crosshatch},pattern color=gray, opacity=0.5,domain=180:135,variable=\x,smooth] plot ({4.05*cos(\x)},{4.05*sin(\x)}) -- ({4.05*cos(135)-1},{4.05*sin(135)+1}) -- ({4.05*cos(135)-2},{4.05*sin(135)});
    
    \fill[pattern={crosshatch dots},pattern color=black, opacity=0.5,domain=135:90,variable=\x,smooth] plot ({4.05*cos(\x)},{4.05*sin(\x)}) -- ({4.05*cos(135)},{4.05*sin(135)+2}) -- ({4.05*cos(135)-1},{4.05*sin(135)+1});
    
    \end{tikzpicture}
    }
    \hfill
    \subfloat[\label{fig:dividingcurvegrazesgammaV}]{
    \begin{tikzpicture}[scale=0.8]
    
    \fill[blue,opacity=0.4,domain=180:90] plot ({4*cos(\x)},{4*sin(\x)}) -- (-1,1);
    \draw[thick,blue,domain=180:90] plot ({4*cos(\x)},{4*sin(\x)});
    
    \draw[ForestGreen,thick] ({-2.82+1},{4*sin(135)+2}) to [out=-135,in=45] ({-2.82},{4*sin(135)}) to [out=-135,in=45] ({-2.82-2},{4*sin(135)-1});
    
    \fill[pattern={crosshatch dots},pattern color=black, opacity=0.5,domain=135:90,variable=\x,smooth] plot ({4*cos(\x)},{4*sin(\x)}) to [out=45,in=135] ({-2.82+1},{4*sin(135)+2});
    \fill[pattern={crosshatch dots},pattern color=black, opacity=0.5,domain=135:180,variable=\x,smooth] plot ({4*cos(\x)},{4*sin(\x)}) to [out=-135,in=45] ({-2.82-2},{4*sin(135)-1});
    
    \fill[pattern={crosshatch},pattern color=gray, opacity=0.5,domain=180:135,variable=\x,smooth] ({-2.82+1},{4*sin(135)+2}) to [out=-135,in=45] ({-2.82},{4*sin(135)}) to [out=-135,in=45] ({-2.82-2},{4*sin(135)-1}) -- (-3.5,3.5);
    
    \end{tikzpicture}
    }
    \caption{The three ways described in the text that a dividing curve can meet a component of $\gamma_{\mathcal{V}}.$ In (a), it meets $\gamma_{\mathcal{V}}$ at the left endpoint of $\mathcal{V}_{j}$ and terminates. In (b) there is a direction along $\gamma_{\mathcal{V}}$ that lies outside $I^-(E_{\mathcal{W}_j}),$ and the dividing curve continues in that direction. In (c) the dividing curve bounces off of $\gamma_{\mathcal{V}}$ and continues in the interior of $\mathscr{F}.$}
    \label{fig:dividing-curve-collisions}
\end{figure}
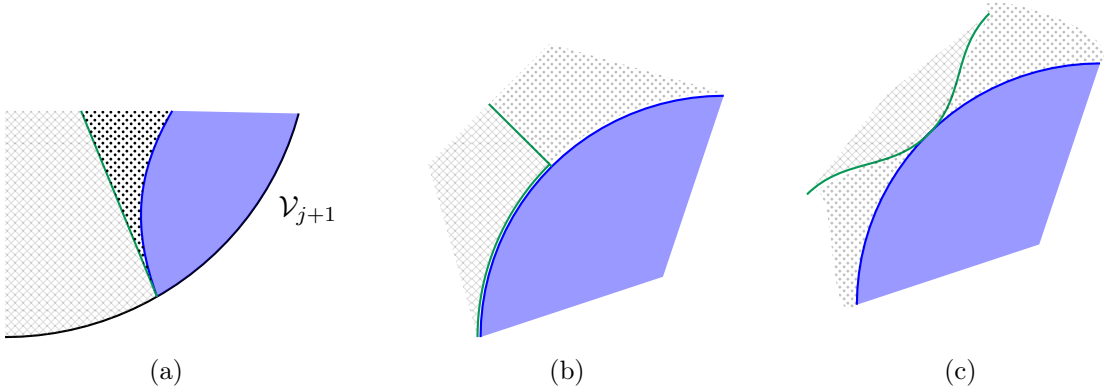

If case (c) of the preceding paragraph was triggered, then we know by the same logic given above that our curve must eventually collide with a component of $\gamma_{\mathcal{V}}$; we therefore repeat the algorithm of the preceding paragraph.
If case (b) was triggered, then we are again in the situation we started with: we are following a component of $\gamma_{\mathcal{V}}$ that does not lie in $I^-(E_{\mathcal{W}_j}),$ and we must eventually either enter $I^-(E_{\mathcal{W}_j})$ or reach the left endpoint of $\mathcal{V}_{j+1}.$
If we enter $I^-(E_{\mathcal{W}_j}),$ we must do so by crossing $\mathscr{P}_j,$ and so we can turn toward the interior of $\mathscr{F}$ along $\mathscr{P}_j$ and repeat the rules described above.

The end result of applying this algorithm will be a curve in $\mathscr{F}$ that connects the right endpoint of $\mathcal{V}_{j}$ to the left endpoint of $\mathcal{V}_{j+1}$ such that every point is either in $\gamma_{\mathcal{V}} \cap I^-(E_{\mathcal{W}_j})^c,$ or in $\del J^-(E_{\mathcal{W}_j}) \cap \gamma_{\mathcal{V}}^c.$
This is the dividing curve $\alpha_j.$
Points in $\mathscr{F}$ that lie just to the left of this curve are in $I^-(E_{\mathcal{W}_j}),$ and points in $\mathscr{F}$ that lie just to the right are in $J^-(E_{\mathcal{W}_j})^c.$

The $n$ dividing curves $\alpha_1, \dots, \alpha_n$ will make up part of the cross-section of $\mathscr{F}$ we are trying to construct.
We would like to prove a topological property of these curves, which is that they do not intersect.
This is where we will use the condition that the $2$-to-all causal graph $\Gamma_{2\rightarrow\text{all}}$ is connected.
Since this is the essential piece of the proof, we present it in the following subsection.

\subsubsection*{Causal connections and lightsheet geometry}

We will now show that whenever $J(E_{\mathcal{V}_j}, E_{\mathcal{V}_k} \rightarrow E_{\mathcal{W}_1}, \dots, E_{E_{\mathcal{W}_n}})$ has nonempty interior, there exists an open family of curves in $\mathscr{F} \cup (E_{\mathcal{V}} \cap \Sigma)$ such that (i) each curve connects the component of $E_{\mathcal{V}}$ containing $\mathcal{V}_{j}$ to the component of $E_{\mathcal{V}}$ containing $\mathcal{V}_k,$ and (ii) each curve lies entirely in the past of all output wedges $E_{\mathcal{W}_j}$. 

For example, in the case $n=3$, the existence of a nonempty interior for $J(E_{\mathcal{V}_1}, E_{\mathcal{V}_3} \rightarrow E_{\mathcal{W}_1}, E_{\mathcal{W}_2}, E_{\mathcal{W}_3)}$ implies either an open family of curves connecting $E_{\mathcal{V}_1}$ to $E_{\mathcal{V}_3}$ across $\mathscr{F}$, or connecting $E_{\mathcal{V}_1}$ to $E_{\mathcal{V}_3}$ by way of $E_{\mathcal{V}_2}$, or both; this is sketched in figure \ref{fig:open-curve-families}.

\begin{figure}
    \centering
    \subfloat[\label{fig:directopencurvefamily}]{
    \begin{tikzpicture}

    \draw[domain=-120:-60,variable=\x,smooth, fill=blue,opacity=0.4] plot ({3*sin(\x)}, {3*cos(\x)}) -- (-3*0.866,3*0.5) to [out=-30,in=30] (-3*0.866,-3*0.5);
    
    \begin{scope}[rotate=120]
    \draw[domain=-120:-60,variable=\x,smooth, fill=blue,opacity=0.4] plot ({3*sin(\x)}, {3*cos(\x)}) -- (-3*0.866,3*0.5) to [out=-30,in=30] (-3*0.866,-3*0.5);
    \end{scope}
    
    \begin{scope}[rotate=240]
    \draw[domain=-120:-60,variable=\x,smooth, fill=blue,opacity=0.4] plot ({3*sin(\x)}, {3*cos(\x)}) -- (-3*0.866,3*0.5) to [out=-30,in=30] (-3*0.866,-3*0.5);
    \end{scope}
    
    \draw[thick] (0,0) circle (3);
    
    \draw[thick,blue] (-3*0.866,3*0.5) to [out=-30,in=30] (-3*0.866,-3*0.5);
    \draw[thick,blue] (0,3) to [out=-90,in=-150] (3*0.866,3*0.5);
    \draw[thick,blue] (0,-3) to [out=90,in=150] (3*0.866,-3*0.5);
    
    \node[left] at (-3,0) {$\mathcal{V}_1$};
    \node[above right] at (3*0.5,3*0.866) {$\mathcal{V}_3$};
    \node[below right] at ({3*cos(-70)},{3*sin(-70)}) {$\mathcal{V}_2$};
    
    \draw[red,thick,dashed] ({3*cos(30)},{3*sin(30)}) -- (3.5,0);
    \draw[red,thick,dashed] ({3*cos(30)},{3*sin(-30)}) -- (3.5,0);
    \draw[red,decorate, decoration={snake, segment length=3mm, amplitude=0.5mm}] (0,0) -- (3.5,0);
    
    \draw[red,thick,dashed] ({3*cos(150)},{3*sin(150)}) -- ({3.5*cos(120)},{3.5*sin(120)});
    \draw[red,thick,dashed] ({3*cos(90)},{3*sin(90)}) -- ({3.5*cos(120)},{3.5*sin(120)});
    \draw[red,decorate, decoration={snake, segment length=3mm, amplitude=0.5mm}] (0,0) -- ({3.5*cos(120)},{3.5*sin(120)});
    
    \draw[red,thick,dashed] ({3*cos(-150)},{3*sin(-150)}) -- ({3.5*cos(-120)},{3.5*sin(-120)});
    \draw[red,thick,dashed] ({3*cos(-90)},{3*sin(-90)}) -- ({3.5*cos(-120)},{3.5*sin(-120)});
    \draw[red,decorate, decoration={snake, segment length=3mm, amplitude=0.5mm}] (0,0) -- ({3.5*cos(-120)},{3.5*sin(-120)});
    
    \fill[red,opacity=0.3] (0,0) -- (3.5,0) -- ({3*cos(30)},{3*sin(-30)}) -- ({3*cos(-30)},{3*sin(-30)}) to [out=150,in=90] (0,{3*sin(-90)}) -- ({3.5*cos(-120)},{3.5*sin(-120)});
    
    \begin{scope}[rotate=120]
    \fill[red,opacity=0.3] (0,0) -- (3.5,0) -- ({3*cos(30)},{3*sin(-30)}) -- ({3*cos(-30)},{3*sin(-30)}) to [out=150,in=90] (0,{3*sin(-90)}) -- ({3.5*cos(-120)},{3.5*sin(-120)});
    \end{scope}
    
    \begin{scope}[rotate=-120]
    \fill[red,opacity=0.3] (0,0) -- (3.5,0) -- ({3*cos(30)},{3*sin(-30)}) -- ({3*cos(-30)},{3*sin(-30)}) to [out=150,in=90] (0,{3*sin(-90)}) -- ({3.5*cos(-120)},{3.5*sin(-120)});
    \end{scope}

    \draw plot [mark=*, mark size=2] coordinates{(0,0)};
    
    \draw[red,decorate, decoration={snake, segment length=3mm, amplitude=0.5mm}] ({2*cos(120)},{2*sin(120)}) -- ({2*cos(120)+0.5},{2*sin(120)+0.5});
    \draw[red,decorate, decoration={snake, segment length=3mm, amplitude=0.5mm}] (1.5,0) -- ({1.5},{-0.75});
    
    \foreach \i in {1,...,30}
    {
    \draw[magenta] ({-3.8+2*cos(\i)},{2*sin(\i)}) -- ({1.8*cos(-\i-100)+1.21*1.5},{1.8*sin(-\i-100)+1.21*2.598});
    }
    
\end{tikzpicture}
    }
    \hfill
    \subfloat[\label{fig:indirectopencruvefamily}]{
    \begin{tikzpicture}

    \draw[domain=-120:-60,variable=\x,smooth, fill=blue,opacity=0.4] plot ({3*sin(\x)}, {3*cos(\x)}) -- (-3*0.866,3*0.5) to [out=-30,in=30] (-3*0.866,-3*0.5);
    
    \begin{scope}[rotate=120]
    \draw[domain=-120:-60,variable=\x,smooth, fill=blue,opacity=0.4] plot ({3*sin(\x)}, {3*cos(\x)}) -- (-3*0.866,3*0.5) to [out=-30,in=30] (-3*0.866,-3*0.5);
    \end{scope}
    
    \begin{scope}[rotate=240]
    \draw[domain=-120:-60,variable=\x,smooth, fill=blue,opacity=0.4] plot ({3*sin(\x)}, {3*cos(\x)}) -- (-3*0.866,3*0.5) to [out=-30,in=30] (-3*0.866,-3*0.5);
    \end{scope}
    
    \draw[thick] (0,0) circle (3);
    
    \draw[thick,blue] (-3*0.866,3*0.5) to [out=-30,in=30] (-3*0.866,-3*0.5);
    \draw[thick,blue] (0,3) to [out=-90,in=-150] (3*0.866,3*0.5);
    \draw[thick,blue] (0,-3) to [out=90,in=150] (3*0.866,-3*0.5);
    
    \node[left] at (-3,0) {$\mathcal{V}_1$};
    \node[above right] at (3*0.5,3*0.866) {$\mathcal{V}_3$};
    \node[below right] at ({3*cos(-70)},{3*sin(-70)}) {$\mathcal{V}_2$};
    
    \draw[red,thick,dashed] ({3*cos(30)},{3*sin(30)}) -- (3.5,0);
    \draw[red,thick,dashed] ({3*cos(30)},{3*sin(-30)}) -- (3.5,0);
    \draw[red,decorate, decoration={snake, segment length=3mm, amplitude=0.5mm}] (0,0) -- (3.5,0);
    
    \draw[red,thick,dashed] ({3*cos(150)},{3*sin(150)}) -- ({3.5*cos(120)},{3.5*sin(120)});
    \draw[red,thick,dashed] ({3*cos(90)},{3*sin(90)}) -- ({3.5*cos(120)},{3.5*sin(120)});
    \draw[red,decorate, decoration={snake, segment length=3mm, amplitude=0.5mm}] (0,0) -- ({3.5*cos(120)},{3.5*sin(120)});
    
    \draw[red,thick,dashed] ({3*cos(-150)},{3*sin(-150)}) -- ({3.5*cos(-120)},{3.5*sin(-120)});
    \draw[red,thick,dashed] ({3*cos(-90)},{3*sin(-90)}) -- ({3.5*cos(-120)},{3.5*sin(-120)});
    \draw[red,decorate, decoration={snake, segment length=3mm, amplitude=0.5mm}] (0,0) -- ({3.5*cos(-120)},{3.5*sin(-120)});
    
    \fill[red,opacity=0.3] (0,0) -- (3.5,0) -- ({3*cos(30)},{3*sin(-30)}) -- ({3*cos(-30)},{3*sin(-30)}) to [out=150,in=90] (0,{3*sin(-90)}) -- ({3.5*cos(-120)},{3.5*sin(-120)});
    
    \begin{scope}[rotate=120]
    \fill[red,opacity=0.3] (0,0) -- (3.5,0) -- ({3*cos(30)},{3*sin(-30)}) -- ({3*cos(-30)},{3*sin(-30)}) to [out=150,in=90] (0,{3*sin(-90)}) -- ({3.5*cos(-120)},{3.5*sin(-120)});
    \end{scope}
    
    \begin{scope}[rotate=-120]
    \fill[red,opacity=0.3] (0,0) -- (3.5,0) -- ({3*cos(30)},{3*sin(-30)}) -- ({3*cos(-30)},{3*sin(-30)}) to [out=150,in=90] (0,{3*sin(-90)}) -- ({3.5*cos(-120)},{3.5*sin(-120)});
    \end{scope}

    \draw plot [mark=*, mark size=2] coordinates{(0,0)};
    
     \draw[red,decorate, decoration={snake, segment length=3mm, amplitude=0.5mm}] ({2*cos(120)},{2*sin(120)}) -- ({2*cos(120)+0.5},{2*sin(120)+0.5});
    \draw[red,decorate, decoration={snake, segment length=3mm, amplitude=0.5mm}] (1.5,0) -- ({1.5},{-0.75});
    
    \foreach \i in {1,...,30}
    {
    \draw[magenta] plot [smooth] coordinates {({-3.81+2*cos(\i)},{2*sin(\i)}) ({-\i/60+2.9*cos(-60)},{\i/50+2.9*sin(-60)}) ({1.8*cos(-\i-100)+1.21*1.5},{1.8*sin(-\i-100)+1.21*2.598})};
    }
    
\end{tikzpicture}
    }
    \caption{As explained in the text, a nonempty interior for $J(E_{\mathcal{V}_1}, E_{\mathcal{V}_3} \rightarrow E_{\mathcal{W}_1}, E_{\mathcal{W}_2}, E_{\mathcal{W}_3)}$ implies that there is an open family of curves contained in $J^-(E_{\mathcal{W}_1}) \cap J^-(E_{\mathcal{W}_2}) \cap J^-(E_{\mathcal{W}_3}) \cap (\mathscr{F} \cup (E_{\mathcal{V}} \cap \Sigma))$ that connects $E_{\mathcal{V}_1}$ to $E_{\mathcal{V}_3}$. The curves may be contained entirely in $\mathscr{F},$ as in (a), or they may pass through $E_{\mathcal{V}_2},$ as in (b).}
    \label{fig:open-curve-families}
\end{figure}
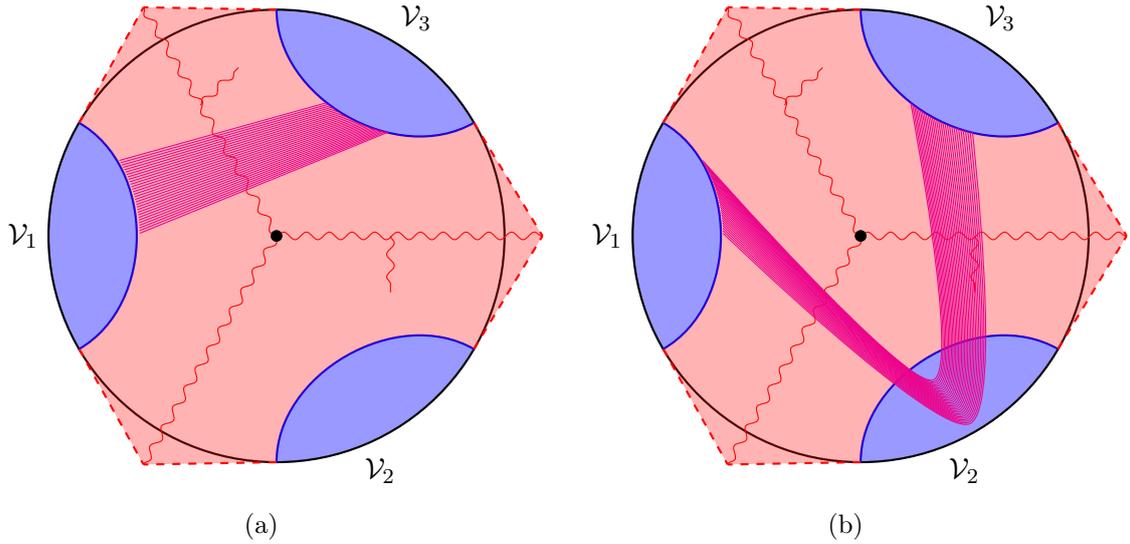

We can construct this open family of curves explicitly.
Let $A$ and $B$ be the components of $E_{\mathcal{V}}$ whose boundaries contain $\mathcal{V}_{j}$ and $\mathcal{V}_k$, respectively.
By entanglement wedge nesting \cite{maximin}, we have $A \supseteq {E_{\mathcal{V}_j}}$ and $B \supseteq {E_{\mathcal{V}_k}}.$
So a nonempty interior for $J(E_{\mathcal{V}_j}, E_{\mathcal{V}_k} \rightarrow E_{\mathcal{W}_1}, \dots, {E_{\mathcal{W}_n}})$ implies a nonempty interior for $J(A,B \rightarrow E_{\mathcal{W}_1}, \dots, E_{\mathcal{W}_n}).$

Let $p$ be a point in $J(A,B \rightarrow E_{\mathcal{W}_1}, \dots, E_{\mathcal{W}_n}).$
If we follow any inextendible timelike curve from $p$ into the past, we will eventually hit either $\del J^+(A \cap \Sigma)$ or $\del J^+(A \cap \Sigma).$
Without loss of generality, assume we hit $\del J^+(A \cap \Sigma).$ The point of intersection lies on some unique generator of $\del J^+(A \cap \Sigma),$ and we can follow that generator into its past until we hit $\del J^+(B \cap \Sigma).$
This gives us a point $q$ that lies on the ``seam'' $\del J^+(A \cap \Sigma) \cap \del J^+(B \cap \Sigma)$, and that is in the past of $p$; since $p$ was in the past of all output wedges $E_{\mathcal{W}_j},$ this property also holds for $q.$

The point $q$ lies on a generator of $\del J^+(A \cap \Sigma)$ and on a generator of $\del J^+(B \cap \Sigma).$
By the same arguments as given in lemma \ref{lem:lightsheet-generators}, the past endpoints of these generators must lie on the spatial boundaries of the corresponding sets, i.e., must lie on the HRRT surface $\gamma_{\mathcal{V}}.$
Every point between $\gamma_{\mathcal{V}}$ and $q$ on one of these generators is in the past of $q,$ and therefore in the past of all output wedges $E_{\mathcal{W}_j}.$
The union of all these points is a curve connecting $A$ to $B,$ contained entirely in the past of all output wedges, for which the points of the curve near $A$ and $B$ lie in $\mathscr{F}.$
However, we are not guaranteed that the entire curve lies in $\mathscr{F}$; for example, the point $q$ need not be in $\mathscr{F}$, as it could be in the timelike future of another entanglement wedge component $C.$
The two possibilities --- $q \in \mathscr{F}$ or $q \in I^+(C \cap \Sigma)$ --- are sketched in figure \ref{fig:possible-q-locations}.\footnote{These are the only two possibilities, thanks to $\mathscr{F} \subseteq J^+(E_{\mathcal{V}} \cap \Sigma) - I^+(E_{\mathcal{V}} \cap \Sigma).$ See lemma \ref{lem:lightsheet-boundary}.}

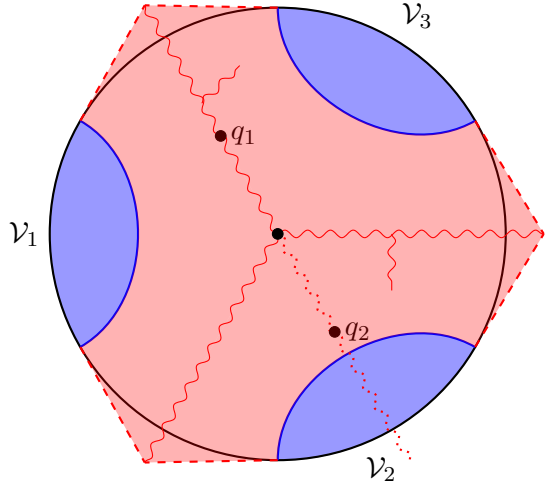
\begin{figure}
    \centering
    \begin{tikzpicture}

    \draw[domain=-120:-60,variable=\x,smooth, fill=blue,opacity=0.4] plot ({3*sin(\x)}, {3*cos(\x)}) -- (-3*0.866,3*0.5) to [out=-30,in=30] (-3*0.866,-3*0.5);
    
    \begin{scope}[rotate=120]
    \draw[domain=-120:-60,variable=\x,smooth, fill=blue,opacity=0.4] plot ({3*sin(\x)}, {3*cos(\x)}) -- (-3*0.866,3*0.5) to [out=-30,in=30] (-3*0.866,-3*0.5);
    \end{scope}
    
    \begin{scope}[rotate=240]
    \draw[domain=-120:-60,variable=\x,smooth, fill=blue,opacity=0.4] plot ({3*sin(\x)}, {3*cos(\x)}) -- (-3*0.866,3*0.5) to [out=-30,in=30] (-3*0.866,-3*0.5);
    \end{scope}
    
    \draw[thick] (0,0) circle (3);
    
    \draw[thick,blue] (-3*0.866,3*0.5) to [out=-30,in=30] (-3*0.866,-3*0.5);
    \draw[thick,blue] (0,3) to [out=-90,in=-150] (3*0.866,3*0.5);
    \draw[thick,blue] (0,-3) to [out=90,in=150] (3*0.866,-3*0.5);
    
    \node[left] at (-3,0) {$\mathcal{V}_1$};
    \node[above right] at (3*0.5,3*0.866) {$\mathcal{V}_3$};
    \node[below right] at ({3*cos(-70)},{3*sin(-70)}) {$\mathcal{V}_2$};
    
    \draw[red,thick,dashed] ({3*cos(30)},{3*sin(30)}) -- (3.5,0);
    \draw[red,thick,dashed] ({3*cos(30)},{3*sin(-30)}) -- (3.5,0);
    \draw[red,decorate, decoration={snake, segment length=3mm, amplitude=0.5mm}] (0,0) -- (3.5,0);
    
    \draw[red,thick,dashed] ({3*cos(150)},{3*sin(150)}) -- ({3.5*cos(120)},{3.5*sin(120)});
    \draw[red,thick,dashed] ({3*cos(90)},{3*sin(90)}) -- ({3.5*cos(120)},{3.5*sin(120)});
    \draw[red,decorate, decoration={snake, segment length=3mm, amplitude=0.5mm}] (0,0) -- ({3.5*cos(120)},{3.5*sin(120)});
    
    \draw[red,thick,dashed] ({3*cos(-150)},{3*sin(-150)}) -- ({3.5*cos(-120)},{3.5*sin(-120)});
    \draw[red,thick,dashed] ({3*cos(-90)},{3*sin(-90)}) -- ({3.5*cos(-120)},{3.5*sin(-120)});
    \draw[red,decorate, decoration={snake, segment length=3mm, amplitude=0.5mm}] (0,0) -- ({3.5*cos(-120)},{3.5*sin(-120)});
    
    \draw[red,decorate, decoration={snake, segment length=3mm, amplitude=0.5mm}] (1.5,0) -- ({1.5},{-0.75});
    
    \draw[red,decorate, decoration={snake, segment length=3mm, amplitude=0.5mm}] ({2*cos(120)},{2*sin(120)}) -- ({2*cos(120)+0.5},{2*sin(120)+0.5});
    
    \draw[thick,dotted,red,decorate, decoration={snake, segment length=3mm, amplitude=0.5mm}] (0,0) -- ({3.5*cos(-60)},{3.5*sin(-60)});
    
    \draw plot [mark=*, mark size=2] coordinates{({1.5*cos(-60)},{1.5*sin(-60)})};
    \node[right] at ({1.5*cos(-60)},{1.5*sin(-60)}) {$q_2$};
    
    \draw plot [mark=*, mark size=2] coordinates{({1.5*cos(120)},{1.5*sin(120)})};
    \node[right] at ({1.5*cos(120)},{1.5*sin(120)}) {$q_1$};
    
    \fill[red,opacity=0.3] (0,0) -- (3.5,0) -- ({3*cos(30)},{3*sin(-30)}) -- ({3*cos(-30)},{3*sin(-30)}) to [out=150,in=90] (0,{3*sin(-90)}) -- ({3.5*cos(-120)},{3.5*sin(-120)});
    
    \begin{scope}[rotate=120]
    \fill[red,opacity=0.3] (0,0) -- (3.5,0) -- ({3*cos(30)},{3*sin(-30)}) -- ({3*cos(-30)},{3*sin(-30)}) to [out=150,in=90] (0,{3*sin(-90)}) -- ({3.5*cos(-120)},{3.5*sin(-120)});
    \end{scope}
    
    \begin{scope}[rotate=-120]
    \fill[red,opacity=0.3] (0,0) -- (3.5,0) -- ({3*cos(30)},{3*sin(-30)}) -- ({3*cos(-30)},{3*sin(-30)}) to [out=150,in=90] (0,{3*sin(-90)}) -- ({3.5*cos(-120)},{3.5*sin(-120)});
    \end{scope}

    \draw plot [mark=*, mark size=2] coordinates{(0,0)};
    
\end{tikzpicture}
    \caption{A point in $\del J^+(E_{\mathcal{V}_1} \cap \Sigma) \cap \del J^+(E_{\mathcal{V}_2} \cap \Sigma)$ can either lie on $\mathscr{F}$ --- like the point $q_1$ in this figure --- or in the timelike future of $\mathcal{V}_3 \cap \Sigma$ --- like the point $q_2$ in this figure.}
    \label{fig:possible-q-locations}
\end{figure}

If $q$ is in $\mathscr{F}$, then the union of the two generators lying between $q$ and $\gamma_{\mathcal{V}}$ forms a curve with the properties we want --- it connects $A$ to $B$ through $\mathscr{F}$ and lies entirely in the past of all output wedges.
If $q$ is not in $\mathscr{F}$, then we deform the union of the two generators into the past along any family of timelike geodesics.
Each of these timelike geodesics, being past-inextendible, eventually hits the surface $\del J^+(E_{\mathcal{V}} \cap \Sigma) = (E_{\mathcal{V}} \cap \Sigma) \cup \mathscr{F}$.\footnote{The curves start inside $J^+(E_{\mathcal{V}} \cap \Sigma),$ but this set has finite past extent, so any past-inextendible timelike curve must eventually leave it.}
So the curve can be deformed into the past so that it lies entirely in $\mathscr{F} \cup (E_{\mathcal{V}} \cap \Sigma)$; since the deformation was past-directed, the resulting curve also lies entirely in $J^-(E_{\mathcal{W}_1}) \cap \dots \cap J^-(E_{\mathcal{W}_n}).$

The summary so far: for each $p \in J(A, B \rightarrow E_{\mathcal{W}_1}, \dots, E_{\mathcal{W}_n}),$ we have constructed a curve from $A$ to $B$ that lies entirely in $\mathscr{F} \cup (E_{\mathcal{V}} \cap \Sigma)$ and in $J^-(E_{\mathcal{W}_1}) \cap \dots \cap J^-(E_{\mathcal{W}_n}).$
Since $J(A, B \rightarrow E_{\mathcal{W}_1}, \dots, E_{\mathcal{W}_n})$ has nonempty interior, there exists an open family of such curves.

We will apply this fact several times in completing the proof; as a first corollary, we will show that $\Gamma_{2\rightarrow \text{all}}$ connected implies that no two dividing curves (cf. the previous subsection) intersect.
Recall that the points just to the right of a dividing curve $\alpha_j$ lie in $J^-(E_{\mathcal{W}_j})^c$; so by deforming $\alpha_j$ slightly to its right, we can construct a curve from the right endpoint of $\mathcal{V}_{j}$ to the left endpoint of $\mathcal{V}_{j+1}$ that lies entirely in $J^-(E_{\mathcal{W}_j})^c,$ except for at its endpoints.
See figure \ref{fig:blocking-curves}a.

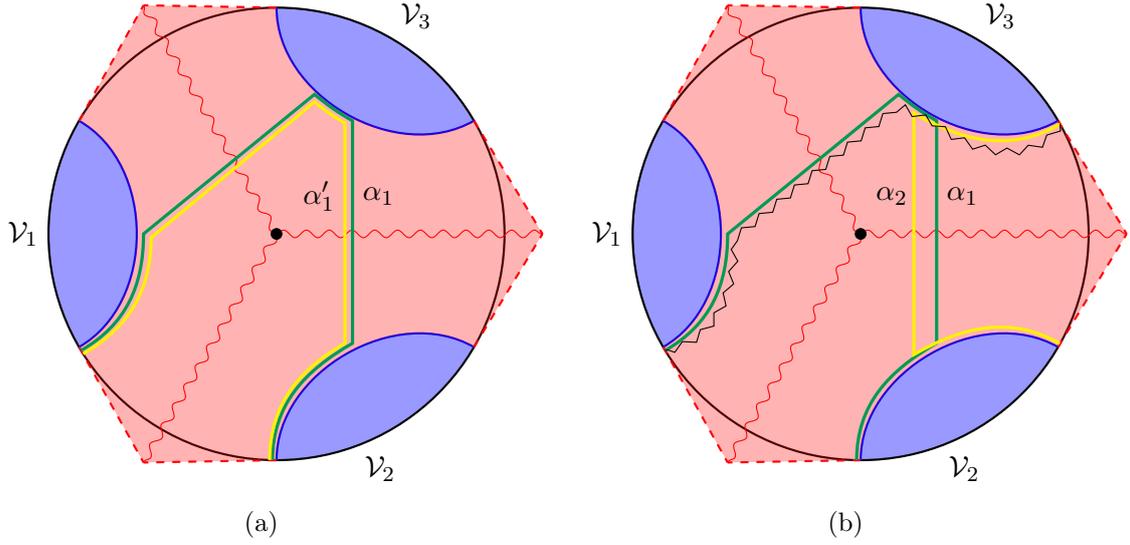
\begin{figure}
    \centering
    \subfloat[\label{fig:dividingcurveanddeformation}]{
    \begin{tikzpicture}

    \draw[domain=-120:-60,variable=\x,smooth, fill=blue,opacity=0.4] plot ({3*sin(\x)}, {3*cos(\x)}) -- (-3*0.866,3*0.5) to [out=-30,in=30] (-3*0.866,-3*0.5);
    
    \begin{scope}[rotate=120]
    \draw[domain=-120:-60,variable=\x,smooth, fill=blue,opacity=0.4] plot ({3*sin(\x)}, {3*cos(\x)}) -- (-3*0.866,3*0.5) to [out=-30,in=30] (-3*0.866,-3*0.5);
    \end{scope}
    
    \begin{scope}[rotate=240]
    \draw[domain=-120:-60,variable=\x,smooth, fill=blue,opacity=0.4] plot ({3*sin(\x)}, {3*cos(\x)}) -- (-3*0.866,3*0.5) to [out=-30,in=30] (-3*0.866,-3*0.5);
    \end{scope}
    
    \draw[thick] (0,0) circle (3);
    
    \draw[thick,blue] (-3*0.866,3*0.5) to [out=-30,in=30] (-3*0.866,-3*0.5);
    \draw[thick,blue] (0,3) to [out=-90,in=-150] (3*0.866,3*0.5);
    \draw[thick,blue] (0,-3) to [out=90,in=150] (3*0.866,-3*0.5);
    
    \node[left] at (-3,0) {$\mathcal{V}_1$};
    \node[above right] at (3*0.5,3*0.866) {$\mathcal{V}_3$};
    \node[below right] at ({3*cos(-70)},{3*sin(-70)}) {$\mathcal{V}_2$};
    
    \draw[red,thick,dashed] ({3*cos(30)},{3*sin(30)}) -- (3.5,0);
    \draw[red,thick,dashed] ({3*cos(30)},{3*sin(-30)}) -- (3.5,0);
    \draw[red,decorate, decoration={snake, segment length=3mm, amplitude=0.5mm}] (0,0) -- (3.5,0);
    
    \draw[red,thick,dashed] ({3*cos(150)},{3*sin(150)}) -- ({3.5*cos(120)},{3.5*sin(120)});
    \draw[red,thick,dashed] ({3*cos(90)},{3*sin(90)}) -- ({3.5*cos(120)},{3.5*sin(120)});
    \draw[red,decorate, decoration={snake, segment length=3mm, amplitude=0.5mm}] (0,0) -- ({3.5*cos(120)},{3.5*sin(120)});
    
    \draw[red,thick,dashed] ({3*cos(-150)},{3*sin(-150)}) -- ({3.5*cos(-120)},{3.5*sin(-120)});
    \draw[red,thick,dashed] ({3*cos(-90)},{3*sin(-90)}) -- ({3.5*cos(-120)},{3.5*sin(-120)});
    \draw[red,decorate, decoration={snake, segment length=3mm, amplitude=0.5mm}] (0,0) -- ({3.5*cos(-120)},{3.5*sin(-120)});
    
    \fill[red,opacity=0.3] (0,0) -- (3.5,0) -- ({3*cos(30)},{3*sin(-30)}) -- ({3*cos(-30)},{3*sin(-30)}) to [out=150,in=90] (0,{3*sin(-90)}) -- ({3.5*cos(-120)},{3.5*sin(-120)});
    
    \begin{scope}[rotate=120]
    \fill[red,opacity=0.3] (0,0) -- (3.5,0) -- ({3*cos(30)},{3*sin(-30)}) -- ({3*cos(-30)},{3*sin(-30)}) to [out=150,in=90] (0,{3*sin(-90)}) -- ({3.5*cos(-120)},{3.5*sin(-120)});
    \end{scope}
    
    \begin{scope}[rotate=-120]
    \fill[red,opacity=0.3] (0,0) -- (3.5,0) -- ({3*cos(30)},{3*sin(-30)}) -- ({3*cos(-30)},{3*sin(-30)}) to [out=150,in=90] (0,{3*sin(-90)}) -- ({3.5*cos(-120)},{3.5*sin(-120)});
    \end{scope}

    \draw plot [mark=*, mark size=2] coordinates{(0,0)};
    
    \draw[ForestGreen,very thick] ({3*cos(-149)},{3*sin(-149)}) to [out=30,in=-90] (-1.75,0) -- (0.5,1.85) to [out=-40,in=150] (1,1.5) -- (1,-1.45) to [out=-150,in=90] ({-0.052},{-3*sin(91)}); 
    
    \draw[yellow,very thick] ({3*cos(-148)},{3*sin(-148)}) to [out=30,in=-90] (-1.65,-0.05) -- (0.5,1.75) to [out=-40,in=150] (0.9,1.47) -- (0.9,-1.45) to [out=-150,in=90] ({-0.09},{-3*sin(92)}); 
    
    \node[right] at (1,0.5) {$\alpha_1$};
    
    \node[left] at (0.9,0.5) {$\alpha_1'$};

\end{tikzpicture}
    }
    \hfill
    \subfloat[\label{fig:overlappingdividingcurves}]{
    \begin{tikzpicture}

    \draw[domain=-120:-60,variable=\x,smooth, fill=blue,opacity=0.4] plot ({3*sin(\x)}, {3*cos(\x)}) -- (-3*0.866,3*0.5) to [out=-30,in=30] (-3*0.866,-3*0.5);
    
    \begin{scope}[rotate=120]
    \draw[domain=-120:-60,variable=\x,smooth, fill=blue,opacity=0.4] plot ({3*sin(\x)}, {3*cos(\x)}) -- (-3*0.866,3*0.5) to [out=-30,in=30] (-3*0.866,-3*0.5);
    \end{scope}
    
    \begin{scope}[rotate=240]
    \draw[domain=-120:-60,variable=\x,smooth, fill=blue,opacity=0.4] plot ({3*sin(\x)}, {3*cos(\x)}) -- (-3*0.866,3*0.5) to [out=-30,in=30] (-3*0.866,-3*0.5);
    \end{scope}
    
    \draw[thick] (0,0) circle (3);
    
    \draw[thick,blue] (-3*0.866,3*0.5) to [out=-30,in=30] (-3*0.866,-3*0.5);
    \draw[thick,blue] (0,3) to [out=-90,in=-150] (3*0.866,3*0.5);
    \draw[thick,blue] (0,-3) to [out=90,in=150] (3*0.866,-3*0.5);
    
    \node[left] at (-3,0) {$\mathcal{V}_1$};
    \node[above right] at (3*0.5,3*0.866) {$\mathcal{V}_3$};
    \node[below right] at ({3*cos(-70)},{3*sin(-70)}) {$\mathcal{V}_2$};
    
    \draw[red,thick,dashed] ({3*cos(30)},{3*sin(30)}) -- (3.5,0);
    \draw[red,thick,dashed] ({3*cos(30)},{3*sin(-30)}) -- (3.5,0);
    \draw[red,decorate, decoration={snake, segment length=3mm, amplitude=0.5mm}] (0,0) -- (3.5,0);
    
    \draw[red,thick,dashed] ({3*cos(150)},{3*sin(150)}) -- ({3.5*cos(120)},{3.5*sin(120)});
    \draw[red,thick,dashed] ({3*cos(90)},{3*sin(90)}) -- ({3.5*cos(120)},{3.5*sin(120)});
    \draw[red,decorate, decoration={snake, segment length=3mm, amplitude=0.5mm}] (0,0) -- ({3.5*cos(120)},{3.5*sin(120)});
    
    \draw[red,thick,dashed] ({3*cos(-150)},{3*sin(-150)}) -- ({3.5*cos(-120)},{3.5*sin(-120)});
    \draw[red,thick,dashed] ({3*cos(-90)},{3*sin(-90)}) -- ({3.5*cos(-120)},{3.5*sin(-120)});
    \draw[red,decorate, decoration={snake, segment length=3mm, amplitude=0.5mm}] (0,0) -- ({3.5*cos(-120)},{3.5*sin(-120)});
    
    \fill[red,opacity=0.3] (0,0) -- (3.5,0) -- ({3*cos(30)},{3*sin(-30)}) -- ({3*cos(-30)},{3*sin(-30)}) to [out=150,in=90] (0,{3*sin(-90)}) -- ({3.5*cos(-120)},{3.5*sin(-120)});
    
    \begin{scope}[rotate=120]
    \fill[red,opacity=0.3] (0,0) -- (3.5,0) -- ({3*cos(30)},{3*sin(-30)}) -- ({3*cos(-30)},{3*sin(-30)}) to [out=150,in=90] (0,{3*sin(-90)}) -- ({3.5*cos(-120)},{3.5*sin(-120)});
    \end{scope}
    
    \begin{scope}[rotate=-120]
    \fill[red,opacity=0.3] (0,0) -- (3.5,0) -- ({3*cos(30)},{3*sin(-30)}) -- ({3*cos(-30)},{3*sin(-30)}) to [out=150,in=90] (0,{3*sin(-90)}) -- ({3.5*cos(-120)},{3.5*sin(-120)});
    \end{scope}

    \draw plot [mark=*, mark size=2] coordinates{(0,0)};
    
    \draw[ForestGreen,very thick] ({3*cos(-149)},{3*sin(-149)}) to [out=30,in=-90] (-1.75,0) -- (0.5,1.85) to [out=-40,in=150] (1,1.5) -- (1,-1.45) to [out=-150,in=90] ({-0.052},{-3*sin(91)}); 
    
    \draw[yellow, very thick] ({3*cos(29)},{3*sin(29)}) to [out=-150,in=-30] (0.7,1.6) --(0.7,-1.6) to [out=30,in=150] ({2.62},{3*sin(-29)});
    
    \draw[decorate,decoration={zigzag, segment length=3mm, amplitude=0.5mm}] ({3*cos(29)},{3*sin(29)}) to [out=-120,in=-40] (1,1.45) -- (0.5, 1.7) -- (-1.65,-0.05) to [out=-90,in=10] ({-2.57},{3*sin(-149)});
    
    \node[right] at (1,0.5) {$\alpha_1$};
    \node[left] at (0.75,0.5) {$\alpha_2$};
    
\end{tikzpicture}
    }
    
    \caption{(a) A dividing curve $\alpha_1$ (green), together with a small deformation $\alpha_1'$ (yellow) that lies entirely in $J^-(E_{\mathcal{W}_j})^c$ except at its endpoints. (b) If $\alpha_1$ (green) and $\alpha_2$ (yellow) intersect, then there is a curve (black zigzag) from the right endpoint of $\mathcal{V}_1$ to the left endpoint of $\mathcal{V}_3$ that lies outside of $J^-(E_{\mathcal{W}_1}) \cap J^-(E_{\mathcal{W}_2})$ except at a finite set of points. This is called a ``blocking curve'' in the text.}
    \label{fig:blocking-curves}
\end{figure}

If dividing curves $\alpha_j$ and $\alpha_k$ intersect, then we can construct a new curve that starts at the right endpoint of $\mathcal{V}_{j},$ follows $\alpha_j$ until the intersection, then follows $\alpha_k$ to the left endpoint of $\mathcal{V}_{k+1}.$
By deforming this curve ``to the right'' away from the intersection point, we can construct a curve from the right endpoint of $\mathcal{V}_{j}$ to the left endpoint of $\mathcal{V}_{k+1}$ that is nowhere in $J^{-}(E_{\mathcal{W}_j}) \cap J^-(E_{\mathcal{W}_k})$ except at its endpoints and at the single point where the two dividing curves intersect.
Call this the \textit{blocking curve}; see figure \ref{fig:blocking-curves}b.
But the blocking curve lies on $\mathscr{F}$ and completely separates all boundary regions from $\mathcal{V}_{k+1}$ to $\mathcal{V}_{j}$ from the complementary set of regions between $\mathcal{V}_{j+1}$ and $\mathcal{V}_{k}.$
Any curve in $(\mathscr{F} \cup (E_{\mathcal{V}} \cap \Sigma)) \cap J^-(E_{\mathcal{W}_1}) \cap \dots \cap J^-(E_{\mathcal{W}_n})$ connecting entanglement wedge components with boundaries on two different sides of the blocking curve would have to pass through the single point in the blocking curve that lies in $J^-(E_{\mathcal{W}_1}) \cap \dots \cap J^-(E_{\mathcal{W}_n})$; but by the lemma we have just proved, and connectedness of $\Gamma_{2 \rightarrow \text{all}}$, there must be some pair of entanglement wedge components on either side of the blocking curve that can be connected by an open family of curves in $(\mathscr{F} \cup (E_{\mathcal{V}} \cap \Sigma)) \cap J^-(E_{\mathcal{W}_1}) \cap \dots \cap J^-(E_{\mathcal{W}_n}).$
This is a contradiction, so the dividing curves $\alpha_j$ and $\alpha_k$ cannot intersect.

We will also use later the fact that $\alpha_j$ and $\alpha_k$ cannot ever be null-separated along a single generator of $\mathscr{F}$.
The same logic applied above forbids this: assume without loss of generality that on the generator they share, $\alpha_j$ is in the past of $\alpha_k$; then the segment of the generator between $\alpha_j$ and $\alpha_k$ lies nowhere in $J^-(E_{\mathcal{W}_j}),$ and can be combined with slightly-deformed versions of $\alpha_j$ and $\alpha_k$ to form a complete cut of $\mathscr{F}$ for which only a finite set of points is in the past of all output wedges.

\subsubsection*{Singularities}
\label{sec:singularities}

The dividing curves, about which we have now proved the useful topological property that they do not intersect, are not the only pieces of the past lightsheets $\mathscr{P}_j$ that we need to make up a good cross-section of $\mathscr{F}.$
We will also need to add in some closed curves in $\mathscr{P}_j \cap \mathscr{F}$ to wrap any singularities on $\mathscr{F}$. 

Any singularity on $\mathscr{F}$ is in the future of $\gamma_{\mathcal{V}}$, so it cannot be in the past of any output wedge $E_{\mathcal{W}_j}$; as shown in lemma \ref{lem:future-past-singularities}, AdS-hyperbolicity means that every singularity must either be a future singularity or a past singularity.
So if a generator of $\mathscr{F}$ hits a spacetime singularity, it must first pass through each past lightsheet $\mathscr{P}_j.$
So for any given singularity, and for each $j$, the singularity is separated from $\gamma_{\mathcal{V}}$ either by the dividing curve $\alpha_j$ or by a closed loop in $\mathscr{P}_j.$

We will assume that there are a finite number of singularities on $\mathscr{F}$.
A setup with two singularities is sketched for $n=3$ in figure \ref{fig:two-singularities}.
The left singularity is separated from $\gamma_{\mathcal{V}}$ by one dividing curve and two closed curves, while the right singularity is separated from $\gamma_{\mathcal{V}}$ by three closed curves.

\begin{figure}
    \centering
    \begin{tikzpicture}

    \draw[domain=-120:-60,variable=\x,smooth, fill=blue,opacity=0.4] plot ({3*sin(\x)}, {3*cos(\x)}) -- (-3*0.866,3*0.5) to [out=-30,in=30] (-3*0.866,-3*0.5);
    
    \begin{scope}[rotate=120]
    \draw[domain=-120:-60,variable=\x,smooth, fill=blue,opacity=0.4] plot ({3*sin(\x)}, {3*cos(\x)}) -- (-3*0.866,3*0.5) to [out=-30,in=30] (-3*0.866,-3*0.5);
    \end{scope}
    
    \begin{scope}[rotate=240]
    \draw[domain=-120:-60,variable=\x,smooth, fill=blue,opacity=0.4] plot ({3*sin(\x)}, {3*cos(\x)}) -- (-3*0.866,3*0.5) to [out=-30,in=30] (-3*0.866,-3*0.5);
    \end{scope}
    
    \draw[thick] (0,0) circle (3);
    
    \draw[thick,blue] (-3*0.866,3*0.5) to [out=-30,in=30] (-3*0.866,-3*0.5);
    \draw[thick,blue] (0,3) to [out=-90,in=-150] (3*0.866,3*0.5);
    \draw[thick,blue] (0,-3) to [out=90,in=150] (3*0.866,-3*0.5);
    
    \node[left] at (-3,0) {$\mathcal{V}_1$};
    \node[above right] at (3*0.5,3*0.866) {$\mathcal{V}_3$};
    \node[below right] at ({3*cos(-70)},{3*sin(-70)}) {$\mathcal{V}_2$};
    
    \draw[red,thick,dashed] ({3*cos(30)},{3*sin(30)}) -- (3.5,0);
    \draw[red,thick,dashed] ({3*cos(30)},{3*sin(-30)}) -- (3.5,0);
    \draw[red,decorate, decoration={snake, segment length=3mm, amplitude=0.5mm}] (3.5,0) -- (2,0);
    
    \draw[red,thick,dashed] ({3*cos(150)},{3*sin(150)}) -- ({3.5*cos(120)},{3.5*sin(120)});
    \draw[red,thick,dashed] ({3*cos(90)},{3*sin(90)}) -- ({3.5*cos(120)},{3.5*sin(120)});
    \draw[red,decorate, decoration={snake, segment length=3mm, amplitude=0.5mm}] ({2*cos(120)},{2*sin(120)}) -- ({3.5*cos(120)},{3.5*sin(120)});
    
    \draw[red,thick,dashed] ({3*cos(-150)},{3*sin(-150)}) -- ({3.5*cos(-120)},{3.5*sin(-120)});
    \draw[red,thick,dashed] ({3*cos(-90)},{3*sin(-90)}) -- ({3.5*cos(-120)},{3.5*sin(-120)});
    \draw[red,decorate, decoration={snake, segment length=3mm, amplitude=0.5mm}] ({1.4*cos(148)},{1.4*sin(148)}) -- ({3.5*cos(-120)},{3.5*sin(-120)});
    
    \fill[red,opacity=0.3] (0,0) -- (3.5,0) -- ({3*cos(30)},{3*sin(-30)}) -- ({3*cos(-30)},{3*sin(-30)}) to [out=150,in=90] (0,{3*sin(-90)}) -- ({3.5*cos(-120)},{3.5*sin(-120)});
    
    \begin{scope}[rotate=120]
    \fill[red,opacity=0.3] (0,0) -- (3.5,0) -- ({3*cos(30)},{3*sin(-30)}) -- ({3*cos(-30)},{3*sin(-30)}) to [out=150,in=90] (0,{3*sin(-90)}) -- ({3.5*cos(-120)},{3.5*sin(-120)});
    \end{scope}
    
    \begin{scope}[rotate=-120]
    \fill[red,opacity=0.3] (0,0) -- (3.5,0) -- ({3*cos(30)},{3*sin(-30)}) -- ({3*cos(-30)},{3*sin(-30)}) to [out=150,in=90] (0,{3*sin(-90)}) -- ({3.5*cos(-120)},{3.5*sin(-120)});
    \end{scope}
    
    \draw[red,decorate, decoration={snake, segment length=3mm, amplitude=0.5mm}] ({1*cos(45)},{1*sin(45)}) -- ({1.2*cos(90)},{1.2*sin(90)});
    
    \draw[red,decorate, decoration={snake, segment length=3mm, amplitude=0.5mm}] (2,0) to [out=135,in=0] ({0.707},{1*sin(45)});
    
    \draw[red,decorate, decoration={snake, segment length=3mm, amplitude=0.5mm}] ({1*cos(45)},{1*sin(45)}) to [out=-90,in=180] (2,0);
    
    \fill[black,decorate, decoration={snake, segment length=3mm, amplitude=0.5mm},opacity=0.6] ({0.707},{1*sin(45)}) to [out=-90,in=180] (2,0); 
    
    \fill[black,decorate, decoration={snake, segment length=3mm, amplitude=0.5mm},opacity=0.6] (2,0) to [out=135,in=0] ({0.707},{1*sin(45)});
    
    \draw[black,opacity=0.2] (2,0) -- ({0.707},{1*sin(45)});
    
    \fill[decorate, decoration={snake, segment length=3mm, amplitude=0.5mm},opacity=0.6] ({1.2*cos(90)},{1.2*sin(90)}) -- (({1.35*cos(150)},{1.35*sin(150)}) -- ({2*cos(120)},{2*sin(120)}) -- ({1.2*cos(90)},{1.2*sin(90)});
    
    \draw[red,decorate, decoration={snake, segment length=3mm, amplitude=0.5mm}] ({1.2*cos(90)},{1.2*sin(90)}) -- (({1.35*cos(150)},{1.35*sin(150)}) -- ({2*cos(120)},{2*sin(120)}) -- ({1.2*cos(90)},{1.2*sin(90)});
    
    \draw plot [mark=*, mark size=2] coordinates{({1.2*cos(90)},{1.2*sin(90)})};
    
    \draw plot [mark=*, mark size=2] coordinates{({1.35*cos(150)},{1.35*sin(150)})};
    
    \draw plot [mark=*, mark size=2] coordinates{({1*cos(45)},{1*sin(45)})};

    \draw plot [mark=*, mark size=2] coordinates{({2*cos(119)},{2*sin(121)+0.05})};
    
    \draw plot [mark=*, mark size=2] coordinates{(2,0)};
    
    \begin{scope}[rotate=60]
    \draw[domain=-120:-60,variable=\x,smooth,magenta,thick] plot ({3*sin(\x)}, {3*cos(\x)}) -- (-3*0.866,3*0.5) to [out=-30,in=30] (-3*0.866,-3*0.5);
    \end{scope}
    
   \draw[magenta,thick] plot [smooth cycle] coordinates {({2.3*cos(120)},{2.3*sin(120)}) ({1.5*cos(160)},{1.5*sin(160)}) (0.5,1.2)};
    
    \draw[magenta,thick] plot [smooth cycle] coordinates {(0.4,0.8) (1,1) (2.4,0) (1,-0.3)};
    
    \draw[purple,thick] ({2.598 },{3*sin(30)}) to [out=-110,in=110] ({2.598},{3*sin(-30)});
    
    \draw[purple,thick] plot [smooth cycle,tension=0.5] coordinates {(0.5,0.8) (1.2,1.2) (2.2,0) (1,-0.5)};
    
    \draw[ForestGreen,thick] plot [smooth cycle,tension=0.5] coordinates {(0.6,0.7) (1.5,1.1) (2.3,0) (1,-0.8)};
    
    \draw[ForestGreen,thick] plot [smooth,tension=0.5] coordinates {(0,3) (-0.1,2) (0.3,1) (-1,0.3) (-2,1.4) (-2.59,1.5)};
    
    \draw[purple,thick] plot [smooth cycle,tension=0.5] coordinates {({2.5*cos(120)},{2.3*sin(120)}) ({1.2*cos(160)},{1.2*sin(160)}) (0.4,1.5)};
    
\end{tikzpicture}
    \caption{An example of a focusing lightsheet in a spacetime that contains two singularities. The singularities are regions of the bulk where the lightsheet ends. The curves where the three past lightsheets $\mathscr{P}_1, \mathscr{P}_2, \mathscr{P}_3$ intersect $\mathscr{F}$ are drawn in three distinct colors. Each of these has an open component (the dividing curve), and lightsheets $\mathscr{P}_1$ and $\mathscr{P}_2$ have two closed components, while lightsheet $\mathscr{P}_3$ has only a single closed component. The left singularity is separated from $\gamma_{\mathcal{V}}$ by the dividing curve $\alpha_3$, while the right singularity is only separated from $\gamma_{\mathcal{V}}$ by closed curves.}
    \label{fig:two-singularities}
\end{figure}
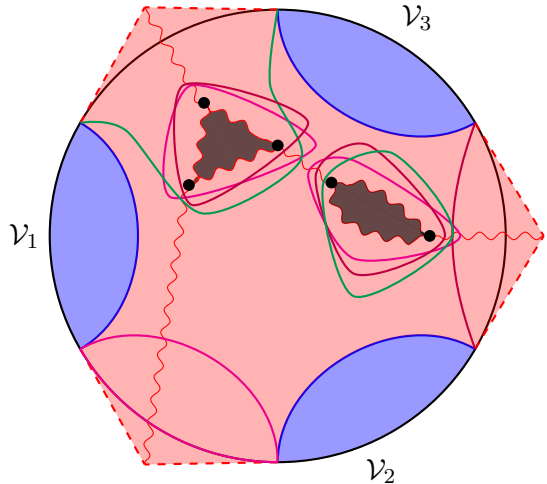

We will ignore any singularity that is separated from $\gamma_{\mathcal{V}}$ by a dividing curve; when we construct the ``good cross-section'' of $\mathscr{F}$ in the next subsection, we will see that these singularities are automatically taken care of by the portions of the good cross-section that lie in dividing curves.

Pick a singularity such that for all $j,$ the singularity is separated from $\gamma_{\mathcal{V}}$ by a closed curve.
These $n$ different closed curves may intersect each other in complicated ways; see again figure \ref{fig:two-singularities}.
We claim, however, that there can be at most one dividing curve $\alpha_j$ that intersects any of the closed curves.
If there were two distinct dividing curves that intersected any of the closed curves, then as in the previous subsection, we would be able to construct a cut across $\mathscr{F}$ for which only a finite set of points is in $J^-(E_{\mathcal{W}_1}) \cap \dots \cap J^-(E_{\mathcal{W}_n}),$ which contradicts the main lemma of the previous subsection.
The construction of such a curve is shown in figure \ref{fig:no-singularity-intersections}.

\begin{figure}
    \centering
    
    \begin{tikzpicture}
    \begin{scope}[rotate=90]
    \draw[thick] (0,0) circle (3);
     
    \begin{scope}[rotate=0]
    \draw[domain=-45:45,variable=\x,smooth, fill=blue,opacity=0.4] plot ({3*cos(\x)}, {3*sin(\x)}) to [out=-135,in=135] (2.12,-2.12);
    \draw[thick,blue] (2.12,2.12) to [out=-135,in=135] (2.12,-2.12);
    \end{scope}
    
    \begin{scope}[rotate=180]
    \draw[domain=-45:45,variable=\x,smooth, fill=blue,opacity=0.4] plot ({3*cos(\x)}, {3*sin(\x)}) to [out=-135,in=135] (2.12,-2.12);
    \draw[thick,blue] (2.12,2.12) to [out=-135,in=135] (2.12,-2.12);
    \end{scope}
    
    \node[below] at (-3,0) {$\mathcal{V}_1$};
    \node[above] at (3,0) {$\mathcal{V}_2$};
    
    \draw[red,thick,dashed] (-2.12,2.12) -- (0,4) -- (2.12,2.12);
    \draw[red,thick,dashed] (-2.12,-2.12) -- (0,-4) -- (2.12,-2.12);
    
    \fill[red,opacity=0.3,domain=135:45,variable=\x,smooth] plot ({3*cos(\x)}, {3*sin(\x)}) -- (2.12,2.12) to [out=-135,in=135] (2.12,-2.12);
    \begin{scope}[rotate=180]
    \fill[red,opacity=0.3,domain=135:45,variable=\x,smooth] plot ({3*cos(\x)}, {3*sin(\x)}) -- (2.12,2.12) to [out=-135,in=135] (2.12,-2.12);
    \end{scope}
    
    \draw[red,decorate, decoration={snake, segment length=3mm, amplitude=0.5mm}] (0,4) -- (0,1);
    \draw[red,decorate, decoration={snake, segment length=3mm, amplitude=0.5mm}] (0,-4) -- (0,-1);
    
    \draw[red,decorate, decoration={snake, segment length=3mm, amplitude=0.5mm}] (0,-1) to [out=45,in=-45] (0,1) to [out=-135,in=135] (0,-1);
    
    \fill[black,decorate, decoration={snake, segment length=3mm, amplitude=0.5mm},opacity=0.6] (0,-1) to [out=45,in=-45] (0,1) to [out=-135,in=135] (0,-1);
    
    \draw plot [mark=*, mark size=2] coordinates{(0,1)};
    \draw plot [mark=*, mark size=2] coordinates{(0,-1)};
    
    \draw[ForestGreen,thick] (-2.12,-2.12) to [out=25,in=155] (2.12,-2.12);
    \draw[purple,thick] (-2.12,2.12) to [out=-25,in=-155] (2.12,2.12);
    
    \draw[ForestGreen,thick] plot [smooth cycle,tension=1] coordinates {(0,2.25) (0.75,0.5) (0,-1.4) (-0.75,0.5)};
    
    \begin{scope}[shift={(0,-1)}]
    \draw[purple,thick] plot [smooth cycle,tension=1] coordinates {(0,2.25) (0.75,0.5) (0,-1.4) (-0.75,0.5)};
    \end{scope}
    
    \draw[decorate, decoration={zigzag, segment length=2mm, amplitude=0.3mm}] (-2,2.2) to [out=-25,in=170] (-0.4,1.7) to [out=-100,in=100] (-0.6,0) to [out=-100,in=100]  (-0.4,-1.7) to [out=-170,in=25] (-2,-2.2);
    
    \end{scope}
    
    \end{tikzpicture}
    \caption{This figure illustrates a contradiction that would arise if more than one dividing curve intersected the closed curves wrapping a singularity. The black zigzag curve in this figure lies in $J^-(E_{\mathcal{W}_1}) \cap J^-(E_{\mathcal{W}_2})$ only at a single set of points, which contradicts the lemma that there exists an open family of curves on $\mathscr{F}$ connecting $E_{\mathcal{V}_1}$ to $E_{\mathcal{V}_2}$.}
    \label{fig:no-singularity-intersections}
\end{figure}
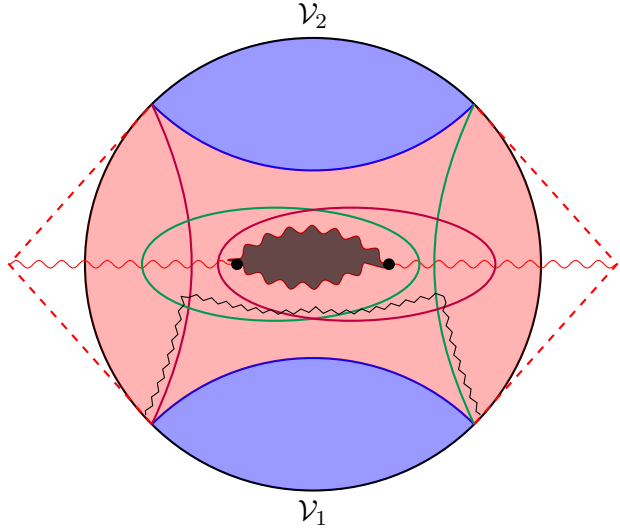

Let $j$ be the index such that $\alpha_j$ is the only dividing curve that intersects any of the closed curves wrapping our chosen singularity.
(If there are no such curves, pick any $j$.)
Since $\mathscr{P}_j \cap \mathscr{F}$ does not self-intersect, the closed curve wrapping the singularity in $\mathscr{P}_j \cap \mathscr{F}$ does not intersect any dividing curve.

We now have a set of $n$ dividing curves $\{\alpha_1, \dots, \alpha_n\},$ and a closed curve in some $\mathscr{P}_j \cap \mathscr{F}$ wrapping our chosen singularity, with the property that none of these curves intersect one another.
If there remains a singularity on $\mathscr{F}$ that is not separated from $\gamma_{\mathcal{V}}$ by a dividing curve or by the singularity-wrapping curve we have already chosen, we can proceed as above to find a closed curve in some $\mathscr{P}_k$ that separates the singularity from $\gamma_{\mathcal{V}}$ and that does not intersect $\{\alpha_1, \dots, \alpha_n\}$ or the singularity-wrapping curve in $\mathscr{P}_j.$

Proceeding by induction, we can add to $\{\alpha_1, \dots, \alpha_n\}$ a set of closed curves on $\mathscr{F}$ such that each curve is in some $\mathscr{P}_j,$ such that no two curves intersect, and such that every singularity on $\mathscr{F}$ is separated from $\gamma_{\mathcal{V}}$ by one of the curves in our set.
We will use the symbol $\mathscr{C}$ to denote the union of all these curves.

\subsubsection*{A good cross-section}

We are now ready to construct a good cross-section of $\mathscr{F}.$
For each generator of $\mathscr{F},$ we truncate it as soon as it hits a point in the set $\mathscr{C}$ constructed in the previous subsection.
Some generators of $\mathscr{F}$ will never hit $\mathscr{C}$; these generators hit a caustic first and leave $\mathscr{F}$.
(The only other possibility would be for the generators to hit a singularity, but we know by construction that every generator that would hit a singularity must hit a point in $\mathscr{C}$ first.)
The generators that hit caustics before hitting $\mathscr{C}$ are truncated at their caustics.
This gives a prescription for where to truncate every single generator of $\mathscr{F},$ with the truncation set consisting of points in $\mathscr{C}$, and some set of caustics on $\mathscr{F}.$

We claim that the full set $\mathscr{C}$ is included in the truncation set, \textit{except} possibly for some null segments of $\mathscr{C}$ that can be added to the truncation set without incurring any increase in area.
This follows immediately from the statement given above that no two dividing curves can be null separated along a single generator, together with an analogous argument applied to the singularity-wrapping curves constructed in the previous subsection.
So the only possible redundancies, where multiple points in $\mathscr{C}$ set share a single generator, happen on null segments of a single curve.

We also claim that the caustic set is nontrivial whenever $E_{\mathcal{V}}$ is disconnected; this follows from the fact that the curves making up $\mathscr{C}$ do not intersect, so they \textit{cannot} by themselves form a complete cross-section of $\mathscr{F}$; there will always be space between components of $\mathscr{C}$ that needs to be filled with caustics. See figure \ref{fig:needed-caustics}.

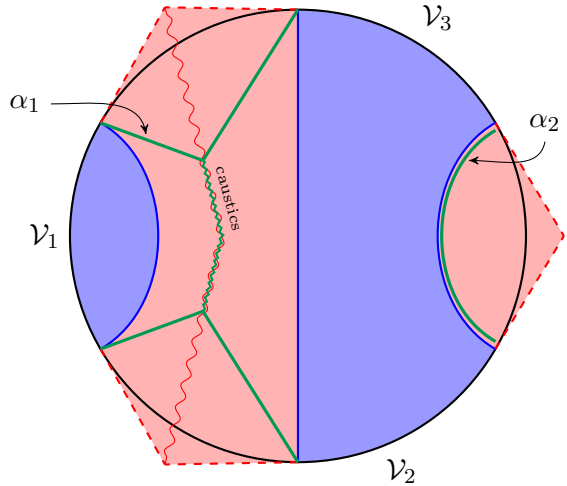
\begin{figure}
    \centering
    \begin{tikzpicture}

    \draw[domain=-120:-60,variable=\x,smooth, fill=blue,opacity=0.4] plot ({3*sin(\x)}, {3*cos(\x)}) -- (-3*0.866,3*0.5) to [out=-30,in=30] (-3*0.866,-3*0.5);
    
    \fill[red,opacity=0.3] ({2.598},{1.5}) -- (3.5,0) -- ({2.598},{-1.5}) to [out=150,in=-150] ({2.598},{1.5});
    
    \fill[blue,opacity=0.4,domain=30:-30] (0,3) plot ({3*cos(\x+60)},{3*sin(\x+60)}) to [out=-150,in=150] ({2.598},{-1.5}) plot ({3*cos(\x-60)},{3*sin(\x-60)}) -- (0,3);
    
    \fill[red,opacity=0.3,domain=30:-30] ({-2.598},{1.5}) -- ({3.5*cos(120)},{3.5*sin(120)}) -- (0,3) -- (0,-3) -- ({3.5*cos(-120)},{3.5*sin(-120)}) -- ({-2.598},{-1.5}) to [out=30,in=-30] ({-2.598},{1.5});
    
    \draw[thick] (0,0) circle (3);
    
    \draw[thick,blue] (-3*0.866,3*0.5) to [out=-30,in=30] (-3*0.866,-3*0.5);
    \draw[thick,blue] (0,3) to (0,-3);
    \draw[thick,blue] (2.598,1.5) to [out=-150,in=150] (2.598,-1.5);
    
    \draw[very thick,ForestGreen] (2.6,1.4) to [out=-150,in=150] (2.6,-1.4);
    
    \node[left] at (-3,0) {$\mathcal{V}_1$};
    \node[above right] at (3*0.5,3*0.866) {$\mathcal{V}_3$};
    \node[below right] at ({3*cos(-70)},{3*sin(-70)}) {$\mathcal{V}_2$};
    
    \draw[red,thick,dashed] ({3*cos(30)},{3*sin(30)}) -- (3.5,0);
    \draw[red,thick,dashed] ({3*cos(30)},{3*sin(-30)}) -- (3.5,0);
    
    \draw[red,thick,dashed] ({3*cos(150)},{3*sin(150)}) -- ({3.5*cos(120)},{3.5*sin(120)});
    \draw[red,thick,dashed] ({3*cos(90)},{3*sin(90)}) -- ({3.5*cos(120)},{3.5*sin(120)});
    \draw[red,decorate, decoration={snake, segment length=3mm, amplitude=0.5mm}] (-1,0) -- ({3.5*cos(120)},{3.5*sin(120)});
    
    \draw[red,thick,dashed] ({3*cos(-150)},{3*sin(-150)}) -- ({3.5*cos(-120)},{3.5*sin(-120)});
    \draw[red,thick,dashed] ({3*cos(-90)},{3*sin(-90)}) -- ({3.5*cos(-120)},{3.5*sin(-120)});
    \draw[red,decorate, decoration={snake, segment length=3mm, amplitude=0.5mm}] (-1,0) -- ({3.5*cos(-120)},{3.5*sin(-120)});
    
    \draw[very thick,ForestGreen] (-2.598,1.5) -- (-1.25,1) -- (0,3);
    \draw[very thick,ForestGreen] (-2.598,-1.5) -- (-1.25,-1) -- (0,-3);
    \draw[thick,ForestGreen,decorate, decoration={zigzag, segment length=1mm, amplitude=0.2mm}] (-1.25,1) -- (-1,0) -- (-1.25,-1);
    
    \node[rotate=-75] at (-0.9,0.5) {\tiny{caustics}};
    
    \draw[->] (3.25,1.25) to [out=-100,in=0] (2.25,1);
    \node[above] at (3.25,1.25) {$\alpha_2$};
    
    \draw[->] (-3.25,1.75) to [out=0,in=100] (-2,1.35);
    \node[left] at (-3.25,1.75) {$\alpha_1$};
    
\end{tikzpicture}    
    \caption{If the entanglement wedge is not fully connected, then the fact that the dividing curves do not intersect means any complete cross-section of $\mathscr{F}$ will need to include caustics.}
    \label{fig:needed-caustics}
\end{figure}

From these considerations, we conclude that $\mathscr{C}$, supplemented with a nontrivial caustic set and up to null segment redundancies, is a complete cross-section of $\mathscr{F}.$
It follows from the focusing theorem that $\mathscr{C}$ has area strictly less than $\gamma_{\mathcal{V}}.$

Furthermore, we can conclude that $\mathscr{C}$ is homologous to $\gamma_{\mathcal{V}}$ through $\mathscr{F}.$
The curve $\mathscr{C}$ is constructed by deforming $\gamma_{\mathcal{V}}$ continuously along $\mathscr{F},$ and deleting points on caustics.
The deleted points do not ruin the homology condition, since caustic points occur where multiple generators of $\mathscr{F}$ collide --- since these points are approached by our deformation on two sides, they are homologically trivial.

\subsubsection*{Completing the proof}

We have constructed a curve $\mathscr{C}$ on $\mathscr{F}$ that is homologous to $\gamma_{\mathcal{V}}$ with strictly smaller area.
Furthermore, each component of $\mathscr{C}$ is either a closed curve, or has endpoints on the maximin slice $\Sigma$, and lies entirely in a single past lightsheet $\mathscr{P}_j.$
We will now use $\mathscr{C}$ to construct a curve in $\Sigma$ homologous to $\gamma_{\mathcal{V}}$ that has strictly less area.

Let $\mathcal{R}$ be the homology region on $\mathscr{F}$ between $\mathscr{C}$ and $\gamma_{\mathcal{V}}$.
Consider the curve
\begin{equation}
    \tilde{\gamma} = \del J^-(\mathcal{R}) \cap \Sigma - (\gamma_{\mathcal{V}} - \mathscr{C}).
\end{equation}
The past lightsheet $\del J^-(\mathcal{R})$ coincides, near a point in $\mathscr{C}$, with the past lightsheet $\mathscr{P}_j$ containing that point.
So by the focusing theorem applied to $\mathscr{P}_j,$ the expansion of $\del J^-(\mathcal{R}),$ as measured ``toward the past,'' is initially negative.
The focusing theorem then implies that $\tilde{\gamma}$ has area no greater than the area of $\mathscr{C},$ and thus has area strictly less than $\gamma_{\mathcal{V}}.$

The final step to complete the proof is to show that $\tilde{\gamma}$ is homologous to $\gamma_{\mathcal{V}}$.
But $\mathscr{C}$ is homologous to $\gamma_{\mathcal{V}},$ and the only way the homology condition could be ``lost'' during focusing would be for some generators of $\del J^-(\mathcal{R})$ to hit the spacetime boundary before hitting $\Sigma$; in this case, $\tilde{\gamma}$ would be homologous to the union of $\gamma_{\mathcal{V}}$ and some set of boundary points.

But this is forbidden, since $\mathscr{C}$ is a subset of $\mathscr{F},$ and no point of $\mathscr{F}$ is in the timelike future of any point in $E_{\mathcal{V}}.$
By lemma \ref{lem:causal-timelike-addition}, this implies that no point of $J^-(\mathcal{R})$ is in the timelike future of any point in $E_{\mathcal{V}}.$
In particular, no point in $J^-(\mathcal{R})$ is in the timelike future of any point in $\mathcal{V}_{j}$ for any $j.$
So the futuremost boundary points that could possibly be in $J^-(\mathcal{R})$ are the futuremost boundary points that are not in the timelike future of any $\mathcal{V}_{j}.$ These are the horizons of the intermediary regions $X_j$ --- which are part of $\Sigma$ due to our application of the restricted maximin formula --- and the past horizons of the input regions $\mathcal{V}_{j},$ which are either in $\Sigma$ or in its past depending on our choice.

\section{Quantum tasks and holography}\label{sec:tasks-section}

In this section we discuss the ``quantum tasks'' framework for studying relativistic quantum information processing, and its application to holography.
Much of the material in this section has appeared before in \cite{may2019quantum, may2020holographic, may2021holographic}, but the discussion has been adapted in light of insights from the present paper.
The framework introduced in this section will be applied in section \ref{sec:QI-argument} to study information-theoretic aspects of the $n$-to-$n$ connected wedge theorem proved in section \ref{sec:GR-proof}.

The main goal of this section is to argue that AdS/CFT provides a way to map bulk tasks to boundary tasks, and that this map can be framed as a map from bulk causal networks with no entanglement to boundary causal networks with the CFT state as a resource.
This allows us to use ideas from AdS/CFT to study which quantum operations can be completed on which causal networks, and which entanglement resources are needed to make this happen.

\subsection{Information processing in causal networks and quantum tasks}
\label{sec:tasks-background}

Much of quantum information theory is concerned with how information processing is constrained by the laws of quantum mechanics.
Relativistic quantum information theory arises when one adds the additional constraints of special relativity, and concerns how information can be processed in spacetime.
To formulate this theory, it is useful to use an operational language by introducing two parties named Alice and Bob.
Alice and Bob are each agencies, with many agents. 
Each agent may move through spacetime and act according to the ordinary rules of quantum mechanics and relativity. 
We consider scenarios where Alice and Bob together carry out a \emph{relativistic quantum task} \cite{kent2012quantum}.\footnote{We also point out that a quantum task can viewed as a generalization of a \emph{quantum game}, see e.g. \cite{cooney2015rank}. A particular quantum task was studied from the perspective of quantum games in \cite{junge2022geometry}.}
Our definition roughly follows \cite{may2021holographic}, though we have developed the definition given there to suit the needs of the present paper. 
\begin{definition} \label{def:quantum-task}
A \textbf{relativistic quantum task} is defined by a tuple
\begin{equation}
    \mathbf{T}=\{\mathcal{M}, \Sigma_0, \Sigma_1, \ket{\Psi}_0, \{C_j\}_{j=1}^{m}, \{R_j\}_{j=1}^{n}, A, B, \{W_s\}, \{\ket{\Phi_s}_{A W_s}\}, \{\Lambda_{B W_s}\}\},
\end{equation} where:
\begin{itemize}
    \item $\mathcal{M}$ is a globally hyperbolic spacetime.
    \item $\Sigma_0$ and $\Sigma_1$ are Cauchy slices of $\mathcal{M}$, with $\Sigma_1$ in the future of $\Sigma_0.$
    \item $\ket{\Psi}_{0}$ is a distributed quantum resource state on $\Sigma_0$, for example the state of any quantum fields present in the spacetime.
    \item $\{C_j\}_{j=1}^{m}$ is a set of disjoint \textit{input regions} on $\Sigma_0,$ and $\{R_j\}_{j=1}^{n}$ is a set of disjoint \textit{output regions} on $\Sigma_1.$
    These regions have the property that either $C_j \subseteq J^-(R_k)$ or $C_j \subseteq J^-(R_k)^c$, and $R_{k} \subseteq J^+(C_j)$ or $R_{k} \subseteq J^+(C_j)^{c}$.\footnote{This requirement is added so that we never have an input location $C_j$ such that part of that region can signal $R_k$ and part cannot. This ensures the input systems causal relationship to the output location is not dependent on how it is recorded into $C_j$.}
    \item $A=A_1 A_2 \dots A_m$ and $B = B_1 B_2 \dots B_n$ are \textit{input and output quantum systems}, i.e., $m$-partite and $n$-partite Hilbert spaces.
    \item $\{W_s\}$ is a set of reference systems that can be held by Bob, indexed by some label $s.$
    \{$\ket{\Phi_s}_{A W_s}$\} is a set of quantum states that can be distributed across $A$ and $W_s.$
    \item For each $s,$ $\{\Lambda_{B W_s}, 1 - \Lambda_{B W_s}\}$ is a POVM on the output systems $B$ and the reference system $W_s.$
\end{itemize}
An \textbf{instantiation} of the task $\mathbf{T}$ is an event where Bob picks a value of $s$ randomly, using a distribution known to him and Alice, and prepares systems $A_1 \dots A_m W_s$ in the state $\ket{\Phi_s}_{A W_s}.$
Alice is not told $s$.
Bob gives each system $A_j$ to Alice in the spacetime region $C_j$, and keeps system $W_s$ for himself. 
Alice returns some state on the quantum systems $B_1 \dots B_n$ at the output locations $R_1$ through $R_n.$
Bob measures the POVM $\{\Lambda_{B W_s}, 1-\Lambda_{B W_s}\}$ and declares the task successful if he finds $\Lambda_{B W_s}$. 
\end{definition}
Intuitively, Alice is given an unknown state distributed across certain spacetime locations, and is tasked with returning another state to another set of spacetime locations.
Bob then performs a measurement to determine whether she has performed his task successfully.

We believe (though have not proved) that all the causal data that determines whether a particular task can be completed is contained in an associated causal network, constructed as follows.
The light cones $\{\del J^-(R_j), \del J^+(C_k)\},$ restricted to $\Sigma_0,$ divide $\Sigma_0$ into multiple components.
Some of these components are the input regions $C_k,$ while some are not.
For each component, we create an input vertex in our causal network.
We create an output vertex for each output region $R_j.$
We then consider the spacetime
\begin{equation}
    \mathcal{M}_{\mathbf{T}} = J^+(\Sigma_0) \cap (J^-(R_1) \cup \dots \cup J^-(R_n)).
\end{equation}
The light cones $\{\del J^-(R_j), \del J^+(C_k)\}$ divide $\mathcal{M}_{\mathbf{T}}$ into spacetime components.
For each such component, we add an intermediate vertex to our network.
For any two vertices in our network, corresponding to spacetime sets $S_1$ and $S_2$, we add a directed edge from $S_1$ to $S_2$ if there exists points $p_1\in S_1$, $p_2\in S_2$ with $p_2\in J^+(p_1)$, i.e. if it is possible to signal from somewhere in $S_1$ to somewhere in $S_2$.\footnote{In fact, because our regions are defined as the components of $\mathcal{M}_\mathbf{T}$, this implies $S_2 \subseteq J^+(S_1) \quad \text{and} \quad S_1 \subseteq J^-(S_2)$, i.e. that any point in $S_1$ can signal some point in $S_2$.}

The causal network thus produced can be simplified.
If an edge connects two vertices that are also connected by a directed path that passes through other vertices, that edge can be removed, because sending a quantum system directly from $S_1$ to $S_2$ will never let you perform a task that can't also be accomplished by sending that system from $S_1$ to $S_3$ to $S_2.$
Furthermore, if a vertex has only one outgoing edge or only one ingoing edge, it can be fused with its unique outgoing or ingoing neighbor.
Applying these two simplifying procedures recursively yields a maximally simple causal network that contains all the causal constraints we think are relevant to $\mathbf{T}.$
\begin{definition}
    The causal network associated to a quantum task $\mathbf{T},$ constructed via the procedure we have just outlined, is denoted $\Gamma_{\mathbf{T}}.$
\end{definition}
\noindent The causal network $\Gamma_{\mathbf{T}}$ should be thought of as a coarse-graining of the geometry $\mathcal{M}$ that preserves the causal features of $\mathcal{M}$ relevant for completing an instance of $\mathbf{T}.$

A useful quantity that characterizes a quantum task is its success probability, $p_{\text{suc}}(\mathbf{T})$.
This is the probability that the outputs pass Bob's POVM test, optimized over choices of Alice's strategy and averaged over the distribution over $s$.
The success probability of a quantum task is sensitive to Bob's distribution of POVMs, the causal network $\Gamma_\mathbf{T}$, and the resource state $\ket{\Psi}_{0}$. 
Two example networks are shown in figure \ref{fig:2-2networksrepeated}. 
In the network on the left, the inputs and outputs can be made to be related by an arbitrary quantum operation, simply by bringing the inputs together, doing the needed operation, and sending the outputs to their respective locations. 
In the network on the right, the causal structure is now weaker, but the same operations can be performed with arbitrarily high $p_{\text{suc}}$ if a sufficient number of EPR pairs are shared between the two input locations \cite{buhrman2014position,beigi2011simplified}. 

\begin{figure}
    \centering
    \subfloat[]{
    \begin{tikzpicture}[scale=1]
    
    \begin{scope}[shift={(0,0)}]
    \draw[mid arrow, red] (0,0) -- (2,2);
    \draw[mid arrow, red] (0,0) -- (-2,2);
    \draw[mid arrow, red] (-2,-2) -- (0,0);
    \draw[mid arrow, red] (2,-2) -- (0,0);
    
    \draw plot [mark=*, mark size=2] coordinates{(0,0)};
    \draw plot [mark=*, mark size=2] coordinates{(-2,-2)};
    \draw plot [mark=*, mark size=2] coordinates{(2,2)};
    \draw plot [mark=*, mark size=2] coordinates{(-2,2)};
    \draw plot [mark=*, mark size=2] coordinates{(2,-2)};
    \end{scope}
    
    \end{tikzpicture}
    }
    \hfill
    \subfloat[\label{fig:2-2repeatednlqc}]{
    \begin{tikzpicture}[scale=1]
    
    \draw[mid arrow, red] (-2,-2) -- (0,0);
    \draw[red, mid arrow] (0,0) -- (2,2);
    \draw[mid arrow, red] (2,-2) -- (0.1,-0.1);
    \draw[mid arrow, red] (-0.1,0.1) -- (-2,2);
    \draw[mid arrow, red] (-2,-2) -- (-2,2);
    \draw[mid arrow, red] (2,-2) -- (2,2);
    
    \draw plot [mark=*, mark size=2] coordinates{(-2,-2)};
    \draw plot [mark=*, mark size=2] coordinates{(2,2)};
    \draw plot [mark=*, mark size=2] coordinates{(-2,2)};
    \draw plot [mark=*, mark size=2] coordinates{(2,-2)};
    
    \end{tikzpicture}
    }
    \caption{Repeated from figure \ref{fig:causalnetworks}. Two possible causal networks, both with two input and two output locations. a) Arbitrary tasks can be completed by bringing the input systems together at the central vertex. b) Arbitrary tasks can be completed on this network, but only when sufficient entanglement is shared between the input locations.}
    \label{fig:2-2networksrepeated}
\end{figure}
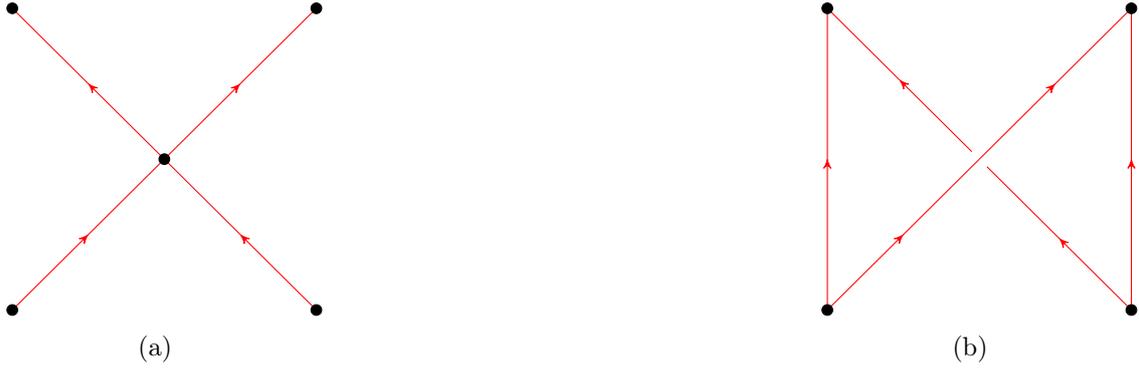

Quantum tasks can be defined directly on causal networks, rather than in spacetime, in an obvious way.
One useful theorem is that, in the presence of unrestricted preshared entanglement, the success probability of a task is completely determined by the set of paths connecting input and output vertices, without any data about the intermediate vertices.
This is the content of a theorem suggested in \cite{buhrman2014position}, and proven in \cite{dolev2019constraining}. 
\begin{theorem}[Dolev] \label{thm:dolev}
    Consider a quantum task
    \begin{equation}
    \mathbf{T}=\{\Gamma, \ket{\Psi}_0, A, B, \{W_s\}, \{\ket{\Phi_s}_{A W_s}\}, \{\Lambda_{B W_s}\}\}.
    \end{equation}
    Let $\tilde{\Gamma}$ be the causal network whose input and output vertices are the input and output vertices of $\Gamma$, which has no intermediary vertices, and for which there is an edge from an input vertex to an output vertex if and only if there is a directed path connecting the corresponding vertices in $\Gamma.$
    
    Then there exists a one parameter family of resource states $\ket{\Psi}_n$ such that the tasks
    \begin{equation}
    \mathbf{T}_n=\{\tilde{\Gamma}, \ket{\Psi_n}_0, A, B, \{W_s\}, \{\ket{\Phi_s}_{A W_s}\}, \{\Lambda_{B W_s}\}\}.
    \end{equation}
    satisfy
    \begin{equation}
        p_{\text{suc}}(T)=\lim_{n\rightarrow \infty} p_{\text{suc}}(T_n).
    \end{equation} 
\end{theorem}
In the proof of this theorem, the state $\Psi_n$ chosen includes a number of distributed EPR pairs that grows with the parameter $n$. 
We use this theorem in proving theorem \ref{thm:2-to-all-sufficiency} below. 

\subsection{Holographic quantum tasks}\label{sec:tasks-holography}

The AdS/CFT correspondence asserts the equivalence of two theories, a bulk gravitational theory and a boundary conformal field theory. 
Because the success probability of a quantum task is sensitive to the initial state and to causal structure, quantum tasks can serve as interesting tools in understanding this correspondence. 
In particular, we can map quantum tasks in the bulk theory to quantum tasks in the boundary theory as follows.

Consider a quantum task in the bulk theory,
\begin{equation}
    \mathbf{T} = \{\mathcal{M}, \Sigma_0, \Sigma_1, \ket{\Psi}_0, \{C_j\}_{j=1}^{m}, \{R_j\}_{j=1}^{n}, A, B, \{W_s\}, \{\ket{\Phi_s}_{A W_s}\}, \{\Lambda_{B W_s}\}\}
\end{equation}
where the spacetime $\mathcal{M}$ together with the quantum field theory state $\ket{\Psi}_0$ is a semiclassical state of the AdS/CFT correspondence.\footnote{Note that while our formal definition of quantum tasks assumed that the spacetime was globally hyperbolic, there is no obstruction to defining quantum tasks on AdS-hyperbolic spacetimes so long as ``Cauchy slice'' is replaced with ``Cauchy slice for the conformal completion $\tilde{\mathcal{M}}.$''}
Because the boundary completely describes the bulk, any process which completes the task $\mathbf{T}$ completes some corresponding task in the boundary theory.
To identify a boundary task corresponding to $\mathbf{T},$ pick any boundary regions $\{\mathcal{C}_j, \mathcal{R}_j\}$ such that $C_j$ is contained in the entanglement wedge of $\mathcal{C}_j,$ and $R_j$ is contained in the entanglement wedge of $\mathcal{R}_j$; in the notation of other sections, this statement is $C_j \subseteq E_{\mathcal{C}_j}, R_j \subseteq E_{\mathcal{R}_j}.$
We now map the elements of the tuple defining $\mathbf{T}$ to boundary quantities as follows.
\begin{itemize}
    \item The bulk geometry $\mathcal{M}$ is replaced with its conformal boundary, $\partial \mathcal{M}$.
    \item The slices $\Sigma_0$ and $\Sigma_1$ are replaced by their restrictions to the conformal boundary, $\hat{\Sigma}_0$ and $\hat{\Sigma}_1.$
    \item The state $\ket{\Psi}_{0}$ describing the bulk fields is replaced with the state of the CFT on $\hat{\Sigma}_0$, which we denote by $\ket{\Psi}_{\text{CFT}}$.
    \item The input and output regions $C_j$ and $R_j$ are replaced by $\mathcal{C}_j$ and $\mathcal{R}_j.$
    \item The input systems $A$ and output systems $B$ are unchanged, as are $W_s,$ $\ket{\Phi_s}_{A W_s},$ and $\Lambda_{B W_s}.$
\end{itemize}

The reasoning behind this identification comes from entanglement wedge reconstruction \cite{CKNR, jafferis2016relative, dong2016reconstruction, cotler2019entanglement}.
Any systems recorded into $C_j$ in the bulk are recorded into $\mathcal{C}_j$ in the boundary; similarly, any systems recorded into $R_j$ in the bulk are recorded into $\mathcal{R}_j$ in the boundary.
Indeed, if the bulk protocol completes the task with some probability, it will also be true that the corresponding boundary protocol completes it with no smaller probability.
This leads to the inequality
\begin{align}\label{eq:psucinequality}
    p_{\text{suc}}(\mathbf{T}) \leq p_{\text{suc}}(\mathbf{\hat{T}}).
\end{align}
This is an inequality, and not an equality, because we defined our bulk task as something that takes place within a single semiclassical background; this means that $\mathbf{T}$ is defined within the strict large-N limit of the AdS/CFT correspondence.
The task $\mathbf{\hat{T}},$ on the other hand, is sensitive to finite-$N$ physics, like backreaction on spacetime.
As a result, while every bulk strategy can be performed in the boundary, not every boundary strategy can be performed in the bulk; this observation results in inequality \eqref{eq:psucinequality}.

Inequality \eqref{eq:psucinequality} is the starting point for several existing papers \cite{may2019quantum,may2020holographic,may2021holographic,may2021bulk,may2021quantum}, which use it to understand the relationship between bulk causal structure and boundary correlation. 
The basic strategy is to pick a task for which bulk causal structure allows $p_{\text{suc}}(\mathbf{T})$ to be large, use \eqref{eq:psucinequality} to lower bound $p_{\text{suc}}(\mathbf{\hat{T}})$, and then argue that the boundary needs some particular entanglement feature to allow for this large success probability.
We will now observe that a stronger inequality than \eqref{eq:psucinequality} holds; this stronger inequality will be used to argue for the $n$-to-$n$ connected wedge theorem in subsection \ref{sec:final-tasks-argument}.

Inequality \eqref{eq:psucinequality} follows from the fact that a strategy used to complete the task $\mathbf{T}$ also completes $\hat{\mathbf{T}}.$
We observed above that these are not the only strategies that can be used to complete $\hat{\mathbf{T}}$.
We now further remark that they are not the only \textit{semiclassical} strategies that can be used to complete $\hat{\mathbf{T}}.$
This is because, given access to a boundary region $\mathcal{C}_j$, a boundary agent can exploit entanglement wedge reconstruction to change the bulk state $\ket{\Psi}_0$ however they like within the entanglement wedges $E_{\mathcal{C}_j}$.
This action is a legal part of a boundary strategy for $\hat{\mathbf{T}},$ but it changes the bulk quantum task $\mathbf{T}$ to a different task $\mathbf{T}_{\Psi_0'}.$
Since any strategy for completing $\mathbf{T}_{\Psi_0'}$ can also be used to complete $\hat{\mathbf{T}},$ we may conclude
\begin{equation} \label{eq:psucnewinequality}
    p_{\text{suc}}(\mathbf{\hat{T}}) \geq \sup_{\Psi_0'} p_{\text{suc}}(\mathbf{T}_{\Psi_0'}),
\end{equation}
where $\ket{\Psi}_0'$ is any state of the bulk quantum fields that agrees with $\ket{\Psi}_0$ outside of $\cup_j E_{\mathcal{C}_j}.$

Throughout this paper, we have also tried to emphasize that equation \eqref{eq:psucinequality} can be used as a tool for understanding relativistic quantum information theory.
AdS/CFT gives us, for any bulk task, a way to map the causal networks $\Gamma_{\mathbf{T}}$ to a different causal network $\Gamma_{\mathbf{\hat{T}}}.$
Often, as in section \ref{sec:GR-proof}, we can use the HRRT formula to associate an entanglement structure for a resource state that should be assigned to $\Gamma_{\mathbf{\hat{T}}}.$
So theorems like theorem \ref{thm:n-to-n-first-statement} hand us a pair of causal networks --- one with no pre-shared entanglement\footnote{Really, the bulk state has an amount of entanglement that is $O(1)$ in the holographic parameter, but this is insignificant compared to the boundary's extensive entanglement.} and one with a specific pattern of entanglement --- whose success probability we know to be related, and motivates us to study these network pairs purely in quantum information terms.

The next section is dedicated to undertaking that study for the specific network pairs produced by theorem \ref{thm:n-to-n-first-statement}.

\section{Quantum information and the connected wedge theorem}
\label{sec:QI-argument}

In this section, we use the quantum tasks framework introduced in the previous section to give a boundary argument for the $n$-to-$n$ connected wedge theorem.
In section \ref{sec:arbitrarytasks} (with details in appendix \ref{sec:bulkprotocolappendix}), we show that a connected $2$-to-all causal graph suffices to perform an arbitrary quantum task in the bulk, so long as the system size does not scale with the holographic parameter.
In section \ref{sec:B84taskinbulk}, we choose a particular task (the $\textbf{B}_{84}^{n, \delta}$ task) that can be performed in any bulk spacetime with a connected $2$-to-all causal graph even when the system size scales with the holographic parameter.
In section \ref{sec:final-tasks-argument}, we argue that implementing this task on the boundary requires the pattern of entanglement implied by the $n$-to-$n$ connected wedge theorem.

\subsection{Arbitrary tasks on a connected 2-to-all causal graph}\label{sec:arbitrarytasks}

The input to the $n$-to-$n$ connected wedge theorem, theorem \ref{thm:n-to-n-first-statement}, is connectedness of the 2-to-all causal graph $\Gamma_{2 \rightarrow \text{all}}.$
The output of the theorem is a particular large-N correlation structure for the boundary state.
As emphasized in the previous section, we can use quantum tasks reasoning to argue a ``causality $\Rightarrow$ correlations'' implication whenever the causal structure can be used to complete an appropriate task with high $p_{\text{suc}}.$
This motivates us to study which tasks can be completed with high success probability on any causal network whose 2-to-all causal graph is connected, even with no resource state.
In fact, we will show now that \textit{any} task of the type ``apply a quantum chanel to the input state'' can be completed on such a network with arbitrarily high success probability.
In the following two subsections, we discuss a specific task that can be implemented efficiently on such networks, and use this to provide a quantum information argument for theorem \ref{thm:n-to-n-first-statement}.

\begin{theorem} \label{thm:2-to-all-sufficiency}
    Let $\Gamma$ be a causal network with $n$ input vertices and $n$ output vertices, for which the associated 2-to-all causal graph $\Gamma_{2 \rightarrow \text{all}}$ is connected.
    If quantum systems $A_1, \dots, A_n$ are distributed among the input vertices, and $B_1, \dots, B_n$ are requested at the output vertices, then any channel $\mathbfcal{N}_{A_1....A_n\rightarrow B_1...B_n}$ can be implemented on $\Gamma$ with arbitrarily high $p_{\text{suc}}$, without the use of pre-shared entanglement.
\end{theorem}
\begin{remark}
    Informally, theorem \ref{thm:2-to-all-sufficiency} shows that an $m\rightarrow n$ causal connection does not allow more tasks to be completed than an appropriate set of $2\rightarrow n$ causal connections, even without pre-shared entanglement. 
\end{remark}
\begin{proof}
    We begin constructing a new causal network $\tilde{\Gamma}$.
    It will have one input vertex for each 2-to-all intermediary vertex in $\Gamma.$
    Its output vertices will be the output vertices of $\Gamma.$
    It has no intermediary vertices.
    Its edge structure will be very simple: there is one edge mapping each input to each output.

    Since $\Gamma_{2 \rightarrow \text{all}}$ is connected, every input system $A_j$ can be sent to \textit{some} $2$-to-all causal vertex in $\Gamma.$
    By picking an arbitrary $2$-to-all causal vertex for each input system and sending it there, we can map any task on $\Gamma$ to a task on $\tilde{\Gamma}.$
    If we were to add an all-to-all intermediary vertex to $\tilde{\Gamma},$ then we could certainly use this enhanced network to apply an arbitrary quantum channel $\mathbfcal{N}_{A_1 \dots A_n \rightarrow B_1 \dots B_n}.$
    But theorem \ref{thm:dolev} tells us that any task that can be done on the network ``$\tilde{\Gamma}$ plus an all-to-all intermediary vertex,'' can be performed on $\tilde{\Gamma}$ directly with sufficient preshared entanglement.

    The key observation is that agents stationed at the input vertices of $\Gamma$ can create entangled states and distribute them among any intermediary vertices of $\Gamma$ that they are able to signal.
    In particular, this allows them to distribute entanglement among certain pairs of input vertices of $\tilde{\Gamma}.$
    A convenient tool for understanding which pairs is the \textit{line graph} $L(\Gamma_{2 \rightarrow \text{all}})$: the graph whose vertices are the edges of $\Gamma_{2 \rightarrow \text{all}},$ and where two vertices in $L(\Gamma_{2 \rightarrow \text{all}})$ share an edge if the corresponding edges of $\Gamma_{2 \rightarrow \text{all}}$ share a vertex.
    The vertices of $L(\Gamma_{2 \rightarrow \text{all}})$ are the same as the input vertices of $\tilde{\Gamma},$ and agents stationed at the input vertices of $\Gamma$ can distribute arbitrary entangled states along the edges of $L(\Gamma_{2 \rightarrow \text{all}}).$
    
    A standard theorem in graph theory says that if $\Gamma_{2 \rightarrow \text{all}}$ is connected, then $L(\Gamma_{2 \rightarrow \text{all}})$ is connected.
    Connectedness of $L(\Gamma_{2 \rightarrow \text{all}})$ allows us to draw a path in $L(\Gamma_{2 \rightarrow \text{all}})$ that touches every vertex (possibly multiple times).
    This means agents starting at input vertices in $\Gamma$ can distribute, among the input vertices of $\tilde{\Gamma}$, an arbitrary number of EPR pairs that trace out a path that includes every input vertex of $\tilde{\Gamma}.$
    But in the proof of theorem \ref{thm:dolev} presented in \cite{dolev2019constraining}, this pattern of entanglement suffices to execute the protocol given there to replace intermediary vertices.

    This establishes that connectedness of $\Gamma_{2 \rightarrow \text{all}}$ allows agents in $\Gamma$ to provide the network $\tilde{\Gamma}$ with a resource state that is sufficient to perform the channel $\mathbfcal{N}_{A_1 \dots A_n \rightarrow B_1 \dots B_n}$ with arbitrarily high success probability.
\end{proof}

The last step in the proof, observing that the pattern of entanglement required by \cite{dolev2019constraining} is a chain of EPR pairs that touch every input vertex, may not be obvious upon consulting that reference.
For this reason, we provide a detailed explanation of the protocol in appendix \ref{sec:bulkprotocolappendix} that emphasizes this aspect.

\subsection{A special task on a 2-to-all causal graph}\label{sec:B84taskinbulk}

In the previous subsection, we showed that any quantum task can be completed with arbitrarily high success probability on a causal network with a connected 2-to-all causal graph.
However, this isn't the same as being able to perform the quantum task with arbitrarily high success probability in a bulk spacetime with the corresponding graph.
This is because the protocol outlined in the previous subsection requires generating a massive number of EPR pairs --- in fact, a number that is at least exponential in the size of the input systems $A_1 \dots A_n.$
When we argue for the $n$-to-$n$ connected wedge theorem using quantum tasks in subsection \ref{sec:final-tasks-argument}, we will want to take the size of the input systems to scale with the holographic parameter $\ell_{AdS}/G_N.$
Creating a number of particles that scales exponentially with this parameter would cause significant backreaction on spacetime, which would prevent us from thinking of the task as something that happens in a fixed background.\footnote{See \cite{may2022complexity} for a recent investigation of related issues.}
To get around this, we will now introduce a specific task that can be completed efficiently (i.e., using only linear-in-system-size entanglement) in any causal network with a connected 2-to-all causal graph.

Our task will be the same one used in \cite{may2019quantum, may2020holographic, may2021holographic}: the $\textbf{B}_{84}$ task discussed in \cite{kent2011quantum,buhrman2014position,lau2011insecurity}.
To define this task, we specify two input locations $c_{1}, c_{2}$ and two output locations $r_{1}, r_{2}$.
At $c_{1}$ a quantum system $Q$ in state $\mathbf{H}^q\ket{b}$ is given, where $b,q\in \{0,1\}$ and $H$ is the Hadamard operator.
At $c_{2}$ the classical bit $q$ is given.
The goal of the task is to deliver the bit $b$ to both $r_{1}$ and $r_{2}$.
If this is done correctly, we declare the task to be completed successfully. 

We also define the $n$-fold parallel repetition of the $\textbf{B}_{84}$ task, which we label $\textbf{B}_{84}^{n,\delta }$.
This repeated task consists of $n$ instances of the  $\textbf{B}_{84}$ task defined on the same input and output locations.
Each instance has independent and randomly chosen values of $q$ and $b$, so that the inputs are now $\{ \mathbf{H}^{q_j}\ket{b_j}\}$ at $c_{1}$ and $\{q_j\}$ at $c_{2}.$
We define the task to be completed successfully if a fraction $1-\delta$ of the output $b_j$ bits are returned correctly. 

Now, suppose we have a causal network $\Gamma$ with the same number of input and output vertices, whose $2$-to-all causal graph $\Gamma_{2\rightarrow \text{all}}$ is connected.
Let $c_1$ and $c_2$ be two of the input vertices of $\Gamma,$ and $r_1$ and $r_2$ be two of the output vertices of $\Gamma.$
We now provide a protocol for completing the $\textbf{B}_{84}^{n,\delta}$ task with high probability\footnote{By ``high probability,'' we mean a probability that converges exponentially to $1$ in the limit $n \rightarrow \infty.$} with no pre-shared entanglement, and needing only to generate a linear-in-$n$ amount of entanglement as part of the protocol.
The protocol uses the teleportation primitive discussed in appendix \ref{sec:NLQCtools}.

\begin{protocol} \label{protocol:B84inbulk} (\textbf{Bulk $\textbf{B}_{84}^{n,\delta}$})
Because the $2$-to-all causal graph is connected, there is a set of edges in $\Gamma_{2 \rightarrow \text{all}}$ that make up a path from $c_{1}$ to $c_{2}$.
Let the total number of vertices in this path be $m,$ relabel $c_2 \mapsto c_m,$ and label the vertices in the path in order as $c_1, \dots, c_m.$
For convenience, we will also relabel the output point $r_2$ to $r_m.$
    
Now:
    \begin{itemize}
        \item At the input points:
        \begin{enumerate}
            \item From $c_1$, send the inputs $\{\mathbf{H}^{q_j}\ket{b_j}_{Q_j}\}$ to the $2$-to-all vertex represented by the edge between $c_1$ and $c_2.$
            \item From $c_m$, send the inputs $\{q_j\}$ to the $2$-to-all vertex represented by the edge between $c_{m-1}$ and $c_m.$
            \item At each $c_k$ with $1<k<m$, prepare $n$ EPR pairs $\ket{\Phi^+}_{A_k B_k}^{\otimes n}$.
            Send system $A_k^{\otimes n}$ to the $2$-to-all vertex represented by the edge between $c_{k-1}$ and $c_k.$
            Send $B_k^{\otimes n}$ to the $2$-to-all vertex represented by the edge between $c_k$ and $c_{k+1}.$
        \end{enumerate}
        \item In the scattering regions:
        \begin{enumerate}
            \item At the vertex represented by the edge between $c_1$ and $c_2$, for each $j \in \{1, \dots, n\},$ measure the system\footnote{$A_{2, j}$ represents the $j$-th term in the tensor product $A_{2}^{\otimes n}.$} $Q_j A_{2, j}$ in the Bell basis.
            Send the measurement outcome to $r_1$ and $r_m$.
            \item At the vertices represented by edges between $c_{k}$ and $c_{k+1}$ for $1 < k < m-1$, for each $j \in \{1, \dots, n\},$ measure the system $B_{k, j} A_{k+1, j}$ in the Bell basis.
            Send the measurement outcome to $r_1$ and $r_m$.
            \item At the vertex represented by the edge between $c_{m-1}$ and $c_{m},$ for each $j \in \{1, \dots, n\},$ measure the system $B_{m-1, j}$ in the $\{\ket{0},\ket{1} \}$ basis if $q_j=0$, and in the Hadamard basis $\{\ket{+},\ket{-} \}$ if $q_j=1$.
            Send the measurement outcome to $r_1$ and $r_m$.
        \end{enumerate}
        \item At the output points:
        \begin{enumerate}
            \item Determine each $b_j$ from the measurement outcomes as described below, and return $b_j$. 
        \end{enumerate}
    \end{itemize}
    
    For a state on the system $S T_1 T_2,$ where $T_1 T_2$ is an EPR pair, the effect of performing a Bell measurement on $S T_1$ is to apply (up to phase) a Pauli operation on $S$ and recording the resulting state into $T_2$, with the specific Pauli operation determined by the classical measurement outcome.
    So performing a chain of Bell measurements on the systems $Q_j A_{2,j} B_{2, j} \dots A_{m-1, j} B_{m-1, j}$ is to apply a chain of up-to-phase Pauli operations on $Q_j \rightarrow B_{m-1, j}.$
    Since the composition of Pauli operations is Pauli, the result of these Bell measurements is some up-to-phase Pauli operation $Q_j \rightarrow B_{m-1, j},$ with the specific Pauli known as a function of all the measurement outcomes.
    So after all the Bell measurements have been performed, the final state on system $B_{m-1, j}$ is phase-related to a state
    \begin{equation}
        \mathbf{P}_{j} \mathbf{H}^{q_j} \ket{b_j}_{B_{m-1, j}},
    \end{equation}
    where the Pauli operator $P_j$ is known at both $r_1$ and $r_m,$ since these vertices have access to the outcomes of all Bell measurements.
    
    Assume $q_j = 0.$ Then $\mathbf{P}_j \mathbf{H}^{q_j} \ket{b_j}$ is a state in the computational basis, with the relationship to $b_j$ known as a function of the Pauli operator $\mathbf{P}_j.$
    So when this system is measured in the computational basis, we will obtain a bit whose relationship to $b_j$ is known in terms of $\mathbf{P}_j.$
    Since this bit and $\mathbf{P}_j$ are both known at $r_1$ and $r_m,$ the data at each of these points is sufficient to reconstruct $b_j.$
    
    Given $q_j = 1,$ the state $\mathbf{P}_j \mathbf{H}^{q_j} \ket{b_j}$ is a state in the Hadamard basis whose relationship to $b_j$ is known as a function of $\mathbf{P}_j.$
    So when this system is measured in the Hadamard basis, we will obtain a bit whose relationship to $b_j$ is known in terms of $\mathbf{P}_j.$
    Since this bit and $\mathbf{P}_j$ are known at $r_1$ and $r_m,$ the data at each of these points is sufficient to reconstruct $b_j.$
\end{protocol}
The total number of EPR pairs required by this protocol is $n (m-1).$

If the protocol is performed noiselessly, then it completes the task $\mathbf{B}_{84}^{n, \delta}$ for $\delta = 0$ with perfect success probability.
More physically, we can allow for some probability of error in each instance of the $\textbf{B}_{84}^{\, n,\delta }$ task.
Suppose $p_{\text{suc}}(\textbf{B}_{84})=1-\epsilon$.
Let $X_j$ be the random variable that takes value $1$ if the $j$-th parallel instance of the task succeeds, and $0$ if the $j$-th parallel instance fails. Then we have
\begin{equation}
    p_{\text{suc}}(\textbf{B}_{84}^{n, \delta})
        = 1 - \text{Prob}\left[\sum_j X_j < (1-\delta) n \right]
        \geq 1 - \text{Prob}\left[\sum_j X_j \leq (1-\delta) n\right].
\end{equation}
This can be rewritten in terms of the random variables $(1 - X_j)$ as
\begin{equation}
    p_{\text{suc}}(\textbf{B}_{84}^{n, \delta})
        \geq 1 - \text{Prob}\left[\sum_j (1 - X_j) \geq n \delta\right].
\end{equation}
But since the expected value of $(1 - X_j)$ is $\epsilon,$ theorem 1 of \cite{hoeffding1994probability} gives
\begin{equation}
    \text{Prob}\left[\sum_j (1 - X_j) \geq n \delta \right]
        \leq e^{-2 n (\delta - \epsilon)^2}
\end{equation}
whenever we have $\delta > \epsilon$, and hence
\begin{equation} \label{eq:final-nfold-lower-bound}
    p_{\text{suc}}(\textbf{B}_{84}^{n, \delta})
        \geq 1 - e^{-2 n (\delta - \epsilon)^2}.
\end{equation}
So for any probability of failure $\epsilon,$ we can choose a value of $\delta$ for which the success probability of the $\textbf{B}_{84}^{\, n,\delta }$ task converges exponentially to $1$ for large $n$.
    
\subsection{Tasks argument for the \texorpdfstring{$n$}{TEXT}-to-\texorpdfstring{$n$}{TEXT} connected wedge theorem}\label{sec:final-tasks-argument}

An important feature of the $\textbf{B}_{84}^{n, \delta}$ task is that for the causal network shown in figure \ref{fig:2-2repeatednlqc}, one can prove that a resource state with large mutual information is needed to complete it with high probability.
This analysis was undertaken in \cite{may2019quantum} based on results from \cite{tomamichel2013monogamy}, and some aspects of it were improved in \cite{may2020holographic, may2021holographic}.
We briefly recount the essential lemmas, adapted to the conventions of the present paper.

Let $\rho_{LR}$ be the resource state shared by the two input vertices of the causal network shown in figure \ref{fig:2-2repeatednlqc}.
The following lemma shows that $p_{\text{suc}}(\textbf{B}_{84}^{n, \delta})$ is small in the case $I(L:R)=0.$
\begin{lemma}\label{lemma:psucupperbound}
    Consider the $\mathbf{B}_{84}^{\, n, \delta}$ task in the causal network of figure \ref{fig:2-2repeatednlqc}, with $\delta \leq 1/2$, and where the input vertices share a state $\rho_{LR}$ with $I(L:R)_{\rho_{LR}}=0$.
    Then any strategy for completing the task has 
    \begin{align}
        p_{\text{suc}}(\mathbf{B}_{84}^{n,\delta}) \leq \left( 2^{h(\delta)}\cos^2(\pi/8) \right)^n = \left( 2^{h(\delta)}\beta \right)^n
    \end{align}
    where $h(\delta)$ is the binary entropy function $h(\delta)\equiv -\delta \log_2\delta -(1-\delta)\log_2(1-\delta)$ and the second equality defines the constant $\beta$. 
\end{lemma}
\begin{proof}
    This follows from section 3.2 of \cite{tomamichel2013monogamy}.
\end{proof}

Notice that for the success probability to go to zero at large $n$, we need $\delta$ to satisfy
\begin{align}
    h(\delta) < -\log_2\beta.
\end{align}
Numerically, this means $\delta\lesssim 0.037.$
We will call this critical value $\delta_*.$

\begin{lemma}[Mutual information lower bound] \label{lemma:MIlowerbound}
    Consider the causal structure shown in figure \ref{fig:2-2repeatednlqc}, with resource system $\rho_{LR}$. 
    Fix $\epsilon, \delta$ with $\epsilon < \delta < \delta_*$ and $2^{h(\delta)} \beta > e^{-2 (\delta - \epsilon)^2}.$
    Then completing the $\mathbf{B}_{84}^{n,\delta}$ task with probability $1-e^{- 2 n(\delta - \epsilon)^2}$ requires
    \begin{align}
        \frac{1}{2}I(L:R)_\rho \geq - n \log_2(2^{h(\delta)} \beta) - 1 + O\left( (2^{h(\delta)} \beta)^n, e^{-2 n (\delta - \epsilon)^2} (2^{h(\delta)} \beta)^{-n} \right).
    \end{align}
\end{lemma}
\begin{proof}
    See appendix \ref{app:lower-bound-proof}.
\end{proof}

Using these two features, we can now argue for the information-theoretic version of the $n$-to-$n$ connected wedge theorem:
\begin{theorem}[Information-theoretic $n$-to-$n$ connected wedge theorem] \label{thm:n-to-n-QI}
    Let $c_1, \dots, c_n$ and $r_1, \dots, r_n$ be points on a global AdS$_{2+1}$ boundary such that
    the sets $\mathcal{V}_j$ and $\mathcal{W}_j$ are nonempty, while also satisfying
    \begin{equation}
        \hat{J}(c_j, c_k \rightarrow r_1, \dots, r_n) = \hat{J}(c_1, \dots, c_n \rightarrow r_j, r_k) = \varnothing.
    \end{equation}
    Let $\mathcal{M}$ be an asymptotically AdS$_{2+1},$ AdS-hyperbolic spacetime satisfying the null curvature condition, with at least one global AdS$_{2+1}$ boundary on which $c_1, \dots, c_n, r_1, \dots, r_n$ are specified.
    Let $E_{\mathcal{V}_j}$ and $E_{\mathcal{W}_j}$ be the entanglement wedges of $\mathcal{V}_j$ and $\mathcal{W}_j.$ 
    
    If the $2$-to-all causal graph $\Gamma_{2\rightarrow\text{all}}$ is connected, then for any bipartition of the input regions $\mathcal{V}_1, \dots, \mathcal{V}_n,$ we have
    \begin{equation}
        I(\mathcal{V}_{j_1} \dots \mathcal{V}_{j_m} : \mathcal{V}_{j_{m+1}} \dots \mathcal{V}_{j_n}) = \Omega(\ell_{AdS}/G_N).
    \end{equation}
\end{theorem}

\begin{remark}
    The information-theoretic argument will only show extensive mutual information for \textit{contiguous} bipartitions respecting the angular ordering, i.e., bipartitions of the form $(\mathcal{V}_{j} \dots \mathcal{V}_{j+m-1} : \mathcal{V}_{j+m} \dots \mathcal{V}_{j - 1}).$
    In classical states of AdS$_3$/CFT$_2$, however, having $\Omega(\ell_{AdS}/G_N)$ mutual information for any contiguous bipartition is sufficient to show that the entanglement wedge of $\mathcal{V}_1 \cup \dots \cup \mathcal{V}_n$ is connected, which then implies that the mutual information across any bipartition is $\Omega(\ell_{AdS}/G_N).$ 
\end{remark}

\begin{argument}
Consider a contiguous bipartition of the input regions of the form appearing in the preceding remark.
For simplicity, we will relabel the input regions so that the bipartition can be written as $(\mathcal{V}_1 \dots \mathcal{V}_{m} : \mathcal{V}_{m+1} \dots \mathcal{V}_n).$
We will relabel the output regions accordingly, so that as in section \ref{sec:timelike-cylinders} the right-endpoint of $\mathcal{V}_j$ and the left-endpoint of $\mathcal{V}_{j+1}$ are cut off by light rays intersecting $r_j.$

Consider a $\textbf{B}_{84}$ task in the bulk whose two input points are in the entanglement wedges $E_{\mathcal{V}_m}$ and $E_{\mathcal{V}_{m+1}},$ and whose output points are in the entanglement wedges $\mathcal{W}_{1}$ and $\mathcal{W}_{m}.$
We choose these input points such that each one is in the past of one of the $2$-to-all regions regions needed to apply protocol \ref{protocol:B84inbulk}.
From inequality \eqref{eq:psucnewinequality}, we know that the success probability of $\textbf{B}_{84}$ in the boundary is lower bounded by the success probability of $\textbf{B}_{84}$ in any bulk spacetime that differs from $\mathcal{M}$ only in the state of bulk quantum fields within $\mathcal{V}_1 \cup \dots \cup \mathcal{V}_n.$
So we can modify the state of the quantum fields to produce EPR pairs within each $\mathcal{V}_j$ that can then be distributed among the $2$-to-all scattering regions as in protocol \ref{protocol:B84inbulk}.

In this modified state, we will assume that protocol \ref{protocol:B84inbulk} for a single instance of the $\textbf{B}_{84}$ task has failure probability $\epsilon < \delta_* \approx 3.7\%.$
We then choose some $\delta$ satisfying $\epsilon < \delta < \delta_*$ and $2^{h(\delta)} \beta > e^{-2 (\delta - \epsilon)^2}.$
Let $M$ be any number that is subleading in the holographic parameter $\ell_{\text{AdS}}/G_N,$ i.e. $M = o(\ell_\text{AdS}/G_N),$ so that $M$ particles of finite energy can be pass through AdS-scale regions of the bulk without causing backreaction.
Then by equation \eqref{eq:final-nfold-lower-bound}, there is a repetition task $\textbf{B}_{84}^{M, \delta}$ with the same input and output points as our $\textbf{B}_{84}$ task, which succeeds with probability at least $1 - e^{-2 M (\delta - \epsilon)^2}.$
Under the prescription of section \ref{sec:tasks-holography}, this task can be mapped to a boundary task $\hat{\textbf{B}}_{84}^{M, \delta}$ whose input and output locations are in $\mathcal{V}_m, \mathcal{V}_{m+1}, \mathcal{W}_1,$ and $\mathcal{W}_m.$
The causal network associated with this structure has two output vertices, labeled by $\mathcal{W}_1$ and $\mathcal{W}_{m},$ and four input vertices, labeled by $\mathcal{V}_{m},$ $\mathcal{V}_{m+1},$ and the two intermediary regions between them.
Because the input and output regions are causal diamonds, we have $J^+(\mathcal{V}_{j}) = J^+(c_j),$ and $J^-(\mathcal{W}_j) = J^-(r_j).$
Lemma \ref{lemma:boundarycausalstructure} therefore tells us that the causal network has an intermediary vertex of the form $\mathcal{V}_m, \mathcal{V}_{m+1} \rightarrow \mathcal{W}_0,$ and an intermediary vertex of the form $\mathcal{V}_m, \mathcal{V}_{m+1} \rightarrow \mathcal{W}_m,$ but no $2$-to-$2$ vertex of the form $\mathcal{V}_m, \mathcal{V}_{m+1} \rightarrow \mathcal{W}_0, \mathcal{W}_m.$
If the only vertices in the causal network were the input regions, the output regions, and these intermediary vertices, then the causal network would be of the form in figure \ref{fig:2-2repeatednlqc}.
This would allow us to conclude, via lemma \ref{lemma:MIlowerbound}, that the mutual information was lower bounded by some constant times $M$ for sufficiently large $M$.
Since $M$ could be anything subleading in $\ell_{\text{AdS}}/G_N,$ this would allow us to conclude $I(\mathcal{V}_m : \mathcal{V}_{m+1}) = \Omega(\ell_{\text{AdS}}/G_N).$

However, one can construct examples of the $n$-to-$n$ connected wedge theorem where all bipartitions have extensive mutual information, while adjacent pairs $\mathcal{V}_m : \mathcal{V}_{m+1}$ have only $O(1)$ mutual information.
The flaw in the logic given above is that the causal network has more vertices than just the input and output regions; in particular, the input vertices in the causal network corresponding to intermediary regions between $\mathcal{V}_m$ and $\mathcal{V}_{m+1}$ cannot be completely ignored.
We do not have a rigorous or complete understanding of how these intermediary regions contribute to the strategies that can be used to complete $\mathbf{\hat{B}}_{84}^{M, \delta}$ in the boundary theory.\footnote{This gap in our knowledge is why this is an ``argument'' rather than a ``proof.''. See also \cite{dolev2022holography} who discuss the role of these intermediate regions.}
In fact in completely general states these regions can be used to remove the need for bipartition entanglement in performing the $\hat{\mathbf{B}}_{84}$ task\footnote{See for example appendix B of \cite{may2020holographic} or \cite{dolev2022holography}.}
Nonetheless it is clear that these intermediate regions are distinguished in an important way from the input regions --- any states appearing there or operations performed are independent of the inputs to the task. 
We will \textit{assume} that $p_{\text{suc}}(\hat{\textbf{B}}_{84}^{M, \delta})$ is unaffected by replacing the causal network for input and output regions $\mathcal{V}_m, \mathcal{V}_{m+1}, \mathcal{W}_1, \mathcal{W}_m$ with the causal network for input and output regions $\mathcal{V}_1, \dots, \mathcal{V}_n, \mathcal{W}_1, \dots, \mathcal{W}_n$, and then deleting from this network all input vertices representing regions $\mathcal{X}_1$ through $\mathcal{X}_n.$
As we will see, this assumption reproduces the connected wedge theorem for all $n$. 
We discuss issues around removing or better justifying this assumption in the discussion.

The final essential observation for our argument is that the success probability of a task does not decrease upon fusing any two input vertices of the network.
If all of the entanglement resources and causal connections available to vertex $q_1$ are combined with all the entanglement resources and causal connections available to vertex $q_2,$ the resulting causal network permits all strategies permitted by the original one.
So we will fuse the input vertices $\mathcal{V}_1, \dots, \mathcal{V}_m$ and the input vertices $\mathcal{V}_{m+1}, \dots, \mathcal{V}_n.$
The resulting causal network has only two input vertices, one of which has access to the CFT state on $\mathcal{V}_1 \cup \dots \cup \mathcal{V}_m,$ and one of which has access to the CFT state on $\mathcal{V}_{m+1} \cup \dots \cup \mathcal{V}_n.$
Because this network was obtained by fusion, the success probability of $\hat{\textbf{B}}_{84}^{M, \delta}$ still exceeds $1 - e^{-2 M (\delta - \epsilon)^2}.$

Clearly both input vertices in this network are causally connected to both output vertices.
But we can apply lemma \ref{lemma:boundarycausalstructure} to show that there is no $2$-to-$2$ causal connection in the resulting network.
For any $\mathcal{V}_j, \mathcal{V}_k$, lemma \ref{lemma:boundarycausalstructure} says that any point that can be signalled by both $\mathcal{V}_j$ and $\mathcal{V}_k$ can either signal the set $\{\mathcal{W}_1, \dots, \mathcal{W}_{j-1}, \mathcal{W}_k, \dots, \mathcal{W}_n\}$ or the set $\{\mathcal{W}_j, \dots, \mathcal{W}_{k-1}\},$ but not both.
So after fusion, we end up with a causal network of the form in figure \ref{fig:2-2repeatednlqc}.
Lemma \ref{lemma:MIlowerbound} then tells us that $I(\mathcal{V}_1 \dots \mathcal{V}_m : \mathcal{V}_{m+1} \dots \mathcal{V}_{n})$ is lower-bounded for large $M$ by a linear function of $M$, and since $M$ was any number subleading in $\ell_{\text{AdS}}/G_N,$ we may conclude
\begin{equation}
    I(\mathcal{V}_1 \dots \mathcal{V}_m : \mathcal{V}_{m+1} \dots \mathcal{V}_{n}) = \Omega(\ell_{\text{AdS}}/G_N).
\end{equation}
\end{argument}

\section{Holographic scattering from quantum error-correction}\label{sec:boundaryprotocols}

The connected wedge theorem reveals that a particular pattern of entanglement exists in the boundary given a class of bulk causal structures. 
We showed in section \ref{sec:QI-argument} that the specified bulk causal structure allows arbitrary quantum tasks to be completed in the boundary, and argued that the associated boundary entanglement is necessary for the boundary to reproduce this bulk physics. 
A remaining problem is to demonstrate that the given pattern of boundary entanglement suffices to reproduce the bulk computation.
In this section, we begin addressing this question by constructing a boundary protocol to implement arbitrary Clifford unitary operators among $n=3$ parties.
This allows in particular a version of the $\mathbf{B}_{84}$ task to be completed in the boundary. 

\subsection{Protocol overview}

Recall that in the case of two input and two output points, the connected wedge theorem gives
\begin{align} \label{eq:n2-pattern}
    I(\mathcal{V}_1:\mathcal{V}_2) = \Omega(\ell_{\text{AdS}}/G_N).
\end{align}
This suggests modeling the boundary resource system with EPR pairs shared between these two regions. 
One can use these EPR pairs to complete an arbitrary quantum task in the boundary by exploiting a combination of Bell basis teleportation and port-teleportation \cite{buhrman2014position,beigi2011simplified}.\footnote{Note that the number of EPR pairs required to complete an arbitrary task using the best known protocols is much higher than the number required to model equation \eqref{eq:n2-pattern}; see \cite{may2020holographic} for a discussion. Here, we think coarsely about patterns of entanglement, not specific quantities.}
Naively generalizing these protocols to $n\geq 3$ leads to protocols that require entanglement between pairs of input regions \cite{dolev2019constraining}, whereas the connected wedge theorem only guarantees entanglement across bipartitions of inputs regions.
More concretely, for $n=3$ input and output points, the connected wedge theorem implies
\begin{align} \label{eq:n3-pattern}
    I(\mathcal{V}_1:\mathcal{V}_2\mathcal{V}_3), I(\mathcal{V}_2:\mathcal{V}_1\mathcal{V}_3),I(\mathcal{V}_3:\mathcal{V}_2\mathcal{V}_1)=\Omega(\ell_{AdS}/G_N),
\end{align}
but allows for e.g. $I(\mathcal{V}_1:\mathcal{V}_2)=\Theta(1)$.
In this section we develop a new protocol which uses only the pattern of entanglement guaranteed by equation \eqref{eq:n3-pattern}.

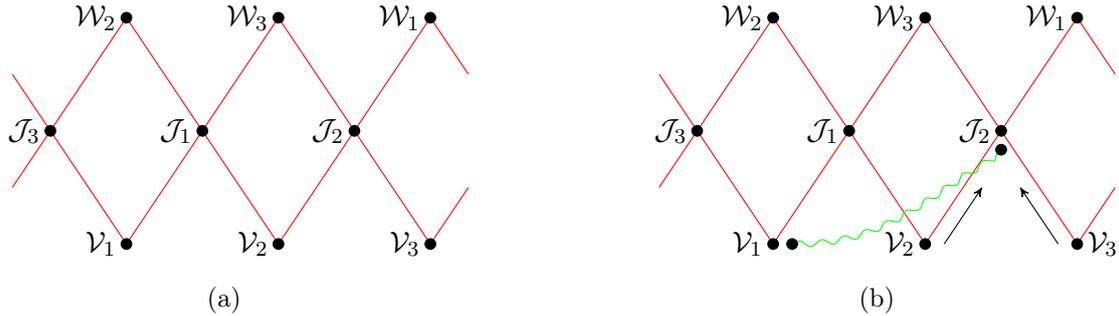
\begin{figure}
\centering
\subfloat[\label{fig:3to3structure}]{
\begin{tikzpicture}[scale=1]

\draw[red] (-1,1.5) -- (0,3);
\draw[red] (1,1.5) -- (0,3);
\draw[red] (1,1.5) -- (2,3);
\draw[red] (-1,1.5) -- (-2,3);
\draw[red] (-3,1.5) -- (-2,3);

\draw[red] (-2,0) -- (-1,1.5);
\draw[red] (-2,0) -- (-3,1.5);
\draw[red] (0,0) -- (-1,1.5);
\draw[red] (0,0) -- (1,1.5);
\draw[red] (2,0) -- (1,1.5);

\draw[red] (2,0) -- (2.5,0.75);
\draw[red] (2,3) -- (2.5,2.25);

\draw[red] (-3,1.5) -- (-3.5,2.25);
\draw[red] (-3,1.5) -- (-3.5,0.75);

\draw plot [mark=*, mark size=2] coordinates{(0,0)};
\node[left] at (0,0) {$\mathcal{V}_2$};

\draw plot [mark=*, mark size=2] coordinates{(-2,0)};
\node[left] at (-2,0) {$\mathcal{V}_1$};

\draw plot [mark=*, mark size=2] coordinates{(2,0)};
\node[left] at (2,0) {$\mathcal{V}_3$};

\draw plot [mark=*, mark size=2] coordinates{(-1,1.5)};
\node[left] at (-1,1.5) {$\mathcal{J}_1$};

\draw plot [mark=*, mark size=2] coordinates{(1,1.5)};
\node[left] at (1,1.5) {$\mathcal{J}_2$};

\draw plot [mark=*, mark size=2] coordinates{(-3,1.5)};
\node[left] at (-3,1.5) {$\mathcal{J}_3$};

\draw plot [mark=*, mark size=2] coordinates{(0,3)};
\node[left] at (0,3) {$\mathcal{W}_3$};

\draw plot [mark=*, mark size=2] coordinates{(-2,3)};
\node[left] at (-2,3) {$\mathcal{W}_2$};

\draw plot [mark=*, mark size=2] coordinates{(2,3)};
\node[left] at (2,3) {$\mathcal{W}_1$};

\end{tikzpicture} 
}
\hfill
\subfloat[\label{fig:3to3structurewithEPRpair}]{
\begin{tikzpicture}

\draw[red] (-1,1.5) -- (0,3);
\draw[red] (1,1.5) -- (0,3);
\draw[red] (1,1.5) -- (2,3);
\draw[red] (-1,1.5) -- (-2,3);
\draw[red] (-3,1.5) -- (-2,3);

\draw[red] (-2,0) -- (-1,1.5);
\draw[red] (-2,0) -- (-3,1.5);
\draw[red] (0,0) -- (-1,1.5);
\draw[red] (0,0) -- (1,1.5);
\draw[red] (2,0) -- (1,1.5);

\draw[red] (2,0) -- (2.5,0.75);
\draw[red] (2,3) -- (2.5,2.25);

\draw[red] (-3,1.5) -- (-3.5,2.25);
\draw[red] (-3,1.5) -- (-3.5,0.75);

\draw plot [mark=*, mark size=2] coordinates{(0,0)};
\node[left] at (0,0) {$\mathcal{V}_2$};

\draw plot [mark=*, mark size=2] coordinates{(-2,0)};
\node[left] at (-2,0) {$\mathcal{V}_1$};

\draw plot [mark=*, mark size=2] coordinates{(2,0)};
\node[right] at (2,0) {$\mathcal{V}_3$};

\draw plot [mark=*, mark size=2] coordinates{(-1,1.5)};
\node[left] at (-1,1.5) {$\mathcal{J}_1$};

\draw plot [mark=*, mark size=2] coordinates{(1,1.5)};
\node[left] at (1,1.5) {$\mathcal{J}_2$};

\draw plot [mark=*, mark size=2] coordinates{(-3,1.5)};
\node[left] at (-3,1.5) {$\mathcal{J}_3$};

\draw plot [mark=*, mark size=2] coordinates{(0,3)};
\node[left] at (0,3) {$\mathcal{W}_3$};

\draw plot [mark=*, mark size=2] coordinates{(-2,3)};
\node[left] at (-2,3) {$\mathcal{W}_2$};

\draw plot [mark=*, mark size=2] coordinates{(2,3)};
\node[left] at (2,3) {$\mathcal{W}_1$};

\draw[green,decorate,decoration={snake, segment length=3mm, amplitude=0.4mm}] (-1.75,0) to [out=0,in=-150] (1,1.25);
\draw plot [mark=*, mark size=2] coordinates{(-1.75,0)};
\draw plot [mark=*, mark size=2] coordinates{(1,1.25)};

\draw[->] (0.25,0) -- (0.75,0.75);
\draw[->] (1.75,0) -- (1.25,0.75);

\end{tikzpicture} 
}
\caption{(a) The boundary causal structure for $n=3$. Signals from any two input regions meet at intermediate regions and scatter to two output regions. (b) We exploit a resource state which has the property that systems from $\mathcal{V}_2$ and $\mathcal{V}_3$ can be collected to prepare an EPR pair with $\mathcal{V}_1$. This is consistent with the condition that $I(\mathcal{V}_1:\mathcal{V}_2\mathcal{V}_3)$ is large, given by the connected wedge theorem. 
}
\end{figure}

Begin by recalling the boundary causal structure for $n=3$, shown in figure \ref{fig:3to3structure}. 
Note that we do not include the $\mathcal{X}_j$ regions in this causal structure. 
A complete description of the boundary should include these, and the existence of the bulk only implies arbitrary computations can be carried out in the causal structure that includes the $\mathcal{X}_j$ regions. 
Here we design protocols for the simpler causal structure without them; we believe that studying the simpler structure already captures interesting aspects of boundary dynamics. 
As well, we believe a generalization of our protocol allows arbitrary tasks to be completed in the structure without the $\mathcal{X}_j$ regions, with only the given entanglement.

There are three layers to the boundary causal network of figure \ref{fig:3to3structure}: the input, intermediate, and output layers. 
Current understanding of the AdS/CFT correspondence shows that degrees of freedom located in the bulk are holographically encoded into boundary degrees of freedom in a quantum error correcting code (QECC) \cite{almheiri2015bulk,dong2016reconstruction,harlow2017ryu,akers2019large,hayden2019learning,akers2021leading,akers2022quantum}. 
Given this, and the structure of our network, we consider protocols of the following form. 
\begin{enumerate}[(1)]
\item Encoding: The time-evolution and signalling from the input regions $\mathcal{V}_{j}$ to intermediate regions $\mathcal{J}_{j}$ will encode initial states into the codeword subspace of some QECC. 
\item Logical operation: The interactions on the intermediate layer will then occur inside the codeword subspace at intermediate regions $\mathcal{J}_{j}$ by implementing logical operations.
\item Decoding: Codewords of the QECC are decoded by the time-evolution and signalling to the output locations $\mathcal{W}_{j}$, completing the holographic scattering task.
\end{enumerate}
Importantly, the causal structure of the boundary network is such that the encoding step in these protocols is not possible without pre-distributed entanglement among the input regions. 
The connected wedge theorem guarantees exactly the pattern of entanglement needed to perform this encoding step.  

In the encoding step, we hope to encode input systems from $\mathcal{V}_{1},\mathcal{V}_{2},\mathcal{V}_{3}$ into a QECC at intermediate regions $\mathcal{J}_{1},\mathcal{J}_{2},\mathcal{J}_{3}$, defined by
\begin{align}
    \mathcal{J}_{1} &\equiv \hat{J}(\mathcal{V}_1,\mathcal{V}_2\rightarrow \mathcal{W}_2,\mathcal{W}_3),\\
    \mathcal{J}_{2} &\equiv \hat{J}(\mathcal{V}_2,\mathcal{V}_3\rightarrow \mathcal{W}_3,\mathcal{W}_1),\\
    \mathcal{J}_{3} &\equiv \hat{J}(\mathcal{V}_3,\mathcal{V}_1\rightarrow \mathcal{W}_1,\mathcal{W}_2).
\end{align}
Focusing on the input system $A_1$ at $\mathcal{V}_1$, notice that $\mathcal{V}_1$ can directly signal only to $\mathcal{J}_{1}$ and $\mathcal{J}_{3}$ but not $\mathcal{J}_2$.
However, we wish to encode $A_1$ into three shares, with one share sent to each of $\mathcal{J}_1$, $\mathcal{J}_2$ and $\mathcal{J}_3$.
Similar issues exist for the remaining two input locations. 
In order to achieve a non-local encoding over $\mathcal{J}_{1},\mathcal{J}_{2},\mathcal{J}_{3}$, we will need to utilize pre-existing entanglement. 
Observe that $\mathcal{V}_1$ is entangled with $\mathcal{V}_2 \cup \mathcal{V}_3$, and both $\mathcal{V}_2$ and $\mathcal{V}_3$ are in the past of $\mathcal{J}_2$. This suggests the following strategy.
Send resource systems from $\mathcal{V}_2$ and $\mathcal{V}_3$ into $\mathcal{J}_2$ which together suffice to produce an EPR pair between $\mathcal{V}_1$ and $\mathcal{J}_2$. 
Then use this entanglement to produce an error correcting code with shares distributed among $\mathcal{J}_1,\mathcal{J}_2, \mathcal{J}_3$, even though $\mathcal{J}_2$ is outside the causal future of $\mathcal{V}_1$.
See figure \ref{fig:3to3structurewithEPRpair}. 

In order for the first step in this protocol to work --- producing an EPR pair between $\mathcal{V}_1$ and $\mathcal{J}_2$ --- we must assume that the resource state shared by $\mathcal{V}_1, \mathcal{V}_2,$ and $\mathcal{V}_3$ can be converted to an EPR pair by acting locally on systems $\mathcal{V}_1$ and $\mathcal{V}_2 \cup \mathcal{V}_3.$
This is not guaranteed by the condition $I(\mathcal{V}_1 : \mathcal{V}_2 \mathcal{V}_3) = \Omega(\ell_{\text{AdS}}/G_N).$
However, it seems reasonable to expect that this can be done in holography --- pure-gravity entanglement is highly compressible \cite{bao2019beyond}, so any two regions that share mutual information at order $\Omega(\ell_{\text{AdS}}/G_N)$ should be well approximated by an appropriate number of EPR pairs.
A simple example of a three party state that can be used as an appropriate resource for our protocol is
\begin{align}
    \Phi = \frac{1}{4}\sum_{k} (\mathbf{I}\otimes \mathbf{P}_k)\Psi^+_{AB}(\mathbf{I}\otimes \mathbf{P}_k^{\dagger}) \otimes \ketbra{k}{k}_C,
\end{align}
where $\mathbf{P}_k$ is an appropriately generalized Pauli operator and $\Psi^+$ is an EPR pair. Notice all the mutual informations between pairs of subsystems are zero, but $B$ and $C$ collected together suffice to produce an EPR pair shared with $A$. 

The second part of the protocol requires using an EPR pair shared between $\mathcal{V}_1$ and $\mathcal{J}_2$ to encode the input systems into a QECC.
To understand how to do this, we recall the notion of an entanglement-assisted quantum error correcting code (EAQECC). 
Typically, we consider encoding unitaries that act on all subsystems. 
For instance, in a code storing one qudit into three, an encoding map has the general form
\begin{align}
    \ket{x}_L \rightarrow \ket{x}_P = \mathbf{U}_{P_1P_2P_3}\ket{x,0,0}_{P_1P_2P_3}.
\end{align}
However, by replacing the ancilla state $\ket{0,0}_{P_2 P_3}$ with an entangled state, we can have an encoding unitary that acts more locally. 
In particular, consider the three qudit code, which stores one $d$ dimensional qudit into three for $d\geq 3$. 
The code words are of the form
\begin{align}
|x\rangle \quad \rightarrow \quad |\widetilde{x}\rangle = \frac{1}{\sqrt{p}}\sum_{y=0}^{p-1}|y, y+x,y+2x\rangle.
\end{align}
As we discuss in appendix \ref{sec:CodeConstructions}, these code words can be prepared by acting on only two subsystems and sharing an EPR state. 
In particular there exist unitaries $\mathbf{U}_{P_1P_2}$, $\mathbf{U}_{P_2P_3}$, $\mathbf{U}_{P_1P_3}$ such that
\begin{align}
|\widetilde{x}\rangle = \mathbf{U}_{P_1P_2}|x\rangle_{P_1} \otimes |\text{EPR}\rangle_{P_2P_3}= \mathbf{U}_{P_2P_3}|x\rangle_{P_2} \otimes |\text{EPR}\rangle_{P_3P_1} = 
 \mathbf{U}_{P_3P_1}|x\rangle_{P_3} \otimes |\text{EPR}\rangle_{P_1P_2}.\nonumber
\end{align}
Codes prepared in this way are entanglement-assisted quantum codes.\footnote{Note that EAQECC's have appeared elsewhere in holography, perhaps most famously in the Hayden-Preskill thought experiment \cite{hayden2007black}.} See figure \ref{fig:EAQECCencoding}. 

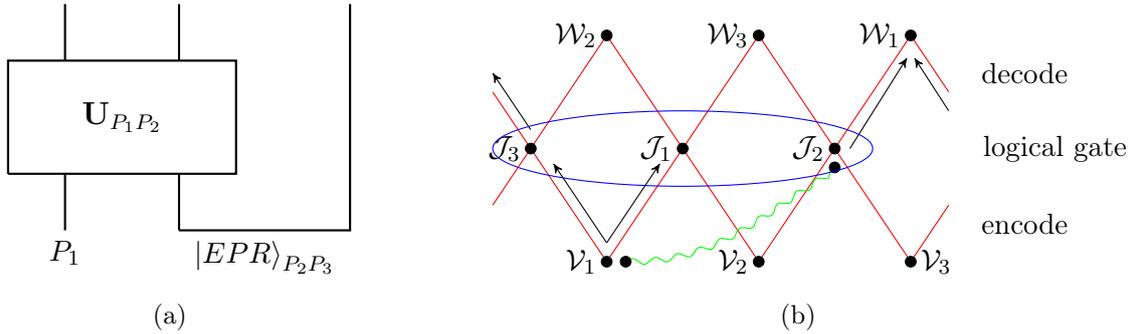
\begin{figure}
\centering
\subfloat[\label{fig:EAQECCencoding}]{
\begin{tikzpicture}[scale=0.75]

\draw[thick] (0,0) -- (4,0) -- (4,2) -- (0,2) -- (0,0);
\node at (2,1) {\large{$\mathbf{U}_{P_1P_2}$}};

\draw[thick] (1,-1) -- (1,0);
\node[below] at (1,-1) {$P_1$};

\draw[thick] (3,-1) -- (3,0);
\draw[thick] (3,-1) -- (6,-1) -- (6,3);
\node[below] at (4.5,-1) {$\ket{EPR}_{P_2P_3}$};

\draw[thick] (1,2) -- (1,3);
\draw[thick] (3,2) -- (3,3);

\end{tikzpicture}
}
\hfill
\subfloat[\label{fig:EAQECsharedistribution}]{
\begin{tikzpicture}

\draw[red] (-1,1.5) -- (0,3);
\draw[red] (1,1.5) -- (0,3);
\draw[red] (1,1.5) -- (2,3);
\draw[red] (-1,1.5) -- (-2,3);
\draw[red] (-3,1.5) -- (-2,3);

\draw[red] (-2,0) -- (-1,1.5);
\draw[red] (-2,0) -- (-3,1.5);
\draw[red] (0,0) -- (-1,1.5);
\draw[red] (0,0) -- (1,1.5);
\draw[red] (2,0) -- (1,1.5);

\draw[red] (2,0) -- (2.5,0.75);
\draw[red] (2,3) -- (2.5,2.25);

\draw[red] (-3,1.5) -- (-3.5,2.25);
\draw[red] (-3,1.5) -- (-3.5,0.75);

\draw plot [mark=*, mark size=2] coordinates{(0,0)};
\node[left] at (0,0) {$\mathcal{V}_2$};

\draw plot [mark=*, mark size=2] coordinates{(-2,0)};
\node[left] at (-2,0) {$\mathcal{V}_1$};

\draw plot [mark=*, mark size=2] coordinates{(2,0)};
\node[right] at (2,0) {$\mathcal{V}_3$};

\draw plot [mark=*, mark size=2] coordinates{(-1,1.5)};
\node[left] at (-1,1.5) {$\mathcal{J}_1$};

\draw plot [mark=*, mark size=2] coordinates{(1,1.5)};
\node[left] at (1,1.5) {$\mathcal{J}_2$};

\draw plot [mark=*, mark size=2] coordinates{(-3,1.5)};
\node[left] at (-3,1.5) {$\mathcal{J}_3$};

\draw plot [mark=*, mark size=2] coordinates{(0,3)};
\node[left] at (0,3) {$\mathcal{W}_3$};

\draw plot [mark=*, mark size=2] coordinates{(-2,3)};
\node[left] at (-2,3) {$\mathcal{W}_2$};

\draw plot [mark=*, mark size=2] coordinates{(2,3)};
\node[left] at (2,3) {$\mathcal{W}_1$};

\draw[green,decorate,decoration={snake, segment length=3mm, amplitude=0.4mm}] (-1.75,0) to [out=0,in=-150] (1,1.25);
\draw plot [mark=*, mark size=2] coordinates{(-1.75,0)};
\draw plot [mark=*, mark size=2] coordinates{(1,1.25)};

\node at (3.5,0.5) {encode};
\node at (3.9,1.5) {logical gate};
\node at (3.5,2.5) {decode};

\draw[->] (-2,0.25) -- (-1.3,1.3);
\draw[->] (-2,0.25) -- (-2.7,1.3);

\draw[->] (-3,1.75) -- (-3.5,2.5);
\draw[->] (2.5,2) -- (2.05,2.7);
\draw[->] (1.2,1.5) -- (1.95,2.7);

\draw[blue] (-1,1.5) ellipse (2.5 and 0.5);

\end{tikzpicture} 
}
\caption{(a) Structure of an encoding in an EAQECC (b) Encoding of $A_1$, which begins at $\mathcal{V}_1$, into $\mathcal{J}_{1},\mathcal{J}_{2},\mathcal{J}_{3}$ by using the EPR pair between $\mathcal{V}_{1}$ and $\mathcal{J}_{2}$, followed by recovery at $\mathcal{W}_1$.
}
\label{fig_EAQECC_12}
\end{figure}

Returning to the description of our protocol, label the input system at $\mathcal{V}_i$ to be $A_i$.
The input systems will be in some joint state $\ket{\Psi}_{RA_1A_2A_3}$, with $R$ a reference system. 
At each $\mathcal{V}_i$, encode $A_i$ into a three qudit code, using the entanglement between $\mathcal{V}_i$ and $\mathcal{J}_{i+1}$. 
Then, the two shares produced locally at $\mathcal{V}_i$ are sent to $\mathcal{J}_{i-1}$ and $\mathcal{J}_i$ so that the three subsystems of the code are each brought to one of the three scattering regions. 
See figure \ref{fig:EAQECsharedistribution}. 

Using this procedure, the state held at the intermediate regions will consist of three copies of the three qudit code, each holding one of the $A_i$ in its logical Hilbert space. 
At this stage, we are able to apply unitary operators of the form
\begin{align}\label{eq:transverseunitary}
    \mathbf{V}_{\mathcal{J}_1\mathcal{J}_2\mathcal{J}_3} = \mathbf{V}_{\mathcal{J}_1}\otimes \mathbf{V}_{\mathcal{J}_2}\otimes \mathbf{V}_{\mathcal{J}_3}.
\end{align}
Here $\mathbf{V}_{\mathcal{J}_i}$ acts jointly on any subsystems in the $\mathcal{J}_i$ region, which here will be three shares, one from each of the three codes. 
Operators of this form, which act as a tensor product across different shares within each code, are known as transverse operations.\footnote{A similar notion of transversal operations, which involves acting across copies of a code, will appear in \cite{dolev2022holography} in a similar context.}
As we discuss in appendix \ref{sec:CodeConstructions}, at least the Clifford operations can be implemented transversally in the three qudit code. 

With the $A_i$ systems transformed by the chosen Clifford unitary, our remaining goal is to produce the $A_i$ system at output location $\mathcal{W}_i$. 
Notice that location $\mathcal{W}_i$ has $\mathcal{J}_{i+1}$ and $\mathcal{J}_{i-1}$ in its past. 
These two shares can be brought to $\mathcal{W}_i$, and used to recover $A_i$, completing the protocol. 

\subsection{Protocol for variant \texorpdfstring{$\mathbf{B}_{84}$}{TEXT}}\label{sec:boundaryB84}

To ensure that our protocol above is non-trivial, we should check that the entanglement it is using is actually necessary. 
To do this, we briefly point out that a variant on the $\mathbf{B}_{84}$ task, which provably requires bipartition entanglement according to the argument given in section \ref{sec:final-tasks-argument}, can be completed using the protocol given.

Our variant $\mathbf{B}_{84}$ task has inputs
\begin{align}
    \mathbf{H}^q \ket{b}_{A_1},\ket{q}_{A_2},\ket{q}_{A_3}
\end{align}
where each $A_i$ is a qutrit, $b,q\in \{0,1,2\}$, and $\mathbf{H}$ is the Hadamard operator for qutrits, 
\begin{align}
\mathbf{H} = \frac{1}{\sqrt{3}}\begin{bmatrix}
1& 1& 1  \\
1& \omega& \omega^2  \\ 
1& \omega^2& \omega 
\end{bmatrix}.
\end{align}
To understand how to complete this task in the boundary, notice that since $q$ is available at $\mathcal{V}_2$ and $\mathcal{V}_3$, it can be sent to all of the intermediate regions $\mathcal{J}_1,\mathcal{J}_2,\mathcal{J}_3$. 
At the intermediate points then, we know $q$, and so know the identity of the logical operation $\mathbf{H}^q$ that needs to be implemented. 
Since $\mathbf{H}^q$ is Clifford, it can be implemented by acting transversally with an operation of the form in equation \eqref{eq:transverseunitary}.
The logical state is now 
\begin{align}
    \ket{\tilde{b},\tilde{q},\tilde{q}}_{A_1A_2A_3}.
\end{align}
Next, apply $\mathbf{X}^{3-q}$ to the (logical) $A_2$ and $A_3$ systems, resetting them to zero. 
Further, apply two $\mathbf{CNOT}$ gates with $A_1$ as the control and $A_2$ and $A_3$ as targets, producing
\begin{align}
     \ket{\tilde{b},\tilde{b},\tilde{b}}
\end{align}
Both the $\mathbf{H}^q$ and $\mathbf{CNOT}$ operations are Clifford, so can be carried out transversally.
Finally, two shares of the code storing $A_1$ can be brought to $\mathcal{W}_1$, and similarly for $A_2$ and $A_3$, allowing the value of $b$ to be decoded at each output location. 

Notice that the task defined here differs from the $\mathbf{B}_{84}$ task defined in section \ref{sec:B84taskinbulk} in two ways. 
First, we have moved from qubits to qutrits. 
This is convenient because the three qudit code is defined only for $d\geq 3$, but isn't essential --- we could instead for instance use the 7 qubit code \cite{Steane} to complete the task for qubits, since it has transverse Clifford operations.  
Second, we have the basis information $q$ input at two input locations, whereas the earlier task it was given at just one. 
This is an important change because if $q$ is given at just one location, we cannot complete this task using a Clifford unitary. 
Instead, we would need to implement a controlled-$\mathbf{H}$ operation, which is non-Clifford. 

After making these changes, the proof in section \ref{sec:final-tasks-argument} that bipartition entanglement is necessary still applies. 
To adapt the proof to qutrits, one can exploit the more general bounds on unentangled success probabilities given in \cite{tomamichel2013monogamy} (we used a special case for qubits to prove lemma \ref{lemma:psucupperbound}). 
To adapt to the case where $q$ is revealed at two input locations, notice that in the proof of theorem \ref{thm:n-to-n-QI} we fused all input locations together aside from the one holding the input quantum system. 
Running that argument for this task, the fused region would obtain multiple copies of $q$, but this doesn't change the success probability. 

\section{Discussion}

We conclude by discussing some supplementary results about the conclusions in the main text, and open questions.

\vspace{0.2cm}
\noindent \textbf{Scattering regions are inside the entanglement wedge}
\vspace{0.2cm}

In the setting of the $2$-to-$2$ connected wedge theorem, it was observed in \cite{may2020holographic,may2021holographic} that the scattering region $J(E_{\mathcal{V}_1},E_{\mathcal{V}_2}\rightarrow E_{\mathcal{W}_1},E_{\mathcal{W}_2})$ is inside of $E_{\mathcal{V}_1\cup \mathcal{V}_2}$. 
This observation generalizes to the setting of the $n$-to-$n$ connected wedge theorem. 
In particular, when the connected wedge theorem holds, each of the $J(E_{\mathcal{V}_j},E_{\mathcal{V}_k}\rightarrow E_{\mathcal{W}_1},\dots,E_{\mathcal{W}_n})$ are inside of $E_{\mathcal{V}}$. 
We show this below. 

Denote by $\mathcal{X}_j$ the causal diamonds in the spacelike complement of $\mathcal{V} \equiv \mathcal{V}_1 \cup \dots \cup \mathcal{V}_n.$
Specifically, order them so that $\mathcal{X}_j$ is immediately to the right of $\mathcal{V}_j,$ as in figure \ref{fig:X-regions}.
Define the boundary regions $\mathcal{X}=\cup_j \mathcal{X}_j$, and $\mathcal{D}_j = D(\mathcal{V} \cup (\mathcal{X} - \mathcal{X}_j))$. 
Notice that since $\mathcal{W}_j \subseteq \mathcal{D}_j$ by lemma \ref{lem:adjacent-cutoffs-agree}, we have $E_{\mathcal{W}_j} \subseteq E_{\mathcal{D}_j}$ by entanglement wedge nesting \cite{maximin}.
This implies $J^-(E_{\mathcal{W}_j})\subseteq J^-(E_{\mathcal{D}_j})$, and therefore
\begin{align}\label{eq:futureboundarycontained}
    \bigcap_{j} J^-(E_{\mathcal{W}_j}) \subseteq \bigcap_j J^-(E_{\mathcal{D}_j}).
\end{align}
Assuming $E_{\mathcal{V}}$ is connected, its HRRT surface is $\gamma_{\mathcal{V}}=\cup_j\gamma_{\mathcal{X}_j}$.
Since $\gamma_{\mathcal{D}_j}=\gamma_{\mathcal{X}_j}$, the future boundary of $\cap_j J^-(E_{\mathcal{D}_j})$ is also the future boundary of $E_{\mathcal{V}}$. 
So equation \eqref{eq:futureboundarycontained} implies that $\cap_{j} J^-(E_{\mathcal{W}_j})$ is to the past of the future boundary of $E_{\mathcal{V}}$. 
Since $J(E_{\mathcal{V}_j},E_{\mathcal{V}_k} \rightarrow E_{\mathcal{W}_1},\dots,E_{\mathcal{W}_n})\subseteq \cap_{\ell} J^-(E_{\mathcal{W}_\ell})$, the scattering region is also to the past of the future boundary of $E_{\mathcal{V}}$.

It remains to show that the scattering region is to the future of the past boundary of ${E}_\mathcal{V}$.
This is immediate by the inclusion $E_{\mathcal{V}_j} \subseteq E_{\mathcal{V}}$, which is an example of entanglement wedge nesting.

\vspace{0.2cm}
\noindent \textbf{Why no \texorpdfstring{$n$}{TEXT}-to-\texorpdfstring{$m$}{TEXT} theorem for \texorpdfstring{$n\neq m$}{TEXT}?}
\vspace{0.2cm}

A natural extension of the $n$-to-$n$ connected wedge theorem would be to consider cases with $n$ input locations and $m$ output locations, where $n\neq m$. 
While we believe some entanglement constraints should be associated with these settings, their form must be different than the most obvious generalization of the connected wedge theorem. 
In particular, defining input regions according to
\begin{align}
    \mathcal{V}_j = J(c_j\rightarrow r_1,\dots,r_m),
\end{align}
we can find counterexamples to the claim that an all-to-all bulk scattering vertex implies the entanglement wedge of $\mathcal{V}=\cup_j\mathcal{V}_j$ is connected.
See figure \ref{fig:ntomcounterexample}. 

Furthermore, the proof techniques used in section \ref{sec:GR-proof} do not seem to generalize to situations with $n \neq m.$
In the proof of that section, each boundary-anchored component of the final ``contradiction surface'' was obtained by focusing along a single past lightsheet $\mathscr{P}_j.$
For this to work, it is essential that the number of boundary-anchored components in the HRRT surface of a connected entanglement wedge is the same as the number of past lightsheets.
Otherwise, it would be necessary to focus a single boundary-anchored component of the contradiction surface along multiple past lightsheets simultaneously, which is not guaranteed to be area-nonincreasing.

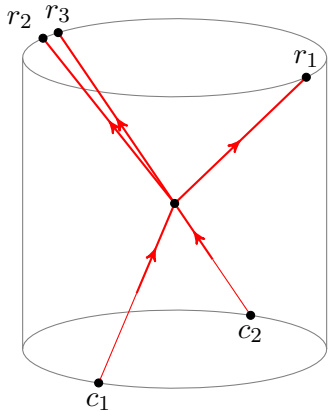
\begin{figure}
    \centering
    \tdplotsetmaincoords{15}{0}
    \begin{tikzpicture}[scale=1.0,tdplot_main_coords]
    \tdplotsetrotatedcoords{0}{30}{0}
    \draw[gray] (-2,1,0) -- (-2,5,0);
    \draw[gray] (2,1,0) -- (2,5,0);
    
    \begin{scope}[tdplot_rotated_coords]
    
    \begin{scope}[canvas is xz plane at y=1]
    \draw[gray] (0,0) circle[radius=2] ;
    \end{scope}
    
    \begin{scope}[canvas is xz plane at y=5]
    \draw[gray] (0,0) circle[radius=2] ;
    \end{scope}
    
    \draw[red] (0,1,-2) -- (0,2,-1);
    \draw[red] (0,1,2) -- (0,2,1);
    
    \draw[thick, red,mid arrow] (0,3,0) -- (2,5,0);
    \draw[thick, red,mid arrow] (0,3,0) -- (-2,5,0);
    \draw[thick, red,mid arrow] (0,3,0) -- (-2*0.985,5,2*0.174);
    
    \draw[thick,red,mid arrow] (0,2,-1) -- (0,3,0);
    \draw[thick,red,mid arrow] (0,2,1) -- (0,3,0);
    
    \draw plot [mark=*, mark size=1.5] coordinates{(2,5,0)};
    \node[above] at (2,5,0) {$r_1$};
    \draw plot [mark=*, mark size=1.5] coordinates{(-2,5,0)};
    \node[above left] at (-2,5,0) {$r_2$};
    \draw plot [mark=*, mark size=1.5] coordinates{(-2*0.985,5,2*0.174)};
    \node[above] at (-2*0.985,5,2*0.174) {$r_3$};
    
    \draw plot [mark=*, mark size=1.5] coordinates{(0,1,-2)};
    \node[below] at (0,1,-2) {$c_1$};
    \draw plot [mark=*, mark size=1.5] coordinates{(0,1,2)};
    \node[below] at (0,1,2) {$c_2$};
    \draw plot [mark=*, mark size=1.5] coordinates{(0,3,0)};
    
    \end{scope}
    \end{tikzpicture}
    \caption{Counterexample to the naive generalization of the connected wedge theorem with $2$ input and $3$ output points. Input points $c_1$ and $c_2$ are place antipodally at time $t=0$, and output points $r_1$ and $r_2$ are at time $\pi$, again antipodal but rotated a quarter turn. These four points define input regions that are exactly on the transition between having connected and disconnected entanglement wedges. Adding an additional output point $r_3$ makes one of the input regions smaller, so that the wedge are disconnected, despite there being a $2$-to-$3$ vertex in the bulk.}
    \label{fig:ntomcounterexample}
\end{figure}

\vspace{0.2cm}
\noindent \textbf{Connected wedge theorems in higher dimensions}
\vspace{0.2cm}

In \cite{may2020holographic, may2021holographic}, it was claimed that the proof of the $2$-to-$2$ connected wedge theorem holds in arbitrary spacetime dimension.
This is not true.
While the theorem itself \textit{may} be true in higher spacetime dimensions, the proof techniques used in those references and in section \ref{sec:GR-proof} do not work above three bulk dimensions.
The reason is the same as in the previous subsection: in three bulk dimensions, when $E_{\mathcal{V}} = E_{\mathcal{V}_1 \cup \dots \cup \mathcal{V}_n}$ is connected, then its HRRT surface contains $n$ boundary-anchored components, which is the same as the number of past lightsheets.
In higher dimensions, the HRRT surface of a connected entanglement wedge generally only has one boundary-anchored component.

The quantum information argument given in section \ref{sec:final-tasks-argument} may seem to be dimension-agnostic, but the key assumptions made in that argument are less easily justified above three bulk dimensions.
For the argument to work, it was necessary to assume that the $\mathcal{X}$ region, which is the spacelike complement of the input regions, is not a useful resource for the $\hat{\mathbf{B}}_{84}^{n, \delta}$ task.
In three bulk dimensions, the $\mathcal{X}$ region splits into $n$ components, which we have been calling $\mathcal{X}_j$.
Each of these components lies in between two neighboring input regions, providing a ``buffer'' between them.
Above three bulk dimensions, the structure of $\mathcal{X}$ is quite different; it does not split up into components, and instead functions as a ``bath'' in which the input regions are immersed.
This causes the entanglement between $\mathcal{X}$ and the input regions to be more complicated; it may be the case that neglecting $\mathcal{X}$ is justified in three bulk dimensions, but not in higher dimensions.

It would be interesting to know whether the connected wedge theorem is true above three dimensions.
If an explicit counterexample were found, it could be used to draw implications about how holographic CFT states in two dimensions differ, as information-theoretic resources, from their higher dimensional counterparts.

\vspace{0.2cm}
\noindent \textbf{Towards a loophole-free quantum information proof}
\vspace{0.2cm}

As emphasized in the previous subsection, a drawback of our quantum tasks argument for the $n$-to-$n$ connected wedge theorem in section \ref{sec:final-tasks-argument} is that we have ignored the $\mathcal{X}_j$ regions in treating the boundary causal network. 
We believe that the argument fails if these are included and the states held there are not constrained. 
It would be interesting to understand a minimal set of assumptions needed for the theorem to follow even with the $\mathcal{X}_j$ regions included, and if these assumptions hold for holographic states. 

\vspace{0.2cm}
\noindent \textbf{Arbitrary quantum tasks with biparition entanglement}
\vspace{0.2cm}

In section \ref{sec:boundaryprotocols}, we explained how to complete certain tasks in the boundary causal structure, using only the pattern of entanglement given by the connected wedge theorem. 
Note, however, that we made two omissions.
First, we didn't make use of the $\mathcal{X}_j$ regions, which allow additional quantum systems to enter the scattering regions, aside from those present in the input regions. 
Considering this complete description of the boundary causal network, the AdS/CFT correspondence implies arbitrary quantum tasks can be completed. 
Second, we did not give a protocol for arbitrary quantum tasks, only a protocol for completing Clifford unitaries. 
It would be interesting to understand if the $\mathcal{X}_j$ regions need to be included to allow arbitrary quantum tasks to be completed, or to design protocols on the weaker causal structure that complete arbitrary quantum tasks.

\vspace{0.2cm}
\noindent \textbf{Signatures for higher degree scattering regions}
\vspace{0.2cm}

In this paper, we have seen that an appropriate set of $2\rightarrow n$ scattering regions can replace, for the purposes of completing quantum tasks, any $n\rightarrow n$ scattering region. 
Towards the general understanding of how bulk causal structure is recorded into boundary entanglement, this leaves an open puzzle: what is the boundary signature of a $n\rightarrow n$ scattering region, with $n>2$?
Such regions do have a signature in terms of Lorentzian CFT correlation functions, but we might also expect them to have consequences for boundary correlations.

One interesting possibility is the following. 
The protocol given in section \ref{sec:arbitrarytasks} requires distributing large numbers of EPR pairs: for computations involving $\Theta(n)$ qubits, the protocol distributes $\Omega(2^n)$ EPR pairs. 
Unless there are more efficient protocols, this would limit the size of computations to $n=o(\log(\ell_{\text{AdS}}/G_N))$ qubits. 
In comparison, given a $n$-to-$n$ scattering region we expect $o(\ell_{\text{AdS}}/G_N)$ qubits can be brought into the region, and that at least low complexity operations can be performed on them. 
Thus, $n$-to-$n$ regions may allow more computations to happen in the bulk in the setting where we do not ignore backreaction. 
The distinction between performing computations of size $\log(\ell_{\text{AdS}}/G_N)$ and $\ell_{\text{AdS}}/G_N$ would then need to appear in the boundary. 
This could take the form of some additional type of resource shared among the $n$ input regions, for instance multipartite entanglement of a form not distinguished from bipartite entanglement by the mutual information. 

\vspace{0.2cm}
\noindent \textbf{Position-verification cheating protocols in higher dimensions}
\vspace{0.2cm}

In the cryptographic task of position-verification \cite{chandran2009position,kent2011quantum,kent2006tagging, malaney2010location}, an honest player performs a quantum task by acting locally within some scattering region(s), while a dishonest player is forced to perform the same task by exploiting entanglement and using some weaker causal network. 
The standard scenario considers position verification in $1+1$ dimensions, and the relevant networks are the two that appear in the connected wedge theorem for $n=2$. 
Considering position-verification in higher dimensions, more involved networks become relevant. 
The pairs of networks appearing in the connected wedge theorem for $n>2$ may in some cases be the relevant pair in this setting. 
It would be interesting to understand this better, and explore the consequences of our boundary protocol and bulk picture for the security of position-verification schemes. 

\vspace{0.2cm}
\noindent \textbf{EAQECC in non-local quantum computation}
\vspace{0.2cm}

The boundary protocol discussed in section \ref{sec:boundaryprotocols} introduced a technique for performing quantum tasks on causal networks: given system $A$ at spacetime point $p$, we used an EAQECC to split $A$ into subsystems, some of which are outside the future of $p$. 
This is also possible using teleportation, but only at the expense of introducing Pauli corruptions of the distant share, or hiding the share in an unknown port. 
This technique may be of use throughout the understanding of quantum information processing in spacetime and causal networks.

In holography, even for $n=2$ the boundary protocol is most naturally modeled as an EAQECC: the particles fall into the bulk is the encoding step, then logical operations are performed on the encoded systems in tensor product form, then the systems are decoded, pulling them out from the bulk. 
Given the tension between the efficiency, in terms of entanglement use, in AdS/CFT vs in existing non-local computation protocols in quantum information \cite{may2022complexity}, it is natural to explore EAQECC based non-local computation protocols in more detail. 
A simple starting point may be to adapt the code based strategies of \cite{cree2022code} to use entanglement-assisted codes. 

\acknowledgments{
We thank Netta Engelhardt, Matt Headrick, Kfir Dolev, and Eshan Kemp for helpful discussions. 
Significant portions of this work were completed while AM and JS were visiting the Perimeter Institute.
AM is supported by the Simons Foundation It from Qubit collaboration, a PDF fellowship provided by Canada’s National Science and Engineering Research council, and by Q-FARM.
JS is supported by AFOSR award FA9550-19-1-0369, CIFAR, DOE award DE-SC0019380 and the Simons Foundation.
Research at Perimeter Institute is supported in part by the Government of Canada through the Department of Innovation, Science and Economic Development Canada and by the Province of Ontario through the Ministry of Colleges and Universities.
}

\appendix

\section{Travel guide: general relativity}\label{sec:GR}
\label{app:GR-travel-guide}

The proof of the $n$-to-$n$ connected wedge theorem given in section \ref{sec:GR-proof} requires some technical background that we will not assume is familiar to all readers. Stating the theorem precisely requires defining terms like ``AdS-hyperbolic,'' and proving the theorem requires applying lemmas from the theory of causal structure. To avoid burdening the main text with technical digressions and specifications, we provide this appendix to (i) give precise definitions of many technical terms used in the main text, and (ii) state (with sources) the essential lemmas that we must use.

\subsection{Glossary}\label{sec:glossary}
\label{app:glossary}

\begin{definition}
    A spacetime $(\mathcal{M}, g)$ satisfies the \textbf{null curvature condition} if for every null vector $k^a$ we have
    \begin{equation}
        R_{ab} k^a k^b \geq 0.
    \end{equation}
\end{definition}

\begin{definition}
    A spacetime is said to have a \textbf{singularity} at the end of the curve $\gamma$ if $\gamma$ is an affinely parameterized, inextendible geodesic with no future endpoint and finite future extent in affine parameter.
\end{definition}

\begin{definition}
    Given a spacetime $(\mathcal{M}, g)$ and a set $S \subseteq \mathcal{M},$ the \textbf{domain of dependence} $D(S)$ is the set of all points $p \in \mathcal{M}$ for which every inextendible causal curve passing through $p$ intersects $S.$
\end{definition}

\begin{definition}
	A spacetime set $S$ is \textbf{achronal} if there are no future-directed timelike curves connecting distinct points of $S$, and \textbf{acausal} if there are no future-directed causal curves connecting distinct points of $S.$
\end{definition} 

\begin{definition}
    A spacetime is \textbf{globally hyperbolic} if there exists a closed, achronal set $S$ with $\mathcal{M} = D(S).$ The set $S$ is then said to be a \textbf{Cauchy slice} of $\mathcal{M}$ or a \textbf{Cauchy surface} for $\mathcal{M}.$
\end{definition}

\begin{definition} 
	Given a spacetime $(\mathcal{M}, g)$ and a set $S \subseteq \mathcal{M},$ the symbol $J^+(S)$ ($J^-(S)$) denotes the \textbf{causal future (past)} of $S,$ i.e., the set of points that can be reached from $S$ by future-directed (past-directed) causal curves. The sets $I^+(S)$ and $I^-(S)$ denote the \textbf{chronological future and past} of $S,$ and are defined equivalently but with timelike curves instead of causal curves.
\end{definition}

\begin{definition}
	The \textbf{edge} or \textbf{spatial boundary} of a closed spacetime set $S$, denoted $\partial_0 S,$ is the set of all points $p \in S$ for which any neighborhood of $p$ contains points $x \in I^-(p), y \in I^+(p)$ and a timelike curve from $x$ to $y$ that does not intersect $S.$
\end{definition}

\begin{definition}
	A $(d+1)$-dimensional spacetime $(\mathcal{M}, g)$ is said to have \textbf{conformal completion} $(\tilde{\mathcal{M}}, \tilde{g})$ if $\tilde{\mathcal{M}}$ is a smooth manifold-with-boundary and there is a smooth embedding $\mathcal{M} \to \tilde{\mathcal{M}}$ for which (i) the image of $\mathcal{M}$ is the interior of $\tilde{\mathcal{M}}$, and (ii) the metric $\tilde{g}$ restricted to the interior of $\tilde{\mathcal{M}}$ is related to the push-forward of $g$ by a Weyl factor.
\end{definition}

\begin{remark}
	From now on, whenever we talk about a spacetime $\mathcal{M},$ we will assume we have also specified a conformal completion $\tilde{\mathcal{M}}.$ Note that \textit{causal sets $J^{\pm}$ and $I^{\pm}$ will always be taken to refer to a conformally completed spacetime, and will include points on the conformal boundary.}
\end{remark}

\begin{definition}
	Given a spacetime with a conformal completion $\tilde{\mathcal{M}}$ and a set $S \subseteq \mathcal{M},$ the symbol $\hat{J}^+(S)$ ($\hat{J}^-(S)$) denotes the the set of points that can be reached from $S$ by future-directed (past-directed) causal curves that lie entirely in the spacetime boundary. The sets $\hat{I}^+(S)$ and $\hat{I}^-(S)$ are defined equivalently but with timelike curves instead of causal curves.
\end{definition}

\begin{definition}
	A spacetime $\mathcal{M}$ will be said to be \textbf{AdS-hyperbolic} if its conformal compactification $\tilde{\mathcal{M}}$ is globally hyperbolic with compact Cauchy surface $\Sigma.$ By abuse of terminology, the restriction of $\Sigma$ to $\mathcal{M}$ is sometimes called a Cauchy surface for $\mathcal{M},$ even though it does not satisfy $D(\Sigma|_{\mathcal{M}}) = \mathcal{M}.$
\end{definition}

\begin{definition}
	A component $\mathcal{B} \subseteq \partial \tilde{\mathcal{M}}$ is said to be \textbf{asymptotically anti-de Sitter} if the spacetime $(\mathcal{B}, \tilde{g}|_{T \mathcal{B}})$ is globally hyperbolic and the AdS equation $R_{ab} = - d\, g_{ab}$ is satisfied in a neighborhood of $\mathcal{B}$ at order $1/z^2$ in Fefferman-Graham coordinates. (This condition is independent of the particular choice of Fefferman-Graham coordinates; see section 2 of \cite{sorce2019cutoffs}. Note that in that paper, spacetime dimension was taken to be $d$, not $d+1.$)
\end{definition}

\begin{definition}
	A component $\mathcal{B} \subseteq \partial \tilde{\mathcal{M}}$ is said to be \textbf{asymptotically globally anti-de Sitter} if it is asymptotically locally anti-de Sitter and is Weyl-equivalent to a flat timelike cylinder.
\end{definition}

\begin{definition}
	A \textbf{boundary region} $\mathcal{R}$ is the domain of dependence within the spacetime $\partial \tilde{\mathcal{M}}$ of an embedded, closed, spacelike, acausal, codimension-one submanifold or submanifold-with-boundary of $\partial \tilde{\mathcal{M}}.$ Any such submanifold is called a \textbf{slice} of $\mathcal{R}.$
\end{definition}

\begin{definition} \label{def:HRRT}
	Given a boundary region $\mathcal{R}$ on an asymptotically anti-de Sitter boundary of a spacetime $\mathcal{M},$ the \textbf{HRRT surface} of $\mathcal{R}$ is the smallest-area, codimension-2 bulk surface such that:
    \begin{enumerate}[(i)]
        \item The surface is locally extremal.
        \item There exists a Cauchy surface $\Sigma$ for $\tilde{M},$ containing the spatial boundary $\del_0 \mathcal{R},$ on which the HRRT surface is homologous to $\Sigma \cap \mathcal{R}.$
    \end{enumerate}

    The \textbf{entanglement wedge} of $\mathcal{R}$, denoted $\text{EW}(\mathcal{R}),$ is the domain of dependence within $\tilde{\mathcal{M}}$ of the homology region. Equivalently, it is the closure of the set of all points spacelike between the HRRT surface and $\mathcal{R}.$
\end{definition}

\subsection{Handbook of causal structure lemmas}
\label{app:causal-lemmas}

\begin{lemma} \label{lem:causal-timelike-addition}
    Given three spacetime points $p, q, r$ with $r \in I^+(q)$ and $q \in J^+(p),$ we have $r \in I^+(p).$

    Given three spacetime points $p, q, r$ with $r \in J^+(q)$ and $q \in I^+(p),$ we have $r \in I^+(p).$
\end{lemma}
\begin{proof}
    This is proposition 2.18 of \cite{penrose1972techniques}.
\end{proof}

\begin{lemma} \label{lem:lightsheet-generators}
    For any closed spacetime set $S$, and any point
    \begin{equation}
        q \in J^+(S) - I^+(S),
    \end{equation}
    there exists a point $p \in S$ and a future-directed null geodesic from $p$ to $q$ that lies entirely in $J^+(S) - I^+(S).$ Furthermore, if $q$ is not in $S$, then $p$ is in $\partial_0 S.$ Analogous statements hold for $J^-(S) - I^-(S).$
\end{lemma}
\begin{proof}
    The existence of a point $p \in S$ and a future-directed null geodesic $p$ to $q$ is given as proposition 2.20 of \cite{penrose1972techniques}, and the statement that the full geodesic segment from $p$ to $q$ must lie in $J^+(S) - I^+(S)$ follows from lemma \ref{lem:causal-timelike-addition}. The statement that $q \notin S$ implies $p \in \partial_0 S$ follows from the definition of the edge $\partial_0 S$: if $p$ were not in $\partial_0 S,$ then it could be deformed to a nearby point $p' \in S$ with $q \in I^+(p').$
\end{proof}

\begin{lemma} \label{lem:point-futures-compact}
	In a globally hyperbolic spacetime, for any points $p$ and $q$, the set $J^+(p) \cap J^-(q)$ is compact.
\end{lemma}
\begin{proof}
	This is proved as theorem 8.3.10 in \cite{waldbook}.
\end{proof}

\begin{lemma} \label{lem:lightsheet-boundary}
	In a globally hyperbolic spacetime, for any compact set $S$, we have
	\begin{equation}
		J^{\pm}(S) - I^{\pm}(S) = \del J^{\pm}(S).
	\end{equation}
\end{lemma}
\begin{proof}
	See the discussion following theorem 8.3.11 in \cite{waldbook}.
\end{proof}

\begin{remark} \label{rem:closed-futures}
    Lemma \ref{lem:lightsheet-boundary} implies that in a globally hyperbolic spacetime, the causal future and past of a compact set are topologically closed.
\end{remark}

\begin{lemma} \label{lem:open-futures}
    For any spacetime set $S$, the sets $I^+(S)$ and $I^-(S)$ are topologically open.
\end{lemma}
\begin{proof}
    See corollary 2.9 of \cite{penrose1972techniques}.
\end{proof}

\begin{lemma} \label{lem:future-past-singularities}
    In a globally hyperbolic spacetime, all singularities are either ``future type'' or ``past type.'' I.e., there exists no singularity that is in the future of a spacetime point $p$ and the past of another spacetime point $q.$
\end{lemma}
\begin{proof}
    By lemma \ref{lem:point-futures-compact}, $J^+(p) \cap J^-(q)$ is compact. So any affinely parametrized geodesic in $J^+(p) \cap J^-(q)$ with finite future affine extent must have a future endpoint. This implies there can be no singularities in this set.
\end{proof}

\section{Arbitrary tasks in the bulk causal structure}\label{sec:bulkprotocolappendix}

\subsection{Tools for performing quantum tasks}\label{sec:NLQCtools}

There are two standard tools used in building protocols to complete quantum tasks: Bell-basis teleportation \cite{bennett1993teleporting} and port-teleportation \cite{ishizaka2008asymptotic, beigi2011simplified}. 
We briefly review both of these below. 
Because our protocols involve repeated applications of these tools it is convenient to introduce some special notation, which we adapt from \cite{dolev2019constraining}. 

In Bell-basis teleportation (usually just called quantum teleportation), a qudit $\ket{\psi}_A$ held at location $p$ is transmitted to a distant location $p'$ by use of a shared maximally entangled state and classical communication. 
In more detail, the maximally entangled state
\begin{equation}
    \ket{\Phi_1}_{EE'} = \frac{1}{\sqrt{d}} \sum_{j=1}^{d} \ket{j j}_{E E'}
\end{equation}
is distributed, with $E$ at location $p$ and $E'$ at $p'$. 
A Bell basis $\{\ket{\Phi_j}\}_{j=1}^{d^2}$ is chosen --- i.e., an orthogonal basis of maximally entangled states --- and $AE$ is measured in this basis, giving measurement outcome $x\in\{1, \dots, d^2\}$ at $p$.
The maximally entangled state $\ket{\Phi_x}_{EE'}$ is related to $\ket{\Phi_1}_{EE'}$ by some unitary $\mathbf{U_x}$ satisfying
\begin{equation}
    \ket{\Phi_x}_{EE'} = (\mathbf{U_x} \otimes I) \ket{\Phi_1}_{EE'} = (I \otimes \mathbf{U_x}^T) \ket{\Phi_1}_{EE'}
\end{equation}
and
\begin{equation}
    \tr(\mathbf{U_x}^{\dagger} \mathbf{U_{x'}}) = d\, \delta_{x x'}.
\end{equation}
For the standard Bell basis on qubits, the operators $\{\mathbf{U_x}\}$ are phase-equivalent to the Pauli operators $\{I, X, Y, Z\}.$
After the Bell measurement is performed and outcome $x$ is registered, the state on $E'$ is
\begin{align}
    \mathbf{U_x}^{\dagger} \ket{\psi}_{E'}.
\end{align}
By sending $x$ from $p$ to $p'$ the unitary $\mathbf{U_x}$ can be applied and the state $\ket{\psi}$ can be reproduced on the $E'$ system.

The term ``teleportation'' often refers to the full process of performing a Bell measurement and applying the unitary $\mathbf{U_x}.$
In the context of quantum tasks, however, it is often helpful to discuss these steps separately.
We will say $A$ is teleported$^*$ from $p$ to $p'$ if the measurement has been performed but $x$ has not yet been transmitted. 
When $A$ has been teleported$^*$ from $p$ to $p'$ using entanglement on $EE'$, we will relabel the system $E'$ to $A[x]$.

In port-teleportation, a qudit in state $\ket{\psi}_A$ held at location $p$ is transmitted to a distant location $p'$ by use of a set of $N$ shared maximally entangled states, $\otimes_{j=1}^N \ket{\Phi}_{E_iE_i'}$.
To do this, a measurement on the $AE_1....E_N$ system\footnote{The details of this measurement are described in \cite{ishizaka2008asymptotic}.} is performed, giving measurement outcome $j\in\{1,...,N\}$. 
The post-measurement state is then within $\epsilon = 4 d^2/\sqrt{N}$ one-norm distance of the state
\begin{align}
    \ketbra{\psi}{\psi}_{E'_j} \otimes \rho_{E'_1 \dots E'_{j-1} E'_{j+1} \dots E'_N}
\end{align}
for some state $\rho$ on the non-$E'_j$ registers.
By sending the classical data $j$ from $p$ to $p'$, the non-$E'_j$ systems can be discarded, and the state $\ket{\psi}$ is approximately recovered on register $E'_j,$ with the quality of the approximation improving as $N$ is increased.
As with Bell-basis teleportation, we designate performing only the measurement step of port-teleportation as port-teleport$^*$.
Further, when $A$ is port-teleported$^*$ from $p$ to $p'$ we relabel the $E_1'...E_N'$ system as $(A)^j$. 
Notice that if $j$ is known we can recover $A$ from $(A)^{j}$ by tracing out all but the $E'_j$ subsystem and relabeling it as $A$. 

During port-teleportation, it is possible to apply arbitrary operations on the teleported system \emph{before} finding out the measurement outcome $j$. 
In particular, to apply a channel $\mathbfcal{N}_{A}$ to the port-teleported$^*$ system we need only apply $\otimes_{j=1}^N \mathbfcal{N}_{E_j'}$. 
This produces a state close in one-norm distance to
\begin{align}
    \mathbfcal{N}(\ketbra{\psi}{\psi}_{E'_j})\otimes \rho'_{E'_1 \dots E'_{j-1} E'_{j+1} \dots E'_N}
\end{align}
for some new state $\rho'$. 
When all but the $E'_j$ system is traced out, we are left with an approximation of the state $\mathbfcal{N}(\ketbra{\psi}{\psi})$. 

We will consider scenarios involving multiple port and Bell-basis teleportations. 
The notation used to describe this is as follows.
Suppose $A_1$ is teleported$^*$ from $p_1$ to $p_2$, then $A_1A_2$ is teleported$^*$ from $p_2$ to $p_3$. 
After the first teleportation$^*$ we will say the system $A_1[x_1]$ is at $p_2$, and after the second teleportation we will say the system $(A_1[x_1] A_2)[x_2]$ is at $p_3$.
If system $(A_1[x_1] A_2)[x_2]$ is now port-teleported to $p_4$, we will say the system $((A_1[x_1] A_2)[x_2])^{(j_3)}$ is at $p_4$. 

As an example to illustrate the use of port and Bell basis teleportation in completing quantum tasks, consider a task defined on the causal network of figure \ref{fig:2-2repeatednlqc}. 
This network has two input and two output points, and no intermediary vertices. 
Consider a task that requires implementing some channel $\mathbfcal{N}_{A_1A_2\rightarrow B_1B_2}$. 
Given unlimited pre-shared entanglement, we can complete the task with arbitrary accuracy as follows. 
\begin{protocol}\textbf{(Beigi-Konig \cite{beigi2011simplified})}
\begin{itemize}
    \item \textbf{At the input points $p_1$ and $p_2$:}
    \begin{itemize}
        \item Teleport$^*$ system $A_1$ from $p_1$ to $p_2$, producing system $A_1[x_1]A_2$ at $p_2$. 
        \item Port-teleport$^*$ system $A_1[x_1]A_2$ from $p_2$ to $p_1$, producing system $(A_1[x_1]A_2)^{j_2}$ at $p_1$. 
        \item At $p_1$, act on every port of $(A_1[x_1]A_2)^{j_2}$ with $P_{x_1} \otimes I$, producing $(A_1A_2)^{j_2}$. 
        \item At $p_1$, act on every port of $(A_1A_2)^{(j_2)}$ with the needed channel $\mathbfcal{N}_{A_1A_2\rightarrow B_1B_2}$, producing $(B_1B_2)^{(j_2)}$.
        \item Send the measurement outcome $j_2$ from $p_2$ to both output points $q_1$ and $q_2.$
        \item Send the $(B_1)^{(j_2)}$ system to $q_1$, and the $B_2^{(j_2)}$ system to $q_{2}$. 
    \end{itemize}
    \item \textbf{At the output points $q_1$ and $q_2$:}
    \begin{itemize}
        \item At $q_1$, use $j_2$ to recover $B_1$ from $B_1^{(j_2)}$.
        \item At $q_2$, use $j_2$ to recover $B_2$ from $B_2^{(j_2)}$.
    \end{itemize}
\end{itemize}
\end{protocol}
The protocol we give in the next section is a generalization of this one.

\subsection{Bulk protocol for arbitrary tasks}

In the proof of theorem \ref{thm:2-to-all-sufficiency} presented in section \ref{sec:arbitrarytasks}, we claimed that in a causal network where every input vertex can signal every output vertex, an arbitrary quantum channel can be implemented so long as there is a way of ordering the input vertices so that every adjacent pair shares an arbitrarily large maximally entangled state.
We describe that protocol explicitly here, using the techniques developed in \cite{beigi2011simplified,dolev2019constraining}. 

Let $\tilde{\Gamma}$ be such a causal structure, let $\{c_1, \dots, c_m\}$ be an ordering of the input vertices such that every adjacent pair shares an arbitrarily large maximally entangled state, and let $\{r_1, \dots, r_n\}$ be the output vertices.
Note that the chain of entanglement may pass through a given input vertex multiple times: i.e., we may have $c_j = c_k$ for $j \neq k.$
The following protocol implements an arbitrary channel from the input points to the output points.

\begin{protocol}(\textbf{Arbitrary bulk tasks})
    \begin{itemize}
        \item \textbf{At the input locations:}
        \begin{enumerate}
            \item Teleport$^*$ all quantum systems at $c_1$ to $c_2.$ Teleport$^*$ all quantum systems at $c_2$ to $c_3.$ Continue on in this manner until all systems have been teleport$^*$-ed to $c_m.$ At the end of this process, all quantum systems are present at $c_m,$ but encoded by a chain of unitaries $\mathbf{U}^{\dagger}_{x_{m-1}} \dots \mathbf{U}^{\dagger}_{x_1}.$
            \item Port-teleport$^*$ all systems from $x_m$ to $x_{m-1}$. Act with $\mathbf{U}_{x_{m-1}}$ on all ports to undo last unitary in the encoding chain. Repeat this process until all input systems are present at $c_1,$ in their original state, on an unknown port..
            \item Perform the desired channel on all ports.
            \item From $c_1$, split all ports into the subsystems corresponding to each output vertex where that subsystem is supposed to end up.
            \item From each $c_2$ through $c_m$, send the measurement outcomes from the port-teleportation$^*$ to all output locations.
        \end{enumerate}
        \item \textbf{At the output locations:}
        \begin{enumerate}
            \item At each output location, use the port-teleportation$^*$ measurement outcomes to determine which port contains the correct output system, then return that system.
        \end{enumerate}
    \end{itemize}
\end{protocol}

\section{Mutual information in \texorpdfstring{$\hat{\mathbf{B}}^{n, \delta}_{84}$}{TEXT}} \label{app:lower-bound-proof}

This appendix provides a proof of lemma \ref{lemma:MIlowerbound} from the main text.

\vspace{0.2cm}
\noindent \textbf{Lemma \ref{lemma:MIlowerbound}.} (Mutual information lower bound) \emph{   
Consider the causal structure shown in figure \ref{fig:2-2repeatednlqc}, with resource system $\rho_{LR}$. 
    Fix $\epsilon, \delta$ with $\epsilon < \delta < \delta_*$ and $2^{h(\delta)} \beta > e^{-2 (\delta - \epsilon)^2}.$
    Then completing the $\mathbf{B}_{84}^{n,\delta}$ task with probability $1-e^{- 2 n(\delta - \epsilon)^2}$ requires
    \begin{align}
        \frac{1}{2}I(L:R)_\rho \geq - n \log_2(2^{h(\delta)} \beta) - 1 + O\left( (2^{h(\delta)} \beta)^n, e^{-2 n (\delta - \epsilon)^2} (2^{h(\delta)} \beta)^{-n} \right).
    \end{align}
}
\vspace{0.2cm}
\begin{proof}
    The one-norm distance between two states $\rho$ and $\sigma$ is defined by
    \begin{equation}
        \lVert \rho - \sigma \rVert|_{1} = \tr(|\rho - \sigma|).
    \end{equation}
    If you are given either state $\rho$ or $\sigma$ with probability $1/2,$ and are tasked with guessing which state you have been given, then your maximum probability of success optimized over all strategies is \cite{helstrom1969quantum,wilde2013quantum}
    \begin{equation} \label{eq:trace-distance-formula}
        p_{\text{dist}}(\rho, \sigma)
            = \frac{1}{2} + \frac{1}{4} \lVert \rho - \sigma \rVert_{1}.
    \end{equation}
    Given a pair of quantum tasks $\mathbf{T}_{\rho}$ and $\mathbf{T}_{\sigma},$ with the only difference being that $\mathbf{T}_{\rho}$ has resource state $\rho$ while $\mathbf{T}_{\sigma}$ has resource state $\sigma,$ we can devise a strategy for distinguishing $\rho$ and $\sigma$ as follows.
    We pick an optimal strategy for completing the task $\mathbf{T}_{\rho},$ and perform that strategy.
    If our we succeed at our task, we guess that we were given the state $\rho$.
    If we fail, we guess that we were given the state $\sigma.$
    The probability of successfully distinguishing the states using this strategy is
    \begin{equation} \label{eq:task-distinguishing-probability}
        \frac{1}{2} \text{Prob}(\text{success} | \rho) + \frac{1}{2} (1 - \text{Prob}(\text{success} | \sigma)) \leq p_{\text{dist}}(\rho, \sigma).
    \end{equation}
    Since the strategy we choose for completing the task is optimal for $\mathbf{T}_{\rho}$ but potentially suboptimal for $\mathbf{T}_{\sigma},$ we have
    \begin{align}
        \text{Prob}(\text{success} | \rho)
            & = p_{\text{suc}}(\mathbf{T}_{\rho}), \\
        \text{Prob}(\text{success} | \sigma)
            & \leq p_{\text{suc}}(\mathbf{T}_{\sigma}).
    \end{align}
    Combining these statements with inequality \eqref{eq:task-distinguishing-probability} and equation \eqref{eq:trace-distance-formula} gives
    \begin{equation}
         p_{\text{suc}}(\mathbf{T}_{\rho}) - p_{\text{suc}}(\mathbf{T}_{\sigma}) \leq \frac{1}{2} \lVert \rho - \sigma \rVert_{1}.
    \end{equation}
    Making the substitution ``$\rho \leftrightarrow \sigma$'' everywhere in the above discussion gives an analogous inequality that, combined with this one, becomes
    \begin{equation} \label{eq:pdistandTD}
         |p_{\text{suc}}(\mathbf{T}_{\rho}) - p_{\text{suc}}(\mathbf{T}_{\sigma})| \leq \frac{1}{2} \lVert \rho - \sigma \rVert_{1}.
    \end{equation}

    Trace distance can be related to mutual information as follows. First, recall that the mutual information is a relative entropy:
    \begin{equation}
        I(L:R)_{\rho_{LR}} = D(\rho_{L R} \Vert \rho_L \otimes \rho_R).
    \end{equation}
    Relative entropy is related to the fidelity $F(\rho, \sigma) = \tr \sqrt{\sigma^{1/2} \rho \sigma^{1/2}}$ by \cite{fuchs1999cryptographic}
    \begin{equation}
        \lVert \rho - \sigma \rVert_{1} \leq 2 \sqrt{1 - F(\rho, \sigma)^2}.
    \end{equation}
    Fidelity is related to relative entropy by
    \begin{equation}
        - 2 \log_2 F(\rho, \sigma) \leq D(\rho \Vert \sigma),
    \end{equation}
    which can be seen from the fact that both sides of this equation are sandwiched R\'{e}nyi relative entropies \cite{muller2013quantum, wilde2014strong} --- the left is $\alpha=1/2$ and the right is $\alpha=1$ --- and these are monotonically increasing in $\alpha.$
    Combining these inequalities, and choosing $\rho = \rho_{LR}$ and $\sigma = \rho_L \otimes \rho_R,$ we have
    \begin{equation}
        I(L:R)_{\rho_{LR}} \geq - \log_{2} (1 - \lVert \rho_{LR} - \rho_{L} \otimes \rho_{R} \rVert_{1}^2 /4).
    \end{equation}
    Using equation \eqref{eq:pdistandTD} to express this in terms of success probabilities, we have
    \begin{equation} \label{eq:Ilowerboundunfinished}
        I(L:R)_{\rho_{LR}} \geq - \log_{2} (1 - |p_{\text{suc}}(\mathbf{T}_{\rho_{LR}}) - p_{\text{suc}}(\mathbf{T}_{\rho_{L} \otimes \rho_R})|^2).
    \end{equation}
    
    Now, recall our assumption
    \begin{align}
        p_{\text{suc}}(\mathbf{T}_{\rho_{LR}}) \geq 1 - e^{-2 n (\delta - \epsilon)^2},
    \end{align}
    and that from lemma \ref{lemma:psucupperbound} we have
    \begin{align}
        p_{\text{suc}}(\mathbf{T}_{\rho_{L} \otimes \rho_R}) \leq (2^{h(\delta)}\beta)^n.
    \end{align}
    Inserting these into \ref{eq:Ilowerboundunfinished} and expanding at large $n$ using the assumption $2^{h(\delta)} \beta > e^{-2 (\delta - \epsilon)^2}$, we obtain the desired bound.
\end{proof}

\section{Code constructions}\label{sec:CodeConstructions}

\subsection{Three-qudit code}

A concrete choice of QECC which suffices for our protocol in section \ref{sec:boundaryprotocols} is the three qudit code~\cite{PhysRevLett.83.648}, which we review here. 
The Hilbert space for each qudit is spanned by $|0\rangle,\cdots, |p-1\rangle$ where $p$ is an odd prime integer. 
The encoding of a single qudit into three qudits is given by
\begin{align}
|x\rangle \quad \rightarrow \quad |\widetilde{x}\rangle = \frac{1}{\sqrt{p}}\sum_{y=0}^{p-1}|y, y+x,y+2x\rangle,
\end{align}
where sums in kets are computed modulo $p$. As we will see below, in this code the encoded state can be reconstructed by having access to any two qudits (i.e., losing one qudit is fine).

The coding properties can be studied by finding stabilizer generators and logical operators. Generalized Pauli operators for qudits are given by
\begin{align}
\mathbf{X} \equiv \sum_{j} |j+1\rangle \langle j |, \qquad \mathbf{Z} \equiv \sum_{j} \omega^j |j\rangle\langle j |,
\end{align}
where $\omega = \exp(2 \pi i/p)$. They obey the commutation relation
\begin{align}
 \mathbf{Z}\mathbf{X} = \omega \mathbf{X}\mathbf{Z}.
\end{align}
Stabilizer generators of the code are 
\begin{align}
\mathbf{S}_{X} = \mathbf X \otimes \mathbf X\otimes \mathbf X, \qquad \mathbf{S}_{Z} = \mathbf{Z} \otimes (\mathbf{Z}^{\dagger})^2 \otimes \mathbf{Z}. \label{eq:Pauli-logical}
\end{align}
One can verify $[\mathbf S_{X},\mathbf S_{Z}]=0$. 
The codeword states $|\widetilde{x}\rangle$ are $+1$ eigenstates of stabilizer generators.
Logical Pauli operators are  
\begin{align}
\mathbf{\widetilde{X}} = \mathbf I \otimes \mathbf X \otimes \mathbf{X}^2 \qquad \mathbf{\widetilde{Z}} = \mathbf{I} \otimes \mathbf{Z}^{\dagger} \otimes \mathbf{Z}
\end{align}
which act like Pauli operators for codeword states $|\widetilde{x}\rangle$. 
The above expressions of logical operators are supported on qudit-$2$ and qudit-$3$. 
In fact, any logical operator can be reconstructed on two any two of the three qudits.
Equivalently, the logical quantum state can be recovered from any two of the three qudits. 

\subsection{Three-qudit code as EAQECC}

In conventional quantum error-correcting codes, quantum information is encoded by some appropriate encoding unitary. 
For the case of the aforementioned three qudit code, the encoding unitary operator will have the form
\begin{align}
|\widetilde{x}\rangle = \mathbf{U}_{123}|x,0,0\rangle ,
\end{align}
where the initial state is prepared in qudit-$1$ while qudit-$2$ and qudit-$3$ are prepared in the state $|0\rangle$. 
The downside of this encoding is that one needs to implement a non-local unitary $\mathbf{U}_{123}$ which involves all three qudits. 

Entanglement assisted quantum error correcting codes (EAQECC) achieve the same encoding with a more local unitary operator by using pre-existing entanglement. 
For the case of the aforementioned qudit code, one can perform the encoding in the form
\begin{align} \label{eq:local-encoder}
|\widetilde{x}\rangle = (\mathbf{U}_{12}\otimes \mathbf{I}_{3})|x\rangle_{1} \otimes |\text{EPR}\rangle_{23}
\end{align}
where the qudit-$2$ and qudit-$3$ are initially in the EPR pair. 
Here the encoding is possible with some two-qudit unitary $\mathbf{U}_{12}$. 
See figure~\ref{fig_EAQECC_12}.

The concrete form of $\mathbf{U}_{12}$ can be obtained by rewriting $|\widetilde{x}\rangle$ as follows:
\begin{align}
|\widetilde{x}\rangle = \frac{1}{\sqrt{p}}\sum_{y}|y,y+x,y+2x\rangle 
= \frac{1}{\sqrt{p}}\sum_{y'}|y'-2x,y'-x,y'\rangle
\end{align}
by setting $y'=y+2x$. The initial state on the right-hand side of equation \eqref{eq:local-encoder} is given by $\frac{1}{\sqrt{p}}\sum_{y} |x,y,y\rangle$, so we can complete our encoding by applying 
\begin{align}
\mathbf{U}_{12}|x,y\rangle = |y-2x,y-x\rangle.
\end{align}
With similar analyses, one can find $\mathbf{U}_{23}$ and $\mathbf{U}_{31}$ such that 
\begin{align}
|\widetilde{x}\rangle = \mathbf{U}_{23}|x\rangle_{2} \otimes |\text{EPR}\rangle_{31} = 
 \mathbf{U}_{31}|x\rangle_{3} \otimes |\text{EPR}\rangle_{12}.
\end{align}

\subsection{Transversal logical gates}

A unitary operator which transforms encoded states inside the codeword subspace is called a logical operator. 
Our goal in section \ref{sec:boundaryprotocols} is to implement a desired unitary logical operation $\mathbf{V}$ at intermediate points $\mathcal{J}_{j}$'s. 
However, not all kinds of unitary operators can be applied in this way.
Since the intermediate points $\mathcal{J}_{1},\mathcal{J}_{2},\mathcal{J}_{3}$ are space-like separated, one can only implement a \emph{transversal logical gate} with a tensor product structure across the intermediate regions, as expressed in equation \eqref{eq:transverseunitary}.
The classification of transversal logical gates is an important subject in the study of QECCs since transversal operations do not spread local errors to other subsystems, and hence are desirable for fault-tolerant quantum computation. 

Concretely, and following what appears in the protocol, suppose we have three copies of the three qudit code, which store $A_i$ into shares $M_1^iM_2^iM_3^i$. 
Then, shares $M_{j}^1$, $M_j^2$, $M_j^3$ are brought to region $\mathcal{J}_j$. 
We are able then to apply unitary operations of the form
\begin{align}
    \mathbf{V} = \mathbf{V}_{M_{1}^1M_1^2M_1^3}^1 \otimes \mathbf{V}_{M_{2}^1M_2^2M_2^3}^2 \otimes \mathbf{V}_{M_{3}^1M_3^2M_3^3}^3.
\end{align}
It turns out that that the three qudit code allows a large group of logical gates, called the Clifford unitary group, to be implemented in this transverse form. 

Clifford operators are unitary operators which transform Pauli operators to Pauli operators under conjugation. 
The Clifford operations on qudits form a group, generated by Pauli $\mathbf{X}$, Pauli $\mathbf{Z}$, and the gates
\begin{align}
    \mathbf{CNOT}_{ab}\ket{i}_{a} \ket{j}_b &= \ket{i}_a\ket{i+j \,\, \text{mod} \,\, p}_{b}, \\
    \mathbf{H} \ket{i} &= \sum_{0 \leq m\leq p-1} \omega^{m i} \ket{m}, \\
    \mathbf{S} \ket{i} &= \omega^{i(i+1)/2} \ket{i}.
\end{align}
One can easily find transverse implementations of each of these operators. 

\subsection{Transversality beyond Clifford gates}

It is interesting to ask what type of code would be needed to extend the protocol given in section \ref{sec:boundaryprotocols} beyond Clifford unitaries.
Using the existing construction, we can actually implement certain non-Clifford unitaries by embedding a small Hilbert space into a larger one.\footnote{We thank Kfir Dolev and Sam Cree, for many related discussions, and who also proposed this idea independently.}
Consider embedding into a $q$-dimensional qudit with $q >p$ by simple identification of $|0\rangle, |1\rangle, \cdots, |p-1\rangle$ to corresponding states of a $q$-dimensional qudit. 
One can then implement Clifford gates for the $q$-dimensional qudits on systems of $p$-dimensional qubits.

We will focus on a diagonal phase gate of the form
\begin{align}
\mathbf V|x_{1},x_{2}\rangle = e^{i \theta  x_{1}x_{2}} |x_{1},x_{2}\rangle, \qquad \theta = \frac{2\pi}{q},
\end{align}
where $x_{1}x_{2}$ represents a product.
One can verify that such a quadratic phase gate is a Clifford operation for $q$-dimensional qudits. Note that diagonal phase gates commute with each other. 
As such, for a system of three $p$-dimensional qudits, the following gates can be transversally implemented:
\begin{align}
\mathbf V|x_{1},x_{2}, x_{3}\rangle = e^{i f(x_{1},x_{2},x_{3})} |x_{1},x_{2}, x_{3}\rangle
\end{align}
where
\begin{align}
f(x_{1},x_{2},x_{3})= \theta( n_{12} x_{1}x_{2} + n_{23} x_{2}x_{3} + n_{31}x_{3} x_{1}) \label{eq:quadratic-phase}
\end{align}
with $n_{12},n_{23},n_{31}$ being integers. By embedding the system into higher-dimensional Hilbert spaces, one can perform rotations with finer angles, inversely proportional to $q$.
In this protocol finer angle rotations require more EPR pairs. 
The required resource entanglement, measured in the unit of qubit EPR pairs, scales as $\sim \log (\frac{1}{\theta})$. 

More generally, we might be interested in exploiting codes that allow arbitrary unitaries to be implemented transversally. 
Such codes are constrained by the Eastin-Knill theorem \cite{eastin2009restrictions,zeng2011transversality,chen2008subsystem}, which states that transversal logical operators must form a discrete group if the QECC is built out of finite-dimensional qudits and has an ability to correct an erasure of at least one qudit. 
In other words, there is no finite-dimensional QECC which has a continuous group of transversal unitary logical gates.\footnote{In fact, there are further restrictions for stabilizer codes \cite{PhysRevA.91.012305}. 
Suppose a stabilizer code is supported on three subsystems $A,B,C$ and is protected from an erasure of one subsystem. 
Then non-Clifford logical gate cannot be written in a transversal form over $A,B,C$.}

However, the Eastin-Knill theorem can be avoided by allowing the code to only approximately recover from errors.  
For large numbers of subsystems, or large qudit dimensions, the error can be made arbitrarily small while allowing any unitary to be implemented transversally. 
Unfortunately, the qudit dimension or number of subsystems must grow quickly as the error is decreased \cite{faist2020continuous,kubica2021using,tajima2021universal,zhou2021new,liu2021approximate,liu2021quantum}, though it is not known if the entanglement used to prepare such a code as an EAQECC needs to grow quickly.  
Aside from having a fixed code which supports arbitrary transverse unitaries, we could also instead consider a code adapted to whatever particular unitary is required for the given task. 
It would be interesting to understand constraints on the existence of such adapted codes, or to construct them.

\bibliographystyle{JHEP}
\bibliography{biblio.bib}

\providecommand{\href}[2]{#2}\begingroup\raggedright\begin{thebibliography}{10}

\bibitem{gao2000theorems}
S.~Gao and R.~M. Wald, {\it Theorems on gravitational time delay and related
  issues},  {\em Classical and Quantum Gravity} {\bf 17} (2000), no.~24 4999,
  [\href{http://arxiv.org/abs/gr-qc/0007021}{{\tt gr-qc/0007021}}].

\bibitem{may2019quantum}
A.~May, {\it Quantum tasks in holography},  {\em Journal of High Energy
  Physics} {\bf 2019} (2019), no.~233
  [\href{http://arxiv.org/abs/1902.06845}{{\tt arXiv:1902.06845}}].

\bibitem{may2020holographic}
A.~May, G.~Penington, and J.~Sorce, {\it {Holographic scattering requires a
  connected entanglement wedge}},  {\em JHEP} {\bf 08} (2020) 132,
  [\href{http://arxiv.org/abs/1912.05649}{{\tt arXiv:1912.05649}}].

\bibitem{may2021holographic}
A.~May, {\it {Holographic quantum tasks with input and output regions}},  {\em
  JHEP} {\bf 08} (2021) 055, [\href{http://arxiv.org/abs/2101.08855}{{\tt
  arXiv:2101.08855}}].

\bibitem{kent2012quantum}
A.~Kent, {\it {Quantum Tasks in Minkowski Space}},  {\em Class. Quant. Grav.}
  {\bf 29} (2012) 224013, [\href{http://arxiv.org/abs/1204.4022}{{\tt
  arXiv:1204.4022}}].

\bibitem{hayden2016summoning}
P.~Hayden and A.~May, {\it {Summoning Information in Spacetime, or Where and
  When Can a Qubit Be?}},  {\em J. Phys. A} {\bf 49} (2016), no.~17 175304,
  [\href{http://arxiv.org/abs/1210.0913}{{\tt arXiv:1210.0913}}].

\bibitem{hayden2019localizing}
P.~Hayden and A.~May, {\it Localizing and excluding quantum information; or,
  how to share a quantum secret in spacetime},  {\em Quantum} {\bf 3} (2019)
  196, [\href{http://arxiv.org/abs/1806.04154}{{\tt arXiv:1806.04154}}].

\bibitem{dolev2021distributing}
K.~Dolev, A.~May, and K.~Wan, {\it {Distributing bipartite quantum systems
  under timing constraints}},  {\em J. Phys. A} {\bf 54} (2021), no.~14 145301,
  [\href{http://arxiv.org/abs/2011.00936}{{\tt arXiv:2011.00936}}].

\bibitem{dolev2019constraining}
K.~Dolev, {\it Constraining the doability of relativistic quantum tasks},
  \href{http://arxiv.org/abs/1909.05403}{{\tt arXiv:1909.05403}}.

\bibitem{chandran2009position}
N.~Chandran, V.~Goyal, R.~Moriarty, and R.~Ostrovsky, {\it Position based
  cryptography},  in {\em Annual International Cryptology Conference},
  pp.~391--407, Springer, 2009.
\newblock \href{http://arxiv.org/abs/1009.2490}{{\tt arXiv:1009.2490}}.

\bibitem{kent2011quantum}
A.~Kent, W.~J. Munro, and T.~P. Spiller, {\it Quantum tagging: {A}uthenticating
  location via quantum information and relativistic signaling constraints},
  {\em Physical Review A} {\bf 84} (2011), no.~1 012326,
  [\href{http://arxiv.org/abs/1008.2147}{{\tt arXiv:1008.2147}}].

\bibitem{kent2006tagging}
A.~P. Kent, W.~J. Munro, T.~P. Spiller, and R.~G. Beausoleil, {\it Tagging
  systems},  July~11, 2006.
\newblock US Patent 7,075,438.

\bibitem{malaney2010location}
R.~A. Malaney, {\it Location-dependent communications using quantum
  entanglement},  {\em Physical Review A} {\bf 81} (2010), no.~4 042319,
  [\href{http://arxiv.org/abs/1003.0949}{{\tt arXiv:1003.0949}}].

\bibitem{leung2010quantum}
D.~Leung, J.~Oppenheim, and A.~Winter, {\it Quantum network communication—the
  butterfly and beyond},  {\em IEEE Transactions on Information Theory} {\bf
  56} (2010), no.~7 3478--3490,
  [\href{http://arxiv.org/abs/quant-ph/0608223}{{\tt quant-ph/0608223}}].

\bibitem{chaves2021causal}
R.~Chaves, G.~Moreno, E.~Polino, D.~Poderini, I.~Agresti, A.~Suprano, M.~R.
  Barros, G.~Carvacho, E.~Wolfe, A.~Canabarro, et~al., {\it Causal networks and
  freedom of choice in {B}ell’s theorem},  {\em PRX Quantum} {\bf 2} (2021),
  no.~4 040323, [\href{http://arxiv.org/abs/2105.05721}{{\tt
  arXiv:2105.05721}}].

\bibitem{ryu2006holographic}
S.~Ryu and T.~Takayanagi, {\it Holographic derivation of entanglement entropy
  from the anti--de sitter space/conformal field theory correspondence},  {\em
  Physical review letters} {\bf 96} (2006), no.~18 181602,
  [\href{http://arxiv.org/abs/hep-th/0603001v1}{{\tt hep-th/0603001v1}}].

\bibitem{ryu2006aspects}
S.~Ryu and T.~Takayanagi, {\it Aspects of holographic entanglement entropy},
  {\em Journal of High Energy Physics} {\bf 2006} (2006), no.~08 045,
  [\href{http://arxiv.org/abs/hep-th/0605073}{{\tt hep-th/0605073}}].

\bibitem{headrick2007holographic}
M.~Headrick and T.~Takayanagi, {\it Holographic proof of the strong
  subadditivity of entanglement entropy},  {\em Physical Review D} {\bf 76}
  (2007), no.~10 106013, [\href{http://arxiv.org/abs/0704.3719}{{\tt
  arXiv:0704.3719}}].

\bibitem{hubeny2007covariant}
V.~E. Hubeny, M.~Rangamani, and T.~Takayanagi, {\it A covariant holographic
  entanglement entropy proposal},  {\em Journal of High Energy Physics} {\bf
  2007} (2007), no.~07 062, [\href{http://arxiv.org/abs/0705.0016}{{\tt
  arXiv:0705.0016}}].

\bibitem{lewkowycz2013generalized}
A.~Lewkowycz and J.~Maldacena, {\it Generalized gravitational entropy},  {\em
  Journal of High Energy Physics} {\bf 2013} (2013) 90,
  [\href{http://arxiv.org/abs/1304.4926}{{\tt arXiv:1304.4926}}].

\bibitem{faulkner2013quantum}
T.~Faulkner, A.~Lewkowycz, and J.~Maldacena, {\it Quantum corrections to
  holographic entanglement entropy},  {\em Journal of High Energy Physics} {\bf
  2013} (2013), no.~74 [\href{http://arxiv.org/abs/1307.2892}{{\tt
  arXiv:1307.2892}}].

\bibitem{dong2016deriving}
X.~Dong, A.~Lewkowycz, and M.~Rangamani, {\it Deriving covariant holographic
  entanglement},  {\em Journal of High Energy Physics} {\bf 2016} (2016) 28,
  [\href{http://arxiv.org/abs/1607.07506}{{\tt arXiv:1607.07506}}].

\bibitem{engelhardt2015quantum}
N.~Engelhardt and A.~C. Wall, {\it Quantum extremal surfaces: Holographic
  entanglement entropy beyond the classical regime},  {\em Journal of High
  Energy Physics} {\bf 2015} (2015), no.~73
  [\href{http://arxiv.org/abs/1408.3203}{{\tt arXiv:1408.3203}}].

\bibitem{dong2018entropy}
X.~Dong and A.~Lewkowycz, {\it Entropy, extremality, euclidean variations, and
  the equations of motion},  {\em Journal of High Energy Physics} {\bf 2018}
  (2018), no.~1 1--33, [\href{http://arxiv.org/abs/1705.08453}{{\tt
  arXiv:1705.08453}}].

\bibitem{almheiri2015bulk}
A.~Almheiri, X.~Dong, and D.~Harlow, {\it Bulk locality and quantum error
  correction in {A}d{S}/{CFT}},  {\em Journal of High Energy Physics} {\bf
  2015} (2015), no.~4 1--34, [\href{http://arxiv.org/abs/1411.7041}{{\tt
  arXiv:1411.7041}}].

\bibitem{may2021bulk}
A.~May, {\it {Bulk private curves require large conditional mutual
  information}},  {\em JHEP} {\bf 09} (2021) 042,
  [\href{http://arxiv.org/abs/2105.08094}{{\tt arXiv:2105.08094}}].

\bibitem{may2021quantum}
A.~May and D.~Wakeham, {\it {Quantum tasks require islands on the brane}},
  {\em Class. Quant. Grav.} {\bf 38} (2021), no.~14 144001,
  [\href{http://arxiv.org/abs/2102.01810}{{\tt arXiv:2102.01810}}].

\bibitem{dolev2022holography}
K.~Dolev and S.~Cree, {\it Holography as a resource},  {\em to appear}.

\bibitem{may2022complexity}
A.~May, {\it {Complexity and entanglement in non-local computation and
  holography}},  \href{http://arxiv.org/abs/2204.00908}{{\tt
  arXiv:2204.00908}}.

\bibitem{maximin}
A.~C. Wall, {\it Maximin surfaces, and the strong subadditivity of the
  covariant holographic entanglement entropy},  {\em Classical and Quantum
  Gravity} {\bf 31} (2014) 225007, [\href{http://arxiv.org/abs/1211.3494}{{\tt
  arXiv:1211.3494}}].

\bibitem{maximin2}
D.~Marolf, A.~C. Wall, and Z.~Wang, {\it Restricted maximin surfaces and {HRT}
  in generic black hole spacetimes},  {\em Journal of High Energy Physics} {\bf
  2019} (2019) 127, [\href{http://arxiv.org/abs/1901.03879}{{\tt
  arXiv:1901.03879}}].

\bibitem{CKNR}
B.~Czech, J.~L. Karczmarek, F.~Nogueira, and M.~V. Raamsdonk, {\it The gravity
  dual of a density matrix},  {\em Classical and Quantum Gravity} {\bf 29}
  (2012), no.~15 [\href{http://arxiv.org/abs/1204.1330}{{\tt
  arXiv:1204.1330}}].

\bibitem{jafferis2016relative}
D.~L. Jafferis, A.~Lewkowycz, J.~Maldacena, and S.~J. Suh, {\it Relative
  entropy equals bulk relative entropy},  {\em Journal of High Energy Physics}
  {\bf 2016} (2016), no.~6 1--20, [\href{http://arxiv.org/abs/1512.06431}{{\tt
  arXiv:1512.06431}}].

\bibitem{dong2016reconstruction}
X.~Dong, D.~Harlow, and A.~C. Wall, {\it Reconstruction of bulk operators
  within the entanglement wedge in gauge-gravity duality},  {\em Physical
  review letters} {\bf 117} (2016), no.~2 021601,
  [\href{http://arxiv.org/abs/1601.05416}{{\tt arXiv:1601.05416}}].

\bibitem{cotler2019entanglement}
J.~Cotler, P.~Hayden, G.~Penington, G.~Salton, B.~Swingle, and M.~Walter, {\it
  Entanglement wedge reconstruction via universal recovery channels},  {\em
  Physical Review X} {\bf 9} (2019), no.~3 031011,
  [\href{http://arxiv.org/abs/1704.05839}{{\tt arXiv:1704.05839}}].

\bibitem{headrick2014causality}
M.~Headrick, V.~E. Hubeny, A.~Lawrence, and M.~Rangamani, {\it Causality \&
  holographic entanglement entropy},  {\em Journal of High Energy Physics} {\bf
  2014} (2014), no.~12 [\href{http://arxiv.org/abs/1408.6300}{{\tt
  arXiv:1408.6300}}].

\bibitem{quantum-maximin}
C.~Akers, N.~Engelhardt, G.~Penington, and M.~Usatyuk, {\it Quantum maximin
  surfaces},  \href{http://arxiv.org/abs/1912.02799}{{\tt arXiv:1912.02799}}.

\bibitem{QFC}
R.~Bousso, Z.~Fisher, S.~Leichenauer, and A.~C. Wall, {\it A quantum focussing
  conjecture},  {\em Physical Review D} {\bf 93} (2016) 064044,
  [\href{http://arxiv.org/abs/1506.02669}{{\tt arXiv:1506.02669}}].

\bibitem{beigi2011simplified}
S.~Beigi and R.~K{\"o}nig, {\it Simplified instantaneous non-local quantum
  computation with applications to position-based cryptography},  {\em New
  Journal of Physics} {\bf 13} (2011), no.~9 093036,
  [\href{http://arxiv.org/abs/1101.1065}{{\tt arXiv:1101.1065}}].

\bibitem{raychaudhuri}
A.~Raychaudhuri, {\it Relativistic cosmology. {I}},  {\em Physical Review} {\bf
  98} (1955) 1123.

\bibitem{penrose1972techniques}
R.~Penrose, {\em Techniques in differential topology in relativity}.
\newblock SIAM, 1972.

\bibitem{cooney2015rank}
T.~Cooney, M.~Junge, C.~Palazuelos, and D.~P{\'e}rez-Garc{\'\i}a, {\it Rank-one
  quantum games},  {\em computational complexity} {\bf 24} (2015), no.~1
  133--196, [\href{http://arxiv.org/abs/1112.3563}{{\tt arXiv:1112.3563}}].

\bibitem{junge2022geometry}
M.~Junge, A.~M. Kubicki, C.~Palazuelos, and D.~P{\'e}rez-Garc{\'\i}a, {\it
  Geometry of {B}anach spaces: a new route towards position based
  cryptography},  {\em Communications in Mathematical Physics} (2022) 1--54,
  [\href{http://arxiv.org/abs/2103.16357}{{\tt arXiv:2103.16357}}].

\bibitem{buhrman2014position}
H.~Buhrman, N.~Chandran, S.~Fehr, R.~Gelles, V.~Goyal, R.~Ostrovsky, and
  C.~Schaffner, {\it Position-based quantum cryptography: {I}mpossibility and
  constructions},  {\em SIAM Journal on Computing} {\bf 43} (2014), no.~1
  150--178, [\href{http://arxiv.org/abs/1009.2490}{{\tt arXiv:1009.2490}}].

\bibitem{lau2011insecurity}
H.-K. Lau and H.-K. Lo, {\it Insecurity of position-based quantum-cryptography
  protocols against entanglement attacks},  {\em Physical review a} {\bf 83}
  (2011), no.~1 012322, [\href{http://arxiv.org/abs/1009.2256}{{\tt
  arXiv:1009.2256}}].

\bibitem{hoeffding1994probability}
W.~Hoeffding, {\it Probability inequalities for sums of bounded random
  variables},  in {\em The collected works of Wassily Hoeffding}, pp.~409--426.
\newblock Springer, 1994.

\bibitem{tomamichel2013monogamy}
M.~Tomamichel, S.~Fehr, J.~Kaniewski, and S.~Wehner, {\it A
  monogamy-of-entanglement game with applications to device-independent quantum
  cryptography},  {\em New Journal of Physics} {\bf 15} (2013), no.~10 103002,
  [\href{http://arxiv.org/abs/1210.4359}{{\tt arXiv:1210.4359}}].

\bibitem{harlow2017ryu}
D.~Harlow, {\it The ryu--takayanagi formula from quantum error correction},
  {\em Communications in Mathematical Physics} {\bf 354} (2017), no.~3
  865--912, [\href{http://arxiv.org/abs/1607.03901}{{\tt arXiv:1607.03901}}].

\bibitem{akers2019large}
C.~Akers, A.~Levine, and S.~Leichenauer, {\it Large breakdowns of entanglement
  wedge reconstruction},  {\em Physical Review D} {\bf 100} (2019), no.~12
  126006, [\href{http://arxiv.org/abs/1908.03975}{{\tt arXiv:1908.03975}}].

\bibitem{hayden2019learning}
P.~Hayden and G.~Penington, {\it Learning the alpha-bits of black holes},  {\em
  Journal of High Energy Physics} {\bf 2019} (2019), no.~12 1--55,
  [\href{http://arxiv.org/abs/1807.06041}{{\tt arXiv:1807.06041}}].

\bibitem{akers2021leading}
C.~Akers and G.~Penington, {\it Leading order corrections to the quantum
  extremal surface prescription},  {\em Journal of High Energy Physics} {\bf
  2021} (2021), no.~4 1--73, [\href{http://arxiv.org/abs/2008.03319}{{\tt
  arXiv:2008.03319}}].

\bibitem{akers2022quantum}
C.~Akers and G.~Penington, {\it Quantum minimal surfaces from quantum error
  correction},  {\em SciPost Physics} {\bf 12} (2022), no.~5 157,
  [\href{http://arxiv.org/abs/2109.14618}{{\tt arXiv:2109.14618}}].

\bibitem{bao2019beyond}
N.~Bao, G.~Penington, J.~Sorce, and A.~C. Wall, {\it Beyond toy models:
  distilling tensor networks in full ads/cft},  {\em Journal of High Energy
  Physics} {\bf 2019} (2019), no.~11 1--63,
  [\href{http://arxiv.org/abs/1812.01171}{{\tt arXiv:1812.01171}}].

\bibitem{hayden2007black}
P.~Hayden and J.~Preskill, {\it Black holes as mirrors: quantum information in
  random subsystems},  {\em Journal of high energy physics} {\bf 2007} (2007),
  no.~09 120, [\href{http://arxiv.org/abs/0708.4025}{{\tt arXiv:0708.4025}}].

\bibitem{Steane}
A.~Steane, {\it {Multiple particle interference and quantum error correction}},
   {\em Proc. Roy. Soc. Lond. A} {\bf 452} (1996) 2551,
  [\href{http://arxiv.org/abs/quant-ph/9601029}{{\tt quant-ph/9601029}}].

\bibitem{cree2022code}
S.~Cree and A.~May, {\it Code-routing: a new attack on position-verification},
  {\em arXiv preprint arXiv:2202.07812} (2022)
  [\href{http://arxiv.org/abs/2202.07812}{{\tt arXiv:2202.07812}}].

\bibitem{sorce2019cutoffs}
J.~Sorce, {\it Holographic entanglement entropy is cutoff-covariant},  {\em
  Journal of High Energy Physics} {\bf 2019} (2019), no.~10 1--42,
  [\href{http://arxiv.org/abs/1908.02297}{{\tt arXiv:1908.02297}}].

\bibitem{waldbook}
R.~M. Wald, {\em General {R}elativity}.
\newblock University of Chicago Press, 1984.

\bibitem{bennett1993teleporting}
C.~H. Bennett, G.~Brassard, C.~Cr{\'e}peau, R.~Jozsa, A.~Peres, and W.~K.
  Wootters, {\it Teleporting an unknown quantum state via dual classical and
  {E}instein-{P}odolsky-{R}osen channels},  {\em Physical review letters} {\bf
  70} (1993), no.~13 1895.

\bibitem{ishizaka2008asymptotic}
S.~Ishizaka and T.~Hiroshima, {\it Asymptotic teleportation scheme as a
  universal programmable quantum processor},  {\em Physical Review Letters}
  {\bf 101} (2008), no.~24 240501, [\href{http://arxiv.org/abs/0807.4568}{{\tt
  arXiv:0807.4568}}].

\bibitem{helstrom1969quantum}
C.~W. Helstrom, {\it Quantum detection and estimation theory},  {\em Journal of
  Statistical Physics} {\bf 1} (1969), no.~2 231--252.

\bibitem{wilde2013quantum}
M.~M. Wilde, {\em Quantum information theory}.
\newblock Cambridge University Press, 2013.

\bibitem{fuchs1999cryptographic}
C.~A. Fuchs and J.~Van De~Graaf, {\it Cryptographic distinguishability measures
  for quantum-mechanical states},  {\em IEEE Transactions on Information
  Theory} {\bf 45} (1999), no.~4 1216--1227,
  [\href{http://arxiv.org/abs/quant-ph/9712042}{{\tt quant-ph/9712042}}].

\bibitem{muller2013quantum}
M.~M{\"u}ller-Lennert, F.~Dupuis, O.~Szehr, S.~Fehr, and M.~Tomamichel, {\it On
  quantum {R}{\'e}nyi entropies: A new generalization and some properties},
  {\em Journal of Mathematical Physics} {\bf 54} (2013), no.~12 122203,
  [\href{http://arxiv.org/abs/1306.3142}{{\tt arXiv:1306.3142}}].

\bibitem{wilde2014strong}
M.~M. Wilde, A.~Winter, and D.~Yang, {\it Strong converse for the classical
  capacity of entanglement-breaking and {H}adamard channels via a sandwiched
  {R}{\'e}nyi relative entropy},  {\em Communications in Mathematical Physics}
  {\bf 331} (2014), no.~2 593--622, [\href{http://arxiv.org/abs/1306.1586}{{\tt
  arXiv:1306.1586}}].

\bibitem{PhysRevLett.83.648}
R.~Cleve, D.~Gottesman, and H.-K. Lo, {\it How to share a quantum secret},
  {\em Phys. Rev. Lett.} {\bf 83} (Jul, 1999) 648--651,
  [\href{http://arxiv.org/abs/quant-ph/9901025}{{\tt quant-ph/9901025}}].

\bibitem{eastin2009restrictions}
B.~Eastin and E.~Knill, {\it Restrictions on transversal encoded quantum gate
  sets},  {\em Physical review letters} {\bf 102} (2009), no.~11 110502,
  [\href{http://arxiv.org/abs/0811.4262}{{\tt arXiv:0811.4262}}].

\bibitem{zeng2011transversality}
B.~Zeng, A.~Cross, and I.~L. Chuang, {\it Transversality versus universality
  for additive quantum codes},  {\em IEEE Transactions on Information Theory}
  {\bf 57} (2011), no.~9 6272--6284,
  [\href{http://arxiv.org/abs/0706.1382}{{\tt arXiv:0706.1382}}].

\bibitem{chen2008subsystem}
X.~Chen, H.~Chung, A.~W. Cross, B.~Zeng, and I.~L. Chuang, {\it Subsystem
  stabilizer codes cannot have a universal set of transversal gates for even
  one encoded qudit},  {\em Physical Review A} {\bf 78} (2008), no.~1 012353,
  [\href{http://arxiv.org/abs/0801.23602}{{\tt arXiv:0801.23602}}].

\bibitem{PhysRevA.91.012305}
F.~Pastawski and B.~Yoshida, {\it Fault-tolerant logical gates in quantum
  error-correcting codes},  {\em Phys. Rev. A} {\bf 91} (Jan, 2015) 012305,
  [\href{http://arxiv.org/abs/1408.1720}{{\tt arXiv:1408.1720}}].

\bibitem{faist2020continuous}
P.~Faist, S.~Nezami, V.~V. Albert, G.~Salton, F.~Pastawski, P.~Hayden, and
  J.~Preskill, {\it Continuous symmetries and approximate quantum error
  correction},  {\em Physical Review X} {\bf 10} (2020), no.~4 041018,
  [\href{http://arxiv.org/abs/1902.07714}{{\tt arXiv:1902.07714}}].

\bibitem{kubica2021using}
A.~Kubica and R.~Demkowicz-Dobrza{\'n}ski, {\it Using quantum metrological
  bounds in quantum error correction: A simple proof of the approximate
  eastin-knill theorem},  {\em Physical Review Letters} {\bf 126} (2021),
  no.~15 150503.

\bibitem{tajima2021universal}
H.~Tajima and K.~Saito, {\it Universal limitation of quantum information
  recovery: symmetry versus coherence},  {\em arXiv preprint arXiv:2103.01876}
  (2021).

\bibitem{zhou2021new}
S.~Zhou, Z.-W. Liu, and L.~Jiang, {\it New perspectives on covariant quantum
  error correction},  {\em Quantum} {\bf 5} (2021) 521,
  [\href{http://arxiv.org/abs/2005.11918}{{\tt arXiv:2005.11918}}].

\bibitem{liu2021approximate}
Z.-W. Liu and S.~Zhou, {\it Approximate symmetries and quantum error
  correction},  {\em arXiv preprint arXiv:2111.06355} (2021)
  [\href{http://arxiv.org/abs/2111.06355}{{\tt arXiv:2111.06355}}].

\bibitem{liu2021quantum}
Z.-W. Liu and S.~Zhou, {\it Quantum error correction meets continuous
  symmetries: fundamental trade-offs and case studies},  {\em arXiv preprint
  arXiv:2111.06360} (2021) [\href{http://arxiv.org/abs/2111.06360}{{\tt
  arXiv:2111.06360}}].

\end{thebibliography}\endgroup

\end{document}